\documentclass[prx,letterpaper,aps,superscriptaddress,floatfix,twocolumn,nofootinbib]{revtex4-2}
\pdfoutput=1

\usepackage{amsmath}
\usepackage{amssymb}
\usepackage{amsthm}
\usepackage{makecell}
\usepackage{amsfonts}
\usepackage{calc}
\usepackage[nomessages]{fp}
\usepackage{physics}
\usepackage{dsfont}
\usepackage{xcolor}
\usepackage{listings}
\usepackage{comment}
\usepackage{tabularx}
\usepackage{enumerate}
\usepackage{latexsym}
\usepackage{inputenc}
\usepackage{tikz}
\usetikzlibrary{patterns}
\usepackage{mathdots}
\usepackage{multirow}
\usepackage{thmtools}
\usepackage{thm-restate}
\usepackage{mathtools}

\let\norm\undefined
\DeclarePairedDelimiter\norm{\lVert}{\rVert}
\usepackage{psfrag}
\usepackage{color}
\usepackage{bm}
\usepackage{graphicx}
\usepackage{subfigure}
\usepackage{graphicx}
\usepackage{amsmath}
\usepackage{amsfonts}
\usepackage{amssymb}
\usepackage{ulem}
\usepackage{xcolor}
\usepackage{slashed}
\usepackage{soul}
\usepackage{comment}
\usepackage{tikz-cd}
\usepackage[colorlinks=true, linkcolor=blue, citecolor=magenta]{hyperref}

\newcommand{\beq}{\begin{equation}}
\newcommand{\eeq}{\end{equation}}

\newcommand{\beqa}{\begin{eqnarray}}
\newcommand{\eeqa}{\end{eqnarray}}

\newcommand{\ea}{\end{array}}
 
\def\eea{\end{eqnarray}}

\def\<{\langle}
\def\>{\rangle}

\usepackage{amsmath}
\usepackage{comment}
\usepackage{amssymb}
\usepackage{amsthm}
\newtheorem{thm}{Theorem}

\newtheorem{lemma}{Lemma}
\newtheorem{prop}{Proposition}

\newtheorem{cor}{Corollary}
\theoremstyle{definition}
\newtheorem{defin}{Definition}

\usepackage{color}
\usepackage{tikz-cd}
\usepackage{comment}

\usepackage{tikz}
\usepackage[export]{adjustbox}
\def\[#1\]{%
  \begin{equation}\begin{gathered}#1\end{gathered}\end{equation}%
}

\usepackage[safe]{tipa}

\newcommand{\eneq}{\end{equation}}

\def\ea{{\it et al.}}
\input{epsf}

\newcommand{\tot}{\text{tot}}

\newcommand{\loc}{\text{loc}}

\newcommand{\nn}{\nonumber}

\newcommand{\mZ}{\mathcal{Z}}

\newcommand{\mP}{\mathcal{P}}
\newcommand{\mA}{\mathcal{A}}

\newcommand{\mC}{\mathcal{C}}

\newcommand{\mV}{\mathcal{V}}

\newcommand{\defn}{:=}

\newcommand{\btp}{\begin{tikzpicture}}
\newcommand{\etp}{\end{tikzpicture}}

\newcolumntype{L}[1]{>{\raggedright\arraybackslash}p{#1}}
\newcolumntype{C}[1]{>{\centering\arraybackslash}p{#1}}
\newcolumntype{R}[1]{>{\raggedleft\arraybackslash}p{#1}}
\makeatletter
\newsavebox{\@brx}
\newcommand{\llangle}[1][]{\savebox{\@brx}{\(\m@th{#1\langle}\)}%
  \mathopen{\copy\@brx\kern-0.5\wd\@brx\usebox{\@brx}}}
\newcommand{\rrangle}[1][]{\savebox{\@brx}{\(\m@th{#1\rangle}\)}%
  \mathclose{\copy\@brx\kern-0.5\wd\@brx\usebox{\@brx}}}
\makeatother
\newcommand{\lgen}{\llangle}
\newcommand{\rgen}{\rrangle}

\usepackage{bbm}
\newcommand{\ktO}{\ket{\overline{0}}}
\newcommand{\ktW}{\ket{W}}
\newcommand{\HImHopLambda}{H_{\text{ImHop},\, \Lambda}}
\newcommand{\hNtotLambda}{\hat{N}_{\text{tot},\, \Lambda}}
\makeatletter
\def\l@subsubsection#1#2{}
\makeatother
\makeatletter
\renewcommand\onecolumngrid{%
  \do@columngrid{one}{\@ne}%
  \def\set@footnotewidth{\onecolumngrid}%
  \def\footnoterule{\kern-6pt\hrule width 1.5in\kern6pt}%
}
\renewcommand\twocolumngrid{%
  \def\footnoterule{%
    \dimen@\skip\footins\divide\dimen@\thr@@
    \kern-\dimen@\hrule width .5in\kern\dimen@}%
  \do@columngrid{mlt}{\tw@}%
}
\makeatother

\begin{document}
\title{Distinct Types of Parent Hamiltonians for Quantum States:\\ Insights from the $W$ State as a Quantum Many-Body Scar}
\author{Lei Gioia}
\email{lei.gioia.y@gmail.com}
\affiliation{Walter Burke Institute for Theoretical Physics, Caltech, Pasadena, CA, USA}
\affiliation{Department of Physics, Caltech, Pasadena, CA, USA}
\author{Sanjay Moudgalya}
\email{sanjay.moudgalya@gmail.com}
\affiliation{School of Natural Sciences, Technische Universit\"{a}t M\"{u}nchen (TUM), James-Franck-Str. 1, 85748 Garching, Germany}
\affiliation{Munich Center for Quantum Science and Technology (MCQST), Schellingstr. 4, 80799 M\"{u}nchen, Germany}
\author{Olexei I.  Motrunich}
\email{motrunch@caltech.edu}
\affiliation{Walter Burke Institute for Theoretical Physics, Caltech, Pasadena, CA, USA}
\affiliation{Department of Physics, Caltech, Pasadena, CA, USA}
\begin{abstract}
The construction of parent Hamiltonians that possess a given state as their ground state is a well-studied problem.
In this work, we generalize this notion by considering simple quantum states and examining local Hamiltonians that have these states as exact eigenstates.
These states often correspond to Quantum Many-Body Scars (QMBS) of their respective parent Hamiltonians.
Motivated by earlier works on Hamiltonians with QMBS, in this work we formalize the differences between \textit{three} distinct types of parent Hamiltonians, which differ in their decompositions into strictly local terms with the same eigenstates.
We illustrate this classification using the $W$ state as the primary example, for which we rigorously derive the complete set of local parent Hamiltonians, which also allows us to establish general results such as the existence of asymptotic QMBS, and distinct dynamical signatures associated with the different parent Hamiltonian types.
Finally, we derive more general results on the parent Hamiltonian types that allow us to obtain some immediate results for simple quantum states such as product states, where only a single type exists, and short-range-entangled states, for which we identify constraints on the admissible types.
Altogether, our work opens the door to classifying the rich structures and dynamical properties of parent Hamiltonians that arise from the interplay between locality and QMBS.
\end{abstract}

\maketitle
\tableofcontents
%
%
\section{Introduction}
The construction of Hamiltonians that possess certain special quantum states as ground states has been a central question of interest in quantum many-body physics. 
While such Hamiltonians, referred to as parent Hamiltonians, are toy models that might differ from naturally occurring microscopic Hamiltonians, they nevertheless are helpful in understanding general conditions that might be necessary to realize particular kinds of ground states, e.g., they have led to rigorous results on the existence of gapped phases of matter~\cite{AffleckKennedyLiebTasaki1988, FannesNachtergaeleWerner1992,  Perez-Garcia2007, PerezGarciaVerstraeteCiracWolf2008,PhysRevB.72.045141}.
Moreover, a better understanding of parent Hamiltonians can also lead to new ways of realizing or approximating Hamiltonians using other realizable interactions, eventually leading to engineering of particular ground states~\cite{chertkovclark2018}.  
Much of the traditional literature on parent Hamiltonians focused on constructing local Hamiltonians that realize particular states as \textit{ground states}, where parent Hamiltonians that realize various exotic states such as fractional quantum Hall states~\cite{haldane1983fractional, bernevig2008model, kapit2010exact, lee2015geometric} and spin liquids~\cite{thomale2009parent, greiter2014parent, zhou2017qsl} have been proposed.
Nevertheless, the exploration of parent Hamiltonians in the context of many other kinds of states as ground states, such as the $W$ state, has only been recently initiated~\cite{gioia2024wstateuniqueground}.
In addition, the discovery of Quantum Many-Body Scars (QMBS)~\cite{Serbyn_2021, Moudgalya_2022, Papic_2022, Chandran_2023, Bernien_2017} has ignited interest in Hamiltonians that realize interesting quantum states not just as ground states, but as arbitrary eigenstates in the spectrum.
In particular, it was quickly noticed that these kinds of generalized parent Hamiltonians, obtained when states from a given set are allowed to be arbitrary eigenstates, are richer than the usual types of parent Hamiltonians studied when the states are restricted to be ground states.
For example, in the context of Matrix-Product States (MPS), the constructed parent Hamiltonians are frustration-free with respect to their ground states~\cite{Perez-Garcia2007,fernándezgonzalez2015}, i.e., the state minimizes the energy of each individual local Hamiltonian term.
However, as is known from the QMBS literature, this is often not the case for the generalized parent Hamiltonians, i.e., when the states can be excited states, where there are parent Hamiltonians for which the states are not eigenstates of the individual terms.
This motivates us to consider the general question of the classification of parent Hamiltonians for a set of quantum states, where the states are allowed to be arbitrary eigenstates rather than just ground states.
While there have been works that introduce numerical methods to find such generalized parent Hamiltonians for a given state~\cite{qiranard2017, chertkovclark2018, greiter2018parent}, the structures in the obtained Hamiltonians have not been studied systematically, which we initiate in this work.
Some partial classifications have been proposed previously in the QMBS literature, which we now review.
QMBS are non-thermal eigenstates embedded in the middle of the spectra of otherwise generic non-integrable Hamiltonians, which can lead to interesting non-thermal dynamics of certain special initial states under time-evolution with such Hamiltonians.
Numerous examples are known by now~\cite{Serbyn_2021, Moudgalya_2022, Papic_2022, Chandran_2023}, and there have also been many different systematic constructions of parent Hamiltonians that exhibit a given QMBS set as eigenstates~\cite{Shiraishi_2017, mark2020unified, moudgalya2020eta, Mark2020Eta, odea2020from, ren2020quasisymmetry, pakrouski2020many, ren2021deformed, pakrouski2021group, rozon2023broken}. 
Motivated from the understanding that QMBS can be viewed as a consequence of an unconventional symmetry, Ref.~\cite{moudgalya2023exhaustive} further unified these systematic constructions, and showed that the exhaustive set of parent Hamiltonians that possesses a particular set of QMBS forms a certain associative algebra of operators on the Hilbert space, for which local generators can be written down.
This algebraic understanding of the space of parent Hamiltonians, which we review in Sec.~\ref{subsec:qmbsconn}, led to a precise two-fold classification---type I and type II---of qualitatively different kinds of local Hamiltonians that exhibit certain QMBS.
Roughly speaking, type I Hamiltonians are those that are easily understood in terms of the generators of the said algebra and are closely related to frustration-free Hamiltonians in the context of ground states, whereas type II were those that have more complicated relations to these generators. 
However, coming from the algebraic framework brings its own baggage, and the classification proposed there had many technical constraints, e.g., for it to be mathematically precise, in many examples one needs to consider QMBS to be \textit{degenerate} eigenstates rather than arbitrary eigenstates.
Although these technicalities can be circumvented in specific cases~\cite{moudgalya2023exhaustive}, it nevertheless prohibits the wide application of this classification. 
In the meantime, there have been recent works~\cite{Omiya_2023fractionalization, wang2024embedding} that have noticed interesting simple structures in type II Hamiltonians, which are complicated to understand in the algebraic framework. 
Hence, in this work, we take a step back and study the exhaustive set of local parent Hamiltonians for a simple but non-trivial quantum state, the $W$ state.
The many-body $W$ state on an $N$ qubit system is given by
\begin{align}
    \ket{W}=\frac{1}{\sqrt{N}}\sum_{j=1}^N s_j^\dag\ket{\bar{0}}\quad,
\end{align}
where $\ket{\bar{0}}:=\ket{0}^{\otimes N}$ and $s_j^\dag\ket{0}_j=\ket{1}_j$ creates a particle on the $j$th site, and describes an equal superposition of a single-particle state at each point in space.
It is arguably the simplest state that exhibits multipartite entanglement and immunity to qubit-loss~\cite{PhysRevA.98.062335}, while maintaining a low bipartite entanglement entropy as well as a finite-bond dimension MPS representation~\cite{RevModPhys.93.045003} (albeit not a translationally invariant MPS).
This gives the state analytical tractability, allowing for the study and classification of its intrinsic quantum entanglement properties~\cite{PhysRevA.62.062314}.
The $W$ state appears in a variety of theoretical quantum many-body contexts such as in bosonic and fermionic lattice systems~\cite{PhysRevA.75.052109,gioia2024wstateuniqueground}, and quantum information theory~\cite{kieferova2024logdepth}.
Alongside its theoretical appeal and simplicity, it also appears ubiquitously in a myriad of quantum platforms, such as recent experiments in photonic systems~\cite{eibl2004experimental}, trapped-ion systems~\cite{haffner2005scalable}, Rydberg atom arrays~\cite{catalano2025experimentalpreparationwstates}, and superconducting circuits~\cite{https://doi.org/10.1002/qute.201900015}.
In this paper, we derive the most general form of a finite-range Hamiltonian for which the $W$ state is an eigenstate, which we then use to illuminate some general characteristics of QMBS systems.
Its form motivates a natural classification scheme that extends the two types in \cite{moudgalya2023exhaustive} into three types, while preserving the meaning in the cases where the algebraic understanding overlaps. 
This is motivated more naturally without any reference to underlying algebras, and is applicable to arbitrary sets of quantum states without any extraneous constraints such as the requirement of degeneracies.
This paper is organized as follows.
In Sec.~\ref{sec:QMBSgeneral}, we introduce the basic technical concepts of parent Hamiltonians, the type classification we propose, and connections to previous classifications in the QMBS literature.
Then, in Sec.~\ref{sec:WstateQMBS}, we demonstrate the full class of parent Hamiltonians for the $W$ state, and consequences that follow due to that structure, such as the existence of other eigenstates, and asymptotic QMBS.
In Sec.~\ref{sec:dynamical}, we demonstrate some dynamical signatures that differentiate between the type I and type II Hamiltonians.
Motivated by these results, in Sec.~\ref{sec:generalQMBS}, we present some general developments that characterize the different types of parent Hamiltonians, which we illustrate using examples from the $W$ state.
In Sec.~\ref{sec:SREQMBS}, we use these to derive constraints on the parent Hamiltonians for other simple eigenstates such as product states and MPS.
We conclude in Sec.~\ref{sec:discussion} with a summary and a discussion of open problems.
We relegate many technical details of proofs and other computations in the main text to the appendices.
%

%
\section{Types of Parent Hamiltonians}
\label{sec:QMBSgeneral}
We now introduce the fundamental problem that we are interested in this work. 
Given a set of states $\{\ket{\psi_{n}}\}$, we are interested in understanding the structure of Hermitian operators $H$ that have these states as eigenstates, i.e.,
\begin{equation}
    H\ket{\psi_{n}} = E_n \ket{\psi_{n}},
\label{eq:Heigenstates}
\end{equation}
where we do not impose any constraint on the eigenvalues $E_n$. 
We further restrict the operator $H$ to be \textit{extensive local}, which means that it is of the form
\begin{equation}
    H = \sum_{j}{h_{[j]}}, 
\label{eq:localHam}
\end{equation}
where $h_{[j]}$ is a \textit{strictly local}\footnote{In this paper we will use the terms \textit{strictly local} and \textit{finite-range} interchangeably since we are dealing with geometrically-local systems.} 
Hermitian operator with a bounded norm and support of at most $R$ that acts in the vicinity of a site $j$ on some underlying lattice (and is not necessarily the same for each $j$).
Whenever the Hamiltonian $H$ satisfies the conditions of Eqs.~(\ref{eq:Heigenstates}) and (\ref{eq:localHam}), we refer to it as a \textit{Parent Hamiltonian} for the states $\{\ket{\psi_n}\}$.
Note that unlike many studies of parent Hamiltonians~\cite{FannesNachtergaeleWerner1992,Perez-Garcia2007,RevModPhys.93.045003,PerezGarciaVerstraeteCiracWolf2008}, we are \textit{not} imposing the condition that the states $\{\ket{\psi_{n}}\}$ are the ground states of the Hamiltonian.
These states could lie anywhere in the spectrum of $H$, and hence these states are typically examples of \textit{Quantum Many-Body Scars} (QMBS)~\cite{Bernien_2017, Turner_2018weak, Turner_2018quantum, Khemani_2019, Moudgalya_2018exact, Moudgalya_2018entanglement, Lin_2019, Surace_2020, Iadecola_2019, Ivanov_2025exact, Serbyn_2021, Moudgalya_2022, Papic_2022, Chandran_2023} of the Hamiltonian $H$.\footnote{Note that traditionally in the literature there are additional checks one needs to perform on $H$ and $\{\ket{\psi_{n}}\}$ for them to qualify as QMBS Hamiltonians and states respectively, such as the non-integrability of the Hamiltonian and the non-thermality of the states~\cite{Serbyn_2021, Moudgalya_2022, Papic_2022, Chandran_2023}.
However, we omit such tests, since we are working with evidently non-thermal states with low entanglement, and we are interested in generic parent Hamiltonians rather than a specific one, and a typical Hamiltonian will be non-integrable.}
\subsection{Definitions of Types}\label{subsec:classes}
Given an extensive local operator of the form of Eq.~(\ref{eq:localHam}), we find that they can be classified according to their locality properties. 
In particular we find that there are three distinct classes of operators that satisfy the conditions of Eq.~(\ref{eq:Heigenstates}):
\begin{enumerate}
    \item Type I: $H$ can be written as a linear combination of Hermitian strictly local terms that individually have the $\{\ket{\psi_{n}}\}$ as eigenstates.
    If the states happen to be ground states, these Hamiltonians are simply frustration-free parent Hamiltonians for the states.
    \item Type II: $H$ can be written as a linear combination of non-Hermitian strictly local terms that individually have $\{\ket{\psi_{n}}\}$ as eigenstates, but not as a linear combination of Hermitian strictly local terms that have those states as eigenstates.
    \item Type III: $H$ is not type I or type II, i.e., it cannot be written as a sum of strictly local terms (Hermitian or non-Hermitian) that have $\{\ket{\psi_{n}}\}$ as eigenstates. 
\end{enumerate}
According to these definitions, adding a type I Hamiltonian to a type II Hamiltonian gives a type II Hamiltonian, and adding a type I or type II to a type III Hamiltonian gives a type III Hamiltonian.
This naturally leads us to the definition of \textit{equivalence classes} of operators.\footnote{Strictly speaking, one should be more careful about ranges of operators involved while defining these classes, see App.~\ref{app:numerical} for a careful discussion.
However, for simplicity, in the main text we continue to use the loose definition.}
For example, two type II (resp. type III) operators are equivalent if some linear combination of them is a type I (resp. type I or type II) operator.
Such equivalence classes have showed up in the study of Quantum Many-Body Scars~\cite{moudgalya2023exhaustive}, and we discuss the connections in the next subsection.
Further, as we discuss later, we believe that these equivalence classes can dictate certain dynamical signatures of these Hamiltonians, and we show examples from the $W$ state. 
Before we move on, we mention some technical subtleties with the above definitions.
Implicit there are notions of ``bounded range'' and ``strict locality'' of operators appearing as terms in the Hamiltonian, and their precise meanings can get somewhat complicated when we think about possible rewritings of the Hamiltonian, i.e., regroupings of the terms.
For example, for $H$ to be type I, the original writing of the Hamiltonian of Eq.~(\ref{eq:localHam}) need not by itself have the property that $h_{[j]}$ themselves have the $\{\ket{\psi_n}\}$ as eigenstates, but perhaps some other rewriting $H = \sum_X h'_X$ does, however at the expense that the regroupings of terms can lead to a somewhat larger range of the terms.
A natural way to make the typeage more precise is to fix the range of terms in the set of Hamiltonian rewritings, e.g., by fixing the range of strictly local terms to be bounded by a fixed number $R_{\text{max}}$, while taking the system size $N$ to be much larger than $R_{\text{max}}$.
Then the structure of the equivalence classes of operators can also depend on this $R_{\text{max}}$, as was found in certain examples of QMBS~\cite{moudgalya2023exhaustive, moudgalya2023numerical}.
However, for the $W$ state illustrated in this work, we find that these complications are surpassed since we are able to characterize the set of parent Hamiltonians well-enough that we can make definite statements when $R_{\text{max}}$ is any finite number in the thermodynamic limit. 
\subsection{Connections to Quantum Many-Body Scars}\label{subsec:qmbsconn}
We now briefly discuss the technical connections of these different classes of operators to Quantum Many-Body Scars (QMBS), as some of these concepts have showed up naturally in that context.
Readers not interested in these connections can safely skip to the next section. 
Given a set of states $\{\ket{\psi_n}\}$, we can write down the full algebra of Hamiltonians that have these states as eigenstates.
As mentioned earlier, we refer to these eigenstates as QMBS since they are generically not ground states, and are non-thermal states in the middle of the spectrum. 
Since any Hermitian Hamiltonian with $\ket{\psi_n}$ as an eigenstate commutes with $\ketbra{\psi_n}$, we consider the algebra of conserved quantities, which is of the form
\begin{equation}
    \mC_{\text{scar}} = \lgen \ketbra{\psi_n} \rgen,
\label{eq:Cscar}
\end{equation}
where the symbol $\lgen \cdots \rgen$ denotes the associative algebra generated by all products and sums of the operators in it, and implicitly includes the identity operator and closure under Hermitian conjugation $\dagger$.
As previous works have demonstrated~\cite{moudgalya2023exhaustive}, this is analogous to algebras of conventional conserved quantities~\cite{moudgalya2022from}, with the difference being that the latter are generated by on-site unitary operators (equivalently, extensive local conserved quantities that are sums of on-site terms), whereas here in Eq.~(\ref{eq:Cscar}) the projectors are highly non-local operators. 
The algebra of all operators that commute with this set of conserved quantities is referred to as the \textit{bond algebra} $\mA_{\rm scar}$. 
In numerous examples of QMBS that are of interest, $\mA_{\rm scar}$ is \textit{locally generated}, i.e., is of the form
\begin{equation}
    \mA_{\rm scar} = \lgen \{\hat{H}_\gamma\}\rgen,
\label{eq:bondalgebras}
\end{equation}
where $\hat{H}_\gamma$ are Hermitian operators that are strictly local or extensive local.\footnote{Note that in \cite{moudgalya2023exhaustive} we made the distinction between \textit{local} and \textit{bond} algebras depending on whether the generators contain extensive local terms or not. For simplicity, we will not make that distinction here, and we refer to both of them as bond algebras.}
Both $\mA_{\rm scar}$ and $\mC_{\rm scar}$ are subalgebras of the algebra $\mathcal{A}_{\mathcal{H}}$ of all operators over the given Hilbert space $\mathcal{H}$.
Furthermore, they are closed under Hermitian conjugation, and contain the identity operator, making them examples of finite-dimensional von Neumann algebras or $\dagger$-algebras~\cite{landsman1998lecture, harlow2017}.
$\dagger$-algebras are special since they satisfy a powerful result known as the Double-Commutant Theorem (DCT)~\cite{landsman1998lecture, harlow2017, moudgalya2022from}.
When applied to $\mA_{\rm scar}$, the DCT says that  $\mC_{\rm scar}$ and $\mA_{\rm scar}$ are \textit{commutants}\footnote{The commutant of a subalgebra $\mathcal{A}\subseteq \mathcal{A}_{\mathcal{H}}$ is defined as $\mathcal{C}:=\{c\in \mathcal{A}_{\mathcal{H}}|ac=ca,\forall a\in \mathcal{A}\}$.} of each other, and this allows one to exhaustively determine the complete set of Hamiltonians that have a particular set of QMBS as eigenstates~\cite{moudgalya2023exhaustive}. 
The DCT, along with the locality of Hamiltonians, motivated the division extensive-local Hamiltonians in $\mA_{\rm scar}$ into two qualitatively distinct types in \cite{moudgalya2023exhaustive}.
There the focus was on bond algebras $\mA_{\rm scar}$ generated by \textit{strictly local} $\hat{H}_\gamma$, and on characterizing the extensive local Hamiltonians in terms of its structure w.r.t.\ these generators.
Hence Ref.~\cite{moudgalya2023exhaustive} defined type I extensive local QMBS Hamiltonians as those that can be written as a linear combination of strictly local terms in $\mA_{\rm scar}$ and referred to all other Hamiltonians as type II, with the same implicit notions of ``strictly local'' and ``bounded range'' as discussed in Sec.~\ref{subsec:classes}.
This definition of type I operators is equivalent to the one defined earlier, even though it did not explicitly impose the condition of Hermiticity on the strictly local operators.\footnote{A proof of this can be as follows. Suppose that there is a Hermitian Hamiltonian $H$ that can be written as $H = \sum_j{g_{[j]}}$, where $g_{[j]}$ is a range $R$ strictly local non-Hermitian operator in $\mA_{\rm scar}$.
Since $\mA_{\rm scar}$ is a $\dagger$-algebra, we have $g_{[j]}^\dagger \in \mA_{\rm scar}$.
$H$ can then be rewritten as $H = (H + H^\dagger)/2 = \sum_j{(g_{[j]} + g_{[j]}^\dagger)/2}$, which is a sum of strictly local Hermitian terms in $\mA_{\rm scar}$, i.e., with the same set of eigenstates.}
This definition of type I and type II Hamiltonians also naturally leads to the concept of \textit{equivalence classes} of type II Hamiltonians~\cite{moudgalya2023exhaustive}, similar to the discussion in Sec.~\ref{subsec:classes}.
These equivalence classes help conjecture the complete structure of Hamiltonians~\cite{moudgalya2023exhaustive, moudgalya2023numerical} that support a given set of QMBS.
This also connects to earlier QMBS literature that noticed the existence of distinct types of Hamiltonians for a given set of QMBS~\cite{Moudgalya2020Large, mark2020unified, Mark2020Eta, odea2020from, ren2020quasisymmetry, ren2021deformed, pakrouski2020many, pakrouski2021group}, i.e., roughly speaking, those that can be understood using the Shiraishi-Mori embedding framework~\cite{Shiraishi_2017}, and those that cannot.
However, this classification of Hamiltonians coming from the algebra perspective is quite restrictive, or incomplete, due to several reasons. 
First, Ref.~\cite{moudgalya2023exhaustive} focused on bond algebras generated by strictly local terms, since it primarily considered the relation of the extensive local Hamiltonians to the generators of the algebra. 
For many standard examples of QMBS, such as the ferromagnetic spin-1/2~\cite{Mark2020Eta}, spin-1 XY~\cite{Iadecola_2019}, $\eta$-pairing~\cite{Mark2020Eta, moudgalya2020eta}, or Affleck–Kennedy–Lieb–Tasaki (AKLT) towers of states~\cite{Moudgalya_2018exact, Moudgalya_2018entanglement}, the bond algebra is generated by strictly local terms only if the entire QMBS manifold is \textit{degenerate}, i.e., if the commutant is of the form 
\begin{equation}
    \mC^{\rm deg}_{\rm scar} = \lgen \{\ketbra{\psi_{n}}{\psi_{m}} \}\rgen.
\label{eq:degeneratescar}
\end{equation}
Bond algebras for non-degenerate QMBS, on the other hand, require the addition of certain extensive local operators to the set of generators of the algebra, and Ref.~\cite{moudgalya2023exhaustive} referred to these operators as \textit{lifting operators}.
However, this addition meant that there was no ``natural'' set of generators unlike in the degenerate case, since the Hamiltonian could itself be a valid generator of the algebra.
Hence in the non-degenerate case there was no sense in which one can ``classify" relations between the extensive local Hamiltonian and the algebra generators.
In this work, we attempt to circumvent these issues by proposing the classification of Sec.~\ref{subsec:classes} that is independent of the algebra and its generators.
Second, Ref.~\cite{moudgalya2023exhaustive}  referred to all operators that are not type I as type II, but it did not attempt to establish a finer characterization of type II operators.
It has since been found that certain type II QMBS Hamiltonians can be expressed as sums of strictly local non-Hermitian operators that have the same QMBS as eigenstates~\cite{Mark2020Eta, Omiya_2023fractionalization, wang2024generalizedspinhelixstates}, hence there appears to be more structure than was identified in \cite{moudgalya2023exhaustive}.
Moreover, extensive local lifting operators were neither considered to be type I nor type II in \cite{moudgalya2023exhaustive} due to the algebra motivation discussed above. 
Yet, they are valid extensive local Hermitian operators that can have the same set of eigenstates, and these need to be treated on the same footing as extensive local Hamiltonians. 
As we find in the case of the $W$ state, there are lifting operators that do not have a decomposition in terms of any set of strictly local operators with the same set of eigenstates, hence we choose to refer to such operators as type III in Sec.~\ref{subsec:classes}.
\subsection{General structure of Parent Hamiltonians}
With all locality constraints imposed, the most general extensive local Hamiltonians with a particular set of eigenstates is of the form
\begin{equation}
    H = c_{\rm I} H_{\rm type\ I} + c_{\rm II} H_{\rm type\ II} + c_{\rm III} H_{\rm type\ III}.
\label{eq:mostgeneral}
\end{equation}
In the above equation, $H_{\rm type\ II}$ and $H_{\rm type\ III}$ may contain several inequivalent type I and type II operators, with potentially finer distinctions stemming from such varieties, which we leave unspecified in the above schematic writing.
Locality restrictions on $H$ also imposes several constraints on the allowed eigenstates $\{\ket{\psi_n} \}$. 
For example, it is known that enforcing a certain state to be an eigenstate can also lead to some other states to be eigenstates.
In the next section, we show an example of this in the context of the $\ket{W}$ state, where locality and the condition that the $\ket{W}$ state is an eigenstate also forces the $\ket{\overline{0}}$ state to be an eigenstate.
In the rest of this work, we will illustrate how all these concepts play out in systems where the $W$ state is an eigenstate/QMBS, where all the algebras and comprehensive statements about Hamiltonian varieties can be rigorously proven.
%

%
\section{$W$ state as an eigenstate/QMBS}\label{sec:WstateQMBS}
\subsection{Summary of main results}
We consider a one-dimensional system of $N$ qubits with the Hilbert space $\mathcal{H} = \left(\mathbb{C}^2\right)^{\otimes N}$, and extensive local Hermitian Hamiltonians $H$ that have the $W$ state as an eigenstate. 
For such Hamiltonians, we show the following theorem (with the proof in App.~\ref{sec:proof}, and summary of main ideas below):
\begin{thm}
\label{thm:HW}
Any extensive local Hermitian Hamiltonian $H$ that satisfies $H \ket{W} = E \ket{W}$ for some $E \in \mathbb{R}$, and is of the form of Eq.~(\ref{eq:localHam}) with $h_{[j]}$ being at most a range-$R$ strictly local,\footnote{We assume that $N \gg R$, although for the proof shown in App.~\ref{sec:proof}, $N > 2R$ suffices.} can be rewritten as
\begin{align}
    H =  \Omega \mathds{1} + \omega \hat{N}_{\rm tot} + t H_{\rm ImHop} + \sum_{X,\,|X|\leq R_{\text{max}}} h_X \quad,
    \label{eq:HW}
\end{align}
where $\Omega,\omega,t\in \mathbb{R}$ with $\omega,t$ having bounded absolute value, $\Omega\leq O(N)$, and $\Omega + \omega = E$.
Further, $h_X$ are Hermitian operators such that $h_X \ket{W} = 0$ and supported on a contiguous region $X$ of range $|X|\leq R_{\text{max}}$, where it suffices to choose $R_{\text{max}} = 2R$, with the sum ranging over $O(N)$ terms labeled by $X$.\footnote{W.l.o.g.\ we can take each $X$ to be a contiguous segment, colloquially referring to $|X|$ also as ``range" of $h_X$ or as ``diameter" of the region $X$.
All possible forms of $h_X$ terms are given in Tab.~\ref{tab:HIterms}.}
$H_{\rm ImHop}$ is given by
\begin{align}
    H_{\rm ImHop}=\frac{i}{2}\sum_{j=1}^N \left(s_j^\dag s_{j+1}-s_{j+1}^\dag s_j\right)\quad,
    \label{eq:HII}
\end{align}
with $s_j$ ($s_j^\dag$) being the hard-core boson annihilation (creation) operators with the convention that $s_j\ket{0}_j=0$,\footnote{Here the subscript ``ImHop'' refers to the hoppings being pure imaginary (the Hamiltonian is still Hermitian).} and $\hat{N}_{\rm tot}=\sum_j s_j^\dag s_j$ is the total number operator with $\hat{N}_{\rm tot}\ket{W}=\ket{W}$.
\end{thm}
This result illustrates the key structures of Hamiltonians with the $W$ state as an eigenstate, and is of the general form shown in Eq.~(\ref{eq:mostgeneral}).
The different components can be summarized as follows:
\begin{itemize}
    \item $\ket{\bar{0}}$ is always an eigenstate of any such $H$.
    Given Eq.~(\ref{eq:HW}), we can see that $\ket{\bar{0}}$ is an eigenstate of $\mathds{1}$, $\hat{N}_{\rm tot}$, and $H_{\rm ImHop}$, and we can also show that $h_X \ket{W} = 0 \implies h_X \ket{\bar{0}} = 0$ for any finite-range operator $h_X$ with $|X| < N$.
    \item $\hat{N}_{\rm tot}$ is the only operator in Eq.~(\ref{eq:HW}) that energetically separates $\ket{W}$ and $\ket{\bar{0}}$ states in the spectrum of $H$, and we use this to show that it is a type III operator from Eq.~(\ref{eq:mostgeneral}).
    We will show that there is only one such equivalence class of type III operators in Prop.~\ref{prop:HIIIequive}.
    \item Since $h_X$ are finite-range Hermitian terms that annihilate the $\ket{\overline{0}}$ and $\ket{W}$ states, the last term in Eq.~(\ref{eq:HW}) is a type I Hamiltonian, the first term in Eq.~(\ref{eq:mostgeneral}).
    Their exact forms in the operator string basis will be given in Lemma~\ref{lem:bondalg} and Tab.~\ref{tab:HIterms}.
    \item As we will show, $H_{\rm ImHop}$  cannot be written as a sum of $h_X$-type terms, but can be written in terms of strictly local non-Hermitian terms, hence it is a type-II Hamiltonian.\footnote{Notice that $H_{\rm ImHop}$ is only finite-range if the system is periodic, i.e., site $j+N := j$, with the property that $H_{\rm ImHop}\ket{W}=0$, and $H\ket{W}=(\Omega + \omega)\ket{W}$. On a chain with OBC we must have $t=0$ for the Hamiltonian to stay extensive-local (since $s_N^\dagger s_1$ term has range $N$ in this case).}
    \item Tied to the presence of the exact eigenstates $\ket{\overline{0}}$ and $\ket{W}$, $H$ has a variety of long-lived low-entanglement states that have vanishing energy variance in the thermormodynamic limit, also known as asymptotic scars in the QMBS literature~\cite{Gotta2023AsymptoticScars, Kunimi2023AsymptoticScarsDMI, ren2024quasi, kunimi2025systematic, Lin_2020slow}.
    We will show their existence in Sec.~\ref{sec:asymptoticscars} and the associated qualitative dynamics is different depending on whether $t = 0$ or $t \neq 0$, i.e., if the Hamiltonian is type I or type II.
\end{itemize}
Bolstered by the intuition gained from the $W$ state and Thm.~\ref{thm:HW}, in Sec.~\ref{sec:generalQMBS} we derive statements which govern all QMBS systems.
In particular, we demonstrate how one can distinguish the distinct Hamiltonian types I, II, and III via a truncation procedure, summarized in Thm.~\ref{thm:Hermitiancut}.
This result allows us to make statements regarding short range entangled states as QMBS in Sec.~\ref{sec:SREQMBS}.
Finally, in Sec.~\ref{sec:dynamical} we explore further dynamical signatures of type I versus type II Hamiltonians from the $W$ state perspective.
We start by presenting an outline of the proof for Thm.~\ref{thm:HW} in the following section.
\subsection{Operator basis and constraints}
We now provide an overview of the operator basis and technique used to show Thm.~\ref{thm:HW}, while the full proof can be found in App.~\ref{sec:proof}. 
Let us first define a convenient basis for strictly local and extensive local operators that are composed of finite-range terms. 
\begin{defin}
\label{def:opbasis}
On a tensor-product Hilbert space of $N$ qubits with PBC ($j+N \equiv j$) with $\mathcal{H}=\left(\mathbb{C}^2\right)^{\otimes N}$, we define a set $\mathcal{L}=\{L\}$ of local operators $L$ that act non-trivially on at most a finite-range $R$ of neighboring qubits.
$L$ are chosen to be
\begin{equation}
L^{n^+_{\ell} \dots n^+_{\ell+R-1}}_{n^-_{\ell} \dots n^-_{\ell+R-1}} = \prod_{j=\ell}^{\ell+R-1}(s^\dag_{j})^{n^+_j}(s_{j})^{n^-_j}\quad,
\label{eq:Lbasis}
\end{equation}
where $\ell\in \{1,\dots, N\}$, $s^\dag_j$ ($s_j$) is the hard-core boson creation (annihilation) operator on $\mathbb{C}^2 = \text{span}\{\ket{0}, \ket{1}\}$ at site $j$, and $n_j^\pm \in \{0, 1\}$.
Note that we are essentially using operator strings with the following on-site operators:
\begin{align}
    \mathds{1}_j &= \ketbra{0}_j + \ketbra{1}_j \,, \;\;(n_j^+ \!=\! n_j^- \!=\! 0) \,, \label{eq:opbasis1} \\
    s_j^\dagger &= \ketbra{1}{0}_j \,, \;\; (n_j^+ \!=\! 1, n_j^{-} \!=\! 0) \,, \\
    s_j &= \ketbra{0}{1}_j \,,\;\; (n_j^+ \!=\! 0, n_j^- \!=\! 1) \,, \\
    s_j^\dagger s_j &= \ketbra{1}{1}_j \,, \;\; (n_j^+ \!=\! 1, n_j^- \!=\! 1) \,. \label{eq:opbasis4}
\end{align}
This is different from the more familiar Pauli operator basis, but it is a valid full basis whose ``normal-ordered'' character with respect to the $\ket{\bar{0}}$ state is particularly convenient for our arguments.\footnote{For a complete specification of the linearly independent string operator basis we could require, e.g., for operators of range precisely $R$: $(n_\ell^+, n_\ell^-) \neq (0,0)$, $(n_{\ell+R-1}^+, n_{\ell+R-1}^-) \neq (0,0)$---i.e., non-trivial start and end of the exhibited string---and $\ell$ running over all sites (if we want operators of range up to $R_{\text{max}}$, we simply combine $R = 0, 1, \dots, R_{\text{max}}$, including the identity operator as $R = 0$).
On the periodic chain, the string can ``pass'' through the ``PBC link'' $N \to N+1 \equiv 1$, and the identification of the start and end of the string is unambiguous as long as $L > 2R$.
Conditions like these, assuming $L$ is sufficiently large compared to some fixed range $R_{\text{max}}$, are usually implicit in our discussion of operator ranges in the PBC system.}
For convenience, we often classify the basis elements by their total $n$ and $m$  numbers, where $n := \sum_j n_j^+ \leq R$ counts the number of creation operators, and $m := \sum_j n_j^- \leq R$ counts the number of annihilation operators.
Correspondingly, it is often convenient to simply exhibit the non-trivially acting operators, writing
\begin{equation}
s_{j_1}^\dagger \dots s_{j_n}^\dagger s_{k_1} \dots s_{k_m} ~,
\end{equation}
with $j_1, \dots, j_n, k_1, \dots, k_m \in A$, where $A$ is a contiguous region of size $|A|\leq R$ (note that $\{j_1, \dots, j_n\}$ and $\{k_1, \dots, k_m\}$ can have non-zero overlap corresponding to cases $n_j^+ \!=\! n_j^- \!=\! 1$).
This set forms an additive basis for the usual strictly local operators (i.e., those with finite support), and, more generally, all extensive local operators (i.e., those that are a sum of strictly local operators).
Note that without loss of generality we exclude the identity operator on $\mathcal{H}$ from $\mathcal{L}$.
\end{defin}
Consider an extensive-local operator $G$ that obeys the following relation
\begin{align}
    G\ket{W} = \lambda \ket{W} \quad,
    \label{eq:conditionW}
\end{align}
where $\lambda \in \mathbb{C}$ depends on $G$, and we initially do not impose Hermiticity.
We can then consider the expansion:
\begin{equation}
G = \sum  c_{j_1, \dots, j_n}^{k_1, \dots, k_m} s_{j_1}^\dagger \dots s_{j_n}^\dagger s_{k_1} \dots s_{k_m} ~,
\label{eq:Gexpand}
\end{equation}
where the sum is over all distinct operator basis strings.
In App.~\ref{sec:proof}, we derive the constraints on the coefficients in this expansion that arise from Eq.~\eqref{eq:conditionW}, which are given in Tab.~\ref{tab:Lterms}.
We then see that the conditions can be exhaustively classified by the operator's $n$ and $m$ numbers.
\begin{table}[t]
    \centering
    \begin{tabular}{c|c|c|c}
        Operators in $G$, Eq.~(\ref{eq:Gexpand}) & $n$ & $m$ & Conditions\\\hline
        $c^{j_1...j_n} s_{j_1}^\dag ...s_{j_n}^\dag$ & $\geq 1$ & 0 & $c^{j_1...j_n}=0$\\
        $c^{j_1...j_n}_{k_1...k_m} s_{j_1}^\dag ...s_{j_n}^\dag s_{k_1}...s_{k_m} $ & $\geq 0$ & $\geq 2$ & None\\
        $\sum_k c^{j_1...j_n}_{k} s_{j_1}^\dag ... s_{j_n}^\dag s_{k}$ & $\geq 2$ & $1$ & $\sum_k c^{j_1...j_n}_k=0$ \\
        $\sum_k c_{k}  s_{k} $ & $0$ & $1$ & $\sum_k c_k=0$\\
        $\sum_k c^j_{k}s_j^\dag s_k$ & 1 & 1 & $\sum_{k}c^j_{k}= \lambda$
    \end{tabular}
    \caption{
    Characterization of an extensive local operator $G$ satisfying $G \ket{W} = \lambda \ket{W}$, via conditions on the expansion coefficients in the specific operator basis employing the hard-core boson creation/annihilation operator language.
    Additionally, there is also a general finite-range condition which implies that $j_1, \dots, j_n, k_1, \dots, k_m \in A$ where $A$ is a contiguous region with $|A| \leq R$.
    This means that $n \leq R$ and $m \leq R$.
    $G$ does not need to be Hermitian up to this point.
   }
    \label{tab:Lterms}
\end{table}

To further reduce the problem, we study a specific set of terms in $G$ that are Hermitian, strictly local, i.e., requiring operators to have support on a contiguous region $X$ with $|X|\leq R_{\rm max}$, and possess $\ket{W}$ as an eigenstate.
These conditions further constrains the form of the operators in Tab.~\ref{tab:Lterms} to be of the forms shown in Tab.~\ref{tab:HIterms}.
This table can be viewed as the basis for type I parent Hamiltonians of $\ket{W}$.
We show that for all values of $n, m$, except when $n = m = 1$, Hermitian operators are spannable by these strictly local Hermitian operators.
The case $n = m = 1$ then reduces to a non-interacting (i.e., hopping) problem, which as we show in App.~\ref{sec:proof}, allows an explicit rewriting in terms of local operators in the last row in Tab.~\ref{tab:HIterms} up to $H_{\rm ImHop}$ and $\hat{N}_{\rm tot}$.
This procedure gives us $H$ in Eq.~\eqref{eq:HW}.
In the following sections, we will show that it is in fact impossible to rewrite $H_{\rm ImHop}$ and $\hat{N}_{\rm tot}$ in terms of strictly local Hermitian annihilators of $\ket{W}$, in particular in terms of the elements in Tab.~\ref{tab:HIterms}, which alludes to the fact that they are type II and III Hamiltonians, respectively.

\begin{table}[t]
    \centering
    \begin{tabular}{c|c|c|c}
        Hermitian Terms $h_X$ & $n$ & $m$ & Conditions\\\hline
        $c^{j_1...j_n}_{k_1...k_m} s_{j_1}^\dag ... s_{j_n}^\dag s_{k_1} ... s_{k_m} + \text{H.c.}$  & $\geq 2$ & $\geq 2$ & None\\
        $\sum_k c^{j_1...j_n}_{k} s_{j_1}^\dag ... s_{j_n}^\dag s_{k} + \text{H.c.}$ & $\geq 2$ & $1$ & $\sum_k c^{j_1...j_n}_{k}=0$ \\
        $ \sum_k c^j_k s_j^\dag s_k + \text{H.c.}$ & 1 & 1 & $\sum_{k}c^j_{k}=\lambda,\lambda=0$ \\
    \end{tabular}
    \caption{
    Continuing from Tab.~\ref{tab:Lterms}, adding Hermiticity conditions on $G$ such that $\lambda\in\mathbb{R}$, we present the forms of all strictly local Hermitian operators $h_X$, that have $\ket{W}$ as an eigenstate, and can form components of $G$ in Eq.~\eqref{eq:Gexpand}.
    Here, strictly local refers to $|X|\leq R_{\rm max}$ for some large enough system-size independent $R_{\rm max}$, where w.l.o.g. $R_{\rm max}=2R$ in this case.
    The conditions arise from the requirement that the operators possess $\ket{W}$ as an eigenstate, which when combined with the strict locality forces $\lambda=0$, as is shown in the App.~\ref{sec:proof} and Prop.~\ref{prop:HXdegW0}.
    Recall that there is also a general finite-range condition which implies that $j_1,...,j_n,k_1,...,k_m\in A$ where $A$ is a contiguous region with $|A|\leq R$.
    Alternatively, we can also view such entries as giving a basis of Hermitian type I operators in the hard-core boson creation/annihilation operator language.
    }
    \label{tab:HIterms}
\end{table}
\subsection{$\ket{\bar{0}}$ must also be an eigenstate}
It is known that $\ket{W}$ cannot be the unique ground state of any local Hamiltonian, and if it is a ground state then so is $\ket{\bar{0}}$~\cite{gioia2024wstateuniqueground}.
In cases where $\ket{W}$ is not a ground state, but just some eigenstate of \textit{some} extensive-local operator (in fact, it may be non-Hermitian), the relationship between $\ket{W}$ and $\ket{\bar{0}}$ can be extended further to the following statement:
\begin{cor}
\label{cor:W0deg}
If $\ket{W}$ is an eigenstate of an extensive-local operator, then $\ket{\bar{0}}$ is also an eigenstate.
\end{cor}
This corollary follows from the proof in App.~\ref{sec:proof} and Tab.~\ref{tab:Lterms}, where it is shown that $\ket{W}$ is never an eigenstate of an operator containing nontrivially (i.e., with nonzero amplitudes) terms that are purely creation operators, e.g., $s_j^\dag$ or $s_{j_1}^\dag \dots s_{j_n}^\dag$, which implies that the normal-ordered basis operators that contribute to $G$ must contain at least one annihilation operator.
This means that $\ket{\bar{0}}$ is also an eigenstate of $G$ with eigenvalue $0$, and therefore an eigenstate of $H$ in Eq.~\eqref{eq:HW} with eigenvalue $\Omega$ (solely determined by the identity operator).
Notice that the statement is based purely on locality, and does not rely on Hermiticity of the operator and holds even when the extensive-local operator is non-Hermitian.
For Hermitian operators where we may be interested in ground states, it also follows from the form of $H$ in Eq.~(\ref{eq:HW}) that if $\ket{W}$ is a ground state, so is $\ket{\bar{0}}$, as was shown in \cite{gioia2024wstateuniqueground}.
\begin{prop}
\label{prop:Wgroundstate}
    For an extensive-local Hamiltonian with bounded local terms and an $N$-independent $\omega$ in Eq.~(\ref{eq:HW})~\footnote{Such a requirement is fulfilled for `regularly-constructed' Hamiltonians, where adding sites does not change the Hamiltonian terms far away from the insertion point.
    Such a condition forbids a chemical potential term that depends on the length of the system, e.g., $\omega = 1/N$.}, if $\ket{W}$ is a ground state, then so is $\ket{\bar{0}}$.
\end{prop}
The full proof is given in Appendix~\ref{app:nonuniqueGS}.
However, the main idea is as follows.
Observe that the only way for $\ket{W}$ to have lower energy than $\ket{\bar{0}}$ is to have $\omega < 0$ in Eq.~(\ref{eq:HW}), since all $h_X$ terms also annihilate $\ket{\bar{0}}$.
However one can show that the state
\begin{equation}
    \ket{W^2} \defn \frac{1}{\sqrt{\binom{N}{2}}} \sum_{j < k} s_j^\dagger s_k^\dagger \ket{\bar{0}}\quad,
    \label{eq:W2}
\end{equation}
has energy expectation value
\begin{equation}
    \bra{W^2}H\ket{W^2}=2\omega+O(N^{-1})\quad
    \label{eq:W2energy}
\end{equation}
lower than $\omega$ for large enough $N$.
It follows that, although $\ket{W^2}$ is not necessarily an eigenstate of $H$, since it is orthogonal to $\ket{W}$ ($\braket{W^2}{W} = 0$), there would exist some other eigenstate that must have an even lower energy than $\ket{W}$ in the thermodynamic limit.
    
We will also encounter the $\ket{W^2}$ later (in Sec.~\ref{sec:asymptoticscars}) as an example of an asymptotic scar state.
\subsection{Different types of parent Hamiltonians}
\label{subsec:Wdifferenttypes}
In this section, we will explore some illustrative examples of possible Hamiltonians, all of which are of the form of $H$ in Eq.~\eqref{eq:HW}.
Note that the full bond algebra of operators that have $\ket{\bar{0}}$ and $\ket{W}$ as eigenstates can be explicitly written down using ideas from the Shiraishi-Mori construction~\cite{Shiraishi_2017,  moudgalya2023exhaustive}, which we discuss in App.~\ref{app:algebras}.  
While the algebra perspective can in principle be used to illustrate the different types of parent Hamiltonians,  it turns out to be simpler to use alternate techniques to show these,  as we will illustrate. 
All type I Hamiltonians can be written as a sum of terms in Tab.~\ref{tab:HIterms}.
Let us note a specific type I Hamiltonian of interest:
\begin{align}
    H_{\rm ReHop}&=\frac{1}{2}\sum_{j=1}^N \left(s_j^\dag s_j+s_{j+1}^\dag s_{j+1}-s_j^\dag s_{j+1}-s_{j+1}^\dag s_j\right),
    \label{eq:HI}
\end{align}
where the subscript `ReHop' refers to the jumping amplitudes being all real, and the term in the parentheses is a range-2 Hermitian operator that annihilates $\ket{W}$ and satisfies the condition in Tab.~\ref{tab:HIterms}.
It follows that $H_{\rm ReHop}$ is a type I Hamiltonian.
On the other hand, recall that type II parent Hamiltonians for a set of states $\{\ket{\psi_n}\}$ cannot be written as a sum of strictly local Hermitian terms that have the same set of eigenstates $\{\ket{\psi_n}\}$, but can be written as a sum of strictly local non-Hermitian terms with those eigenstates.
In fact, $H_{\rm ImHop}$ in Eq.~\eqref{eq:HII} is a type II Hamiltonian (only with PBC, since it is not a local operator with OBC) in relation to the $\ket{W}$ state, as we will show below.
\begin{prop}
\label{prop:HII}
On a periodic system, $H_{\rm ImHop}$ is a type II Hamiltonian w.r.t.\ the $W$ state.
\end{prop}
The detailed proof of this statement is given in App.~\ref{app:typeIIness}.
However, here we will present the idea behind the proof since it illustrates notable differences between type I and the other types of Hamiltonians. 
The key observation is that for any physically reasonable type I Hamiltonian $H_\text{I}$, i.e., one that has bounded norm $O(N)$, the variational state
\begin{equation}
    \ket{W_q}:= e^{-i q \sum_j j \hat{n}_j} \ktW = \frac{1}{\sqrt{N}} \sum_j e^{-i q j} s_j^\dagger \ket{\bar{0}}\quad,
    \label{eq:Wqboost}
\end{equation}
where $q=\frac{2\pi m}{N}$ with $m\in\mathbb{Z}$, has an energy expectation value that is bounded by
\begin{equation}
    \big|\bra{W_q}H_{\rm I} \ket{W_q} \big| \leq O(q^2)\quad.
\end{equation}
On the other hand, $\ket{W_q}$ is an eigenstate of $H_{\rm ImHop}$ with
\begin{align}
    H_{\rm ImHop}\ket{W_q}= \sin q\ket{W_q}\quad,
\end{align}
such that there is a contradiction with $H_{\rm ImHop}$ being a type I Hamiltonian if $q\ll 1$ with $\bra{W_q}H_{\rm ImHop}\ket{W_q}=q+O(q^2)$.
We can see this difference explicitly when considering the dispersion of $H_{\rm ReHop}$ (a type I Hamiltonian) versus $H_{\rm ImHop}$. $\ket{W_q}$ is an exact eigenstate of $H_{\rm ReHop}$, with relationship
\begin{align}
    H_{\rm ReHop}\ket{W_q}= \left(1-\cos q\right)\ket{W_q}\quad,
\end{align}
such that $\bra{W_q}H_{\rm ReHop}\ket{W_q}=O(q^2)$ when $q\ll 1$, as expected of a type I Hamiltonian.
We show a pictorial representation of this difference in Fig.~\ref{fig:HIvsHII}.

Finally, we notice that $H_{\rm ImHop}$ can be rewritten as
\begin{equation}
    H_{\rm ImHop}= \frac{i}{2} \sum_j (s_j^\dag s_{j+1}- s_{j+1}^\dag s_j-s_j^\dag s_j+ s_{j+1}^\dag s_{j+1})\,,
\label{eq:HImHop_nonherm_annih}
\end{equation}
where the term inside the parentheses is a strictly local non-Hermitian operator that possesses $\ket{W}$ as an eigenstate with eigenvalue 0, i.e., annihilates $\ket{W}$. This shows that $H_{\rm ImHop}$ is a type II Hamiltonian, and not a type III Hamiltonian.

\begin{figure}
    \centering
    \includegraphics[width=\linewidth]{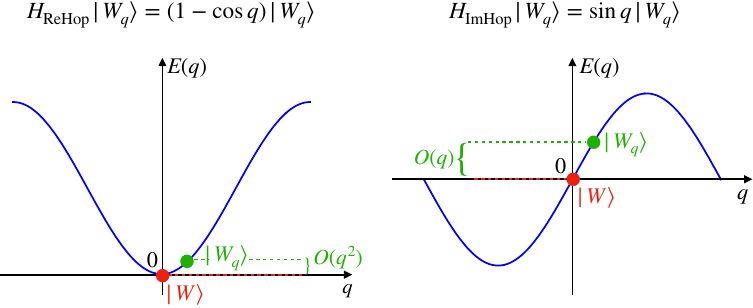}
    \caption{Dispersion of a type I Hamiltonian $H_{\rm ReHop}$ [Eq.~(\ref{eq:HI})] versus a type II Hamiltonian $H_{\rm ImHop}$ [Eq.~(\ref{eq:HII})].
    Here, we have diagonalized the respective Hamiltonians for the single-particle states $\ket{W_q}$. 
    We see that the variational state $\ket{W_q}$ with $q\ll1$ has energy $O(q^2)$ for $H_{\rm ReHop}$, while it is linear for $H_{\rm ImHop}$, exemplifying a key difference between the two types of Hamiltonians.
    While the figure is an illustration for free-particle Hamiltonians, the qualitative difference between ``dispersions'' $q^2$ vs $q$ holds for trial energies of generic type I vs type II Hamiltonians, see text for details.
    }
    \label{fig:HIvsHII}
\end{figure}
One may wonder how many equivalence classes of type II parent Hamiltonians there are when we demand that $\ket{W}$ is an eigenstate.
Recall that the equivalence relation is the addition of strictly local type I operators, given in Tab.~\ref{tab:HIterms}.
We have two candidates: $H_{\rm ImHop}$ and $\hat{N}_{\rm tot}$.
Here, we will show that $\hat{N}_{\rm tot}$ cannot be spanned by strictly local non-Hermitian operators that possess $\ket{W}$ as an eigenstate, leaving $H_{\rm ImHop}$ the representative member of the only type II Hamiltonian equivalence class.
To show this, let us first derive the following proposition.
\begin{prop}
    \label{prop:HXdegW0}
    For some operator $g_X$ that has support on a patch $X$ for $|X| < N$, if $g_X \ket{W} = \lambda \ket{W}$, then $g_X \ket{\bar{0}}=\lambda\ket{\bar{0}}$. 
    Note that $g_X$ is not necessarily Hermitian.
\end{prop}
\begin{proof}
    This essentially follows from our proof of Thm.~\ref{thm:HW} in Sec.~\ref{sec:proof} since we did not use Hermiticity to derive Tab.~\ref{tab:Lterms}.
    However, here we present an independent and simpler argument for our purpose.
    If $g_X\ket{W} = \lambda \ket{W}$ and $|X| < N$, we can perform a Schmidt decomposition of $\ket{W}$ on regions $X$ and $X^c$, where $X^c$ is the complement of $X$:
    \begin{equation}
    \ket{W}=\sqrt{\frac{|X|}{N}}\ket{W}_X\otimes\ket{\bar{0}}_{X^c}+\sqrt{\frac{|X^c|}{N}}\ket{\bar{0}}_X\otimes\ket{W}_{X^c}\quad,
    \label{eq:WSchmidt}
    \end{equation}
    where by assumption $|X^c|>0$, and $\ket{\bar{0}}_X:=\otimes_{j\in X}\ket{0}$, $\ket{W}_X:=\frac{1}{\sqrt{|X|}}\sum_{j\in X} s_j^\dag\ket{\bar{0}}_X$, and similarly for states on $X^c$.
    Writing down the eigenvalue equation for $g_X$, we obtain
    \begin{align}
        &\sqrt{\frac{|X|}{N}} (g_X - \lambda)\ket{W}_X \otimes \ket{\bar{0}}_{X^c} \nonumber\\
        &+ \sqrt{\frac{|X^c|}{N}} (g_X - \lambda)\ket{\bar{0}}_X \otimes \ket{W}_{X^c}= 0\quad.
    \end{align}
    Since $\ket{\bar{0}}_{X^c}$ and $\ket{W}_{X^c}$ are linearly independent, we obtain that $g_X \ket{W}_X = \lambda\ket{W}_X$ and $g_X\ket{\bar{0}}_X = \lambda\ket{\bar{0}}_X$.
    The latter condition directly implies that $g_X \ket{\bar{0}} = \lambda \ket{\bar{0}}$. 
\end{proof}
The following two corollaries follow directly from this proposition.
\begin{cor}
    $\hat{N}_{\rm tot}$ cannot be written as a sum of strictly local operators (Hermitian or non-Hermitian) that have $\ket{W}$ as an eigenstate.
\label{cor:Ntotnotnonherm}
\end{cor}
\begin{proof}
    Using Prop.~\ref{prop:HXdegW0}, it immediately follows that $\hat{N}_{\rm tot}$ cannot be spanned by general (i.e., including non-Hermitian) strictly local operators since the former splits the degeneracy between $\ket{W}$ and $\ket{\bar{0}}$ whereas the latter keeps them degenerate.
\end{proof}
(As an aside, we note that the above arguments generalize easily to the higher-particle-number Dicke states:
If $g_X \ket{W^p} = \lambda \ket{W^p}$, assuming $|X|^c \geq p$, it follows that $g_X \ket{W^{p'}} = \lambda \ket{W^{p'}}$, $\forall p' \leq p$.
It then follows that $N_{\text{tot}}$ is type III with respect to $\ket{W^p}$.)

From the above proposition and corollary, it follows that there is only one equivalence class of type II Hamiltonians, as we will show in the following corollary.
\begin{cor}
    There exists only one equivalence class of type II parent Hamiltonians of $\ket{W}$, 
    with $H_{\rm ImHop}$ being a representative member. All other type II Hamiltonians can be written as additions of type I Hamiltonians upon $H_{\rm ImHop}$.
    \label{cor:equivalenceHII}
\end{cor}
\begin{proof}
    This follows immediately from the ability to write a generic extensive-local Hamiltonian $H$ in the form of Eq.~\eqref{eq:HW}, where we have shown in Corollary~\ref{cor:Ntotnotnonherm} that $\hat{N}_{\rm tot}$ is not a type II Hamiltonian.
    We see that all other terms are encapsulated in $h_X$, so there is only one equivalence class of type II Hamiltonians.
\end{proof}
\begin{prop}
\label{prop:HIIIequive}
    There is only one equivalence class of type III parent Hamiltonians of $\ket{W}$, with $\hat{N}_{\rm tot}$
    being a representative member. All other type III Hamiltonians can be written as additions of type I and II Hamiltonians.
\end{prop}
\begin{proof}
We have already shown in Corollary~\ref{cor:equivalenceHII} that $H_{\rm ImHop}$ is a type II Hamiltonian, and thus only $\hat{N}_{\rm tot}$ remains in Eq.~\eqref{eq:HW}, as a result of Thm.~\ref{thm:HW}.
From Cor.~\ref{cor:Ntotnotnonherm} we know that $\hat{N}_{tot}$ is not a superposition of strictly local non-Hermitian operators.
Thus, it follows that $\hat{N}_{\rm tot}$ is a type III Hamiltonian, of which there is only one equivalence class.
\end{proof}
Now that we have fully understood the Hamiltonian equivalence classes for $\ket{W}$, we will use our knowledge to explore asymptotic scars in the context of this system.
\subsection{Asymptotic QMBS}
\label{sec:asymptoticscars}
Given that the $\ket{W}$ can also be interpreted as QMBS of the associated Hamiltonians, we briefly discuss the associated asymptotic QMBS that appear for such Hamiltonians with the $\ket{W}$ and $\ket{\bar{0}}$ as exact scars. 
Asymptotic QMBS are states with low entanglement that are orthogonal to the exact QMBS states, they are \textit{not} exact eigenstates of the QMBS Hamiltonians but have an energy variance $\Delta H^2$ that vanishes in the thermodynamic limit~\cite{Gotta2023AsymptoticScars, Kunimi2023AsymptoticScarsDMI, Desaules2022SchwingerScars, Desaules2023weak}.
This low energy variance leads them to have diverging relaxation times in the thermodynamic limit under the dynamics of the QMBS Hamiltonian~\cite{Gotta2023AsymptoticScars, moudgalya2024symmetries}, since they are roughly related via Heisenberg's uncertainty principle as~\cite{PhysRevA.81.022113,Gotta2023AsymptoticScars}
    \begin{align}
         \Delta t \geq \frac{\hbar}{2\Delta E}\quad, \;\;\;\Delta E = \sqrt{\Delta H^2}.
    \label{eq:lifetime}
    \end{align}
Here we demonstrate certain asymptotic QMBS of the $\ket{W}$ state in the exhaustive class of Hamiltonians given by Eq.~(\ref{eq:HW}), and compute their variances for the different types of Hamiltonians.
We will first show that $\ket{W_q}$, introduced in Eq.~\eqref{eq:Wqboost}, which is a boosted $W$ state of momentum $q$, is an asymptotic QMBS for the whole $H$ family when $q = \frac{2\pi m}{N}$, $m \in \mathbb{Z}$ and $|m| \ll  N$.
\begin{prop}
\label{prop:asymptoticscar}
    $\ket{W_q}$ is an asymptotic QMBS for any extensive-local Hamiltonian with $\ket{W}$ as an exact QMBS, when $q = \frac{2\pi m}{N}$ with $m \in \mathbb{Z}$, $|m| \ll N$.
\end{prop}
The proof is given in App.~\ref{app:asymptoticscar}. 
There we show that the variance of the energy of $\ket{W_q}$ for small $q \ll 1$ is bounded by
    \begin{align}
        \Delta H_{W_q}^2 \defn \langle H^2 \rangle_q - \langle H \rangle_q^2 \leq O(q^2) \quad,
    \end{align}
    where $H$ is the general Hamiltonian given in Eq.~(\ref{eq:HW}) and $\langle \dots \rangle_q \equiv \bra{W_q} \dots \ket{W_q}$.
We conclude using Eq.~(\ref{eq:lifetime}) that the $\ket{W_q}$ states have lifetime scaling as $\sim 1/q \sim N$.

\begin{figure}
    \centering
    \includegraphics[width=0.9\linewidth]{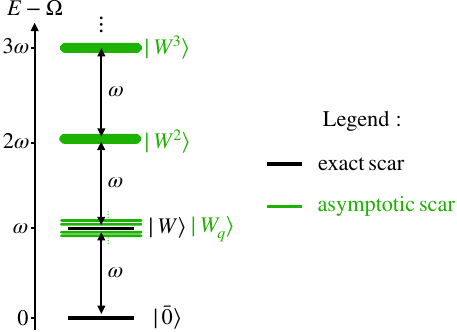}
    \caption{Relative energy locations of exact (black) and asymptotic (green) QMBS for Hamiltonians with $\ket{W}$ as an exact eigenstate (i.e., QMBS).
    The energy separation between $\ket{\bar{0}}$ and $\ket{W}$ is $\omega$ defined in Eq.~(\ref{eq:HW}), with asymptotic $\ket{W_q}$ QMBS (for $q \ll 1$) around $\omega$ (with lifetime $\gtrsim 1/q \sim N$) and $\ket{W^2}$ around $2\omega$ (with lifetime $\gtrsim \sqrt{N}$).
    We also prove that there exist asymptotic scars $\ket{W^p}$ at energies $p\omega$ for $p\in\{3,4,...\}$, where $\ket{W^p}$ is given by Eq.~\eqref{eq:Whigher}.
    }
    \label{fig:asymptoticscar}
\end{figure}
One may further wonder whether higher-particle number states, such as the aforementioned $\ket{W^2}$, could also be asymptotic QMBS.
$\ket{W^2}$ is a good candidate since it is an orthogonal state to both $\ket{\bar{0}}$ and $\ket{W}$, and also an eigenstate of $\hat{N}_{\rm tot}$ and $H_{\rm ImHop}$.\footnote{The fact that the $\ket{W^2}$ is an eigenstate of $H_{\rm ImHop}$ is a rather non-trivial fact known from the studies of spin-1/2 models with the ferromagnetic tower of states as exact eigenstates~\cite{Moudgalya_2022}: $\ket{\bar{0}}$, $\ket{W}$, $\ket{W^2}$ are all members of the tower while $H_{\rm ImHop}$ corresponds to so-called Dzyaloshinskii-Moriya interaction (DMI), which is also a type II Hamiltonian in that context~\cite{moudgalya2023exhaustive,wang2024generalizedspinhelixstates}.}
Here, we verify the intuition that just requiring that $\ket{W}$ is an exact scar implies that $\ket{W^2}$ is an asymptotic scar.
\begin{prop}
\label{prop:W2}
    $\ket{W^2}$, defined in Eq.~\eqref{eq:W2}, is an asymptotic scar with the energy variance upper-bounded by $O(N^{-1})$.
\end{prop}
The full proof is given in Appendix~\ref{app:W2}, where we show that the energy variation of $\ket{W^2}$ is bounded by
\begin{align}
    \Delta H_{W^2}^2 \defn \langle H^2 \rangle_{W^2} - \langle H \rangle_{W^2}^2 \leq O(N^{-1}) \quad,
    \label{eq:W2variance}
\end{align}
where $H$ is the general Hamiltonian given in Eq.~(\ref{eq:HW}), and $\langle \dots \rangle_{W^2} \defn \bra{W^2} \dots \ket{W^2}$.
This implies that the lifetime of $\ket{W^2}$ is $\sim\sqrt{N}$.
The key observations that lead to the result are:
(1) $\ket{W^2}$ is an eigenstate of both $H_{\rm ImHop}$ and $\hat{N}_{\rm tot}$, which reduces the variational energy considerations of $H$ to just the type I terms $\sum h_X$ terms.
(2) Any terms in $H^2$ is made up of $h_X h_Y$ terms -- when $X\cap Y=\emptyset$ such terms will always annihilate $\ket{W^2}$, which allows us to bound $\langle H^2 \rangle_{W^2}$.
Encouraged by this relatively simple demonstration, we are led to consider yet higher-particle number states and the following
\begin{prop}
\label{prop:Wpasymptoticscar}
    $\ket{W^p}$, defined as
    \begin{align}
        \ket{W^p} \defn \frac{1}{\sqrt{\binom{N}{p}}}\sum_{j_1 < \dots < j_n} s_{j_1}^\dag \dots s_{j_p}^\dag \ket{\bar{0}} \quad,
        \label{eq:Whigher}
    \end{align}
    is also an asymptotic scar for any fixed $p$ and $N \to \infty$.
\end{prop}
The proof is given in App.~\ref{app:Wpasymptoticscar}, and follows a similar logic to the proof for $\ket{W^2}$, where we employ successive Schmidt-decompositions of the state to calculate contributions to the variance of $H$ from its terms.
We depict our above findings on asymptotic scars in Fig.~\ref{fig:asymptoticscar}.
\subsection{ Non-Hermitian parent operators}
\label{sec:nonHerm}
Finally, let us now briefly mention how our results change when we loosen our condition of Hermiticity on the Hamiltonians with $\ket{W}$ as an eigenstate.
If we relax the Hermiticity condition, but still require an operator $G$, with $\ket{W}$ as an eigenstate, to be extensive-local, then the general form of this operator is given by
\begin{align}
    \label{eq:GW}
    G = \Omega\mathds{1}+\omega \hat{N}_{\rm tot} + \sum_{|X|\leq R_{\text{max}}} g_X\quad,
\end{align}
where $\Omega,\omega\in\mathbb{C}$ with $\omega$ having bounded absolute value and $|\Omega| \leq O(N)$, and $g_X$ are operators (potentially non-Hermitian) such that $g_X\ket{W}=0$ supported on a contiguous region $X$ of range $R_{\text{max}}$ (or less).
Similarly to the statement of Thm.~\ref{thm:HW}, if $G$ has an expression with terms of range bounded by $R$, we need $R_{\text{max}}$ of at most $2R$ to find the claimed rewriting, with comparable scaling of the number of terms and their norms as in the original expression.
Corollary~\ref{cor:W0deg} still holds in this context, and thus $\ket{\bar{0}}$ is also an eigenstate of $G$.
The proof of this form follows directly from arguments similar to those proving Thm.~\ref{thm:HW}, Prop.~\ref{prop:HII} (the part showing that $H_{\text{ImHop}}$ can be rewritten as a sum of non-hermitian local annihilators of $\ket{W}$), and Corollary~\ref{cor:Ntotnotnonherm}.
Note that with non-Hermitian extensive local operators, the distinction between type I and type II as defined in Sec.~\ref{subsec:classes} disappears, but there is still a notion of type III operators. 
These facts are evident in the exhaustive understanding of Hamiltonians provided by Eq.~(\ref{eq:GW}).
%

%
\section{Dynamical signatures of type I vs type II Hamiltonians}
\label{sec:dynamical}
\begin{figure*}
\includegraphics[scale=1]{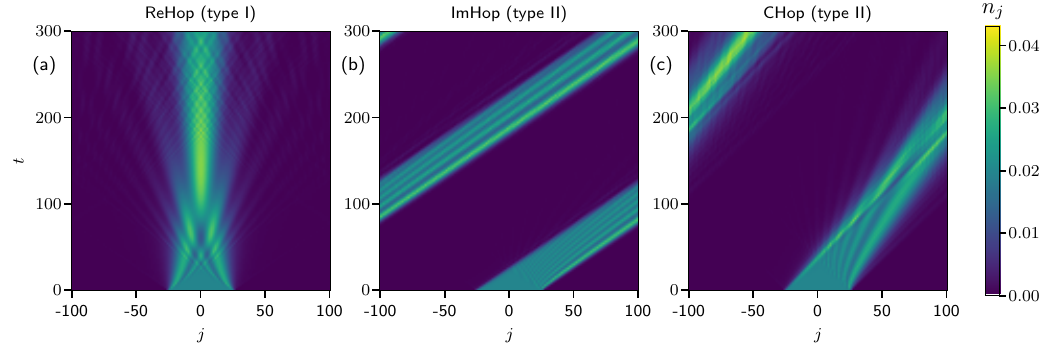}
\caption{\label{fig:Wdropletnumerics}Evolution of the $W$ droplet under three Hamiltonians: (a) $H_{\rm ReHop}$; (b) $H_{\rm ImHop}$; and (c) $H_{\rm CHop}$ (with $\alpha = \beta = 0.5$).
The system is $[-100,\dots,100]$ ($N=201$) with PBC and the droplet starts on $[-25,\dots,25]$ ($M=51$).
The colors indicate different values of $n_j(t) \defn |\!\braket{j}{\phi(t)}\!|^2$. 
Note that there is a clear directionality in the evolution in the ImHop and CHop cases, which indicates a type II component.
}
\end{figure*}
We now describe some scenarios where the different types of Hamiltonians can exhibit qualitatively different types of dynamics.
We focus on QMBS Hamiltonians with the $\ket{W}$ and $\ket{\bar{0}}$ states as eigenstates, and consider the following setup.
On a periodic chain of size $N$, we prepare an initial state that is a $W$-like ``droplet'' (or ``domain'') on a subregion of size $M$ and vacuum-like outside the subregion:
\begin{equation}
\ket{\psi_0} = \frac{1}{\sqrt{M}} \sum_{j=1}^M s_j^\dagger \ket{\bar{0}} = \ket{W}_{[1\dots M]} \otimes \ket{\bar{0}}_{[M+1\dots N]} ~. 
\end{equation}
We want to study the time evolution of this state under our QMBS Hamiltonians, $\ket{\psi(t)} = e^{-i H t} \ket{\psi_0}$.

We will specifically consider $H_{\text{ReHop}}$ [Eq.~(\ref{eq:HI})] vs $H_{\text{ImHop}}$ [Eq.~(\ref{eq:HII})] instances of type I vs type II Hamiltonians; we will also consider a combination of the two, 
\begin{equation}
H_{\text{CHop}} = \alpha H_{\text{ReHop}} + \beta H_{\text{ImHop}},\;\;\;\alpha,\beta\in\mathbb{R},
\label{eq:CHopdefn}
\end{equation}
which is a more general type II Hamiltonian.
Since these models preserve the particle number, we can view this setup as a single-particle dynamics problem.
We can label states with precisely one particle by $\ket{j} \equiv s_j^\dagger \ket{\bar{0}}$, $j = 1, \dots, N$.
The initial state is then simply ``$k = 0$ orbital'' on $[1, \dots, M]$,
\begin{equation}
\ket{\phi_0} = \frac{1}{\sqrt{M}} \sum_{j=1}^M \ket{j},
\end{equation}
and we study the orbital evolution $\ket{\phi(t)} = e^{-iHt} \ket{\phi_0}$.
The system evolves under a translationally invariant single-particle Hamiltonian $H^{\text{1-particle}}$, which we can diagonalize in the plane wave basis $\ket{q} \defn \frac{1}{\sqrt{N}} \sum_j e^{i q j} \ket{j}$, $q = 2\pi n/N, n=0, 1, \dots, N-1$, with dispersions:\footnote{
Note a small change in conventions for this section only: here $\ket{q} = \ket{W_{-q}}$ of Eq.~(\ref{eq:Wqboost}), and $H_{\text{ImHop}} = \frac{1}{2} \sum_j (-i s_j^\dagger s_{j+1} + \text{H.c.})$; this is to match more familiar writing of plane wave orbitals and to have positive dispersion slope of $\epsilon_q^{\text{ImHop}}$ at $q = 0$.
}
\begin{align}
& H^{\text{1-particle}} = \sum_q \epsilon_q \ketbra{q} ~, \\
& \epsilon_q^{\text{ReHop}} = w [1 - \cos(q)] ~, \label{eq:epsReHop} \\
& \epsilon_q^{\text{ImHop}} = w \sin(q) ~, \label{eq:epsImHop} \\
& \epsilon_q^{\text{CHop}} = \alpha w [1 - \cos(q)] + \beta w \sin(q) ~. \label{eq:epsCHop}
\end{align}
For the pure $H_{\text{ReHop}}$ and $H_{\text{ImHop}}$ models, we have denoted the corresponding hopping amplitudes as $w$ (in the earlier sections, in particular in Fig.~\ref{fig:HIvsHII}, we had set $w = 1$).
We have also included the more general type II example, Eq.~(\ref{eq:CHopdefn}), with dispersion $\epsilon^{\text{CHop}}_q$.
A simple calculation gives
\begin{align}
& \ket{\phi_0} = \sum_q f_q \ket{q} \implies \ket{\phi(t)} = \sum_q f_q e^{-i\epsilon_q t} \ket{q} ~, \\
& f_q = \frac{1}{\sqrt{M N}} \frac{\sin(qM/2)}{\sin(q/2)} e^{-iq (M+1)/2} ~,
\end{align}
allowing us to study properties of the time-evolved state $\ket{\phi(t)}$.
(Note that the exact $f_{q=0} = \sqrt{M/N}$ matches the formal $q \to 0$ limit of the above expression.)
Thus, the site occupations $n_j(t) \defn \bra{\phi(t)} n_j \ket{\phi(t)}$ are given by
\begin{equation}
n_j(t) = |\!\braket{j}{\phi(t)}\!|^2,\;\;\;
\braket{j}{\phi(t)} = \frac{1}{\sqrt{N}} \sum_q f_q e^{i q j} e^{-i \epsilon_q t}.
\label{eq:njdefn}
\end{equation}
Figure~\ref{fig:Wdropletnumerics} shows color plot visualizations of $n_j(t)$ for the three dispersions, Eqs.~(\ref{eq:epsReHop})-(\ref{eq:epsCHop}).
We see that the droplet persists in the same location for a long time in the type I case $H_{\text{ReHop}}$, while it moves ballistically in the type II cases $H_{\text{ImHop}}$ and $H_{\text{CHop}}$.
The droplet appears to eventually ``melt away'', already visibly in the $H_{\text{ReHop}}$ and $H_{\text{CHop}}$ cases, but not yet in the $H_{\text{ImHop}}$ case.
In addition, there is a ``directionality" and ballistic propagation of the droplet in the $H_{\rm ImHop}$ and $H_{\rm CHop}$ cases.

We can develop an analytical understanding of the observed behaviors by studying the overlap with the initial state (which is always of interest), as well as overlap with the $W$ droplet state translated by $G$ lattice sites $\ket{T_G \phi_0}$ (which is motivated by the ballistic propagation observed in some cases), given by
\begin{align}
\ket{T_G \phi_0} &\defn \frac{1}{\sqrt{M}} \sum_{j=G+1}^{G+M} \ket{j} ~, \nonumber\\
\braket{T_G \phi_0}{\phi(t)} &= \sum_q |f_q|^2 e^{i q G} e^{-i\epsilon_q t} =: 1 - \Upsilon_G(t,M,N) ~, \nonumber\\
\Upsilon_G(t,M,N) &= \frac{1}{M N} \sum_q \frac{\sin^2(qM/2)}{\sin^2(q/2)} (1 - e^{i (q G - \epsilon_q t)}) ~.
\label{eq:defUpsilon}
\end{align}
We can obtain the overlap with the initial state by setting $G = 0$, while for the type II Hamiltonians we will be also interested in time-dependent $G$, and we initially proceed by keeping $G$ general.
In the expression for the overlap, we have used $\sum_q |f_q|^2 = 1$ to focus on the difference of the overlap from $1$: for $G = 0$, the $\Upsilon_G$ gives the reduction of the overlap from the initial value of $1$, while for general $G$ it quantifies the ``distance'' of the time-evolved state from the $\ket{T_G\phi_0}$ state.\footnote{
Note that this anticipates that $\epsilon_{q=0} = 0$ in our cases of interest. If this were not true, we could instead consider $\sum_q |f_q|^2 e^{i q G} e^{-i\epsilon_q t} = e^{-i\epsilon_0 t} \sum_q |f_q|^2 e^{i q G} e^{-i (\epsilon_q - \epsilon_0) t} =: e^{-i\epsilon_0 t} [1 - \Upsilon_G(t,M,N)]$.
This is equivalent to noting that the $\omega N_{\text{tot}}$ term, by virtue of being a conserved quantity here, does not have any interesting dynamical effects.
\label{foot:epsq0}
}
We first observe that as long as we keep $M$ fixed, the quadratic vanishing of the denominator $\sin^2(q/2)$ at small $q$ (important in the thermodynamic limit) is completely controlled by the vanishing of the numerator $\sin^2(qM/2)$.
Hence, we can obtain a well-defined thermodynamic limit for a fixed $M$ and $N$-independent $G$,
\begin{equation}
\begin{aligned}
\Upsilon_G(t,M) &= \lim_{N \to \infty} \Upsilon_G(t,M,N) \\
&= \frac{1}{M} \int_{-\pi}^\pi \frac{dq}{2\pi} \,\frac{\sin^2(qM/2)}{\sin^2(q/2)} \, (1 - e^{i (q G-\epsilon_q t)}) ~,
\end{aligned}
\label{eq:Upsilon_tM}
\end{equation}
which we interpret as describing the corresponding overlap evolution from the initial $W$ droplet of size $M$ surrounded by the vacuum in an infinite system.

We now state the main analytical results on the motion of the droplet and its melting observed in the different cases in Fig.~\ref{fig:Wdropletnumerics}, and we refer to App.~\ref{app:dynamics} for the details of the computations.
\subsubsection{Early time behavior}
First, Eq.~(\ref{eq:Upsilon_tM}) can be analyzed at early times, by expanding $e^{i(q G - \epsilon_q t)}$ in the expression, which we do in App.~\ref{app:earlytime}.
This regime can also be analyzed using the real-space forms of the Hamiltonians in terms of the Hermitian vs non-Hermitian local annihilators (of the $\ktW$ and $\ktO$ states), which then initially act non-trivially only near the droplet boundaries.
We find that the early-time behavior of the fidelity decrease (studied by focusing on $G = 0$) is qualitatively similar for the two types.
However, there is a clear difference in the directionality that is already visible in the early-time computation.
This can be understood using the expectation values of current operators in the system, which vanish everywhere for $H_{\rm ReHop}$, while they are non-vanishing within the droplet for $H_{\rm ImHop}$ and $H_{\rm CHop}$. 
This leads to non-vanishing current divergence at the edges of the droplet, which in turn provide the directionality seen in Fig.~\ref{fig:Wdropletnumerics}.
\subsubsection{Intermediate time behavior}
\label{subsubsec:intermediate_time}
To highlight even more qualitative differences between the different types, we study the intermediate time regime $wt \gg 1$ with typical upper bounds of the form $wt \ll M^z$, where $z$ depends on the system and quantity being studied.
We first observe that in all our cases $\epsilon_{q=0} = 0$, hence in the expression in Eq.~(\ref{eq:Upsilon_tM}) at any finite $t$ and $G$, the vanishing denominator $\sin^2(q/2)$ at small $q$ is also controlled by the corresponding $1 - \cos(\epsilon_q t - q G)$ or $\sin(\epsilon_q t - q G)$ factors in the numerator, besides the aforementioned $\sin^2(qM/2)$ factor.
This leads to the appearance of  interesting intermediate-time scaling regimes as follows.
We are interested in the behaviors at ``large'' $M \gg 1$ and at a given ``largish'' $t \gtrsim w^{-1}$.
Since $\epsilon_{q=0} = 0$, we expect that ``important'' wave-vectors that determine $\Upsilon_G(t,M)$ in Eq.~(\ref{eq:Upsilon_tM}) are in the range of ``smallish'' $q_*(t)$ around $|\epsilon_{q_*(t)} t - q_* G| \sim O(1)$.
Specifically, we assume that $M$ is large enough such that $q_*(t) M \gg 1$.
In this regime the factor $\sin^2(qM/2)$ in Eq.~(\ref{eq:Upsilon_tM}) oscillates very quickly over the important $q$ integration range, and we can replace it by its average value of $1/2$, see Eq.~(\ref{eq:Upsilon_t_largeM}) in App.~\ref{app:inttime}.
Further, we define 
\begin{equation}
    \upsilon_G(t) = M \Upsilon_G(t,M) ~,
\end{equation}
which in this regime is independent of $M$ and characterizes the boundary behavior of the droplet in the regime $N \to \infty$ (thermodynamic limit) first and then $M \to \infty$ (very large initial domain size) second.
This gives a particular form of the time dependence of the overlaps of interest,
\begin{equation}
\braket{T_G \phi_0}{\phi(t)} = 1 - \frac{1}{M} \upsilon_G(t)
\end{equation}
in appropriate regimes $1 \ll wt \ll M^z$ specified below.
We now describe the behavior of $\upsilon_G(t)$ for various cases.
\subsubsection{Droplet motion and boundary melting}

To characterize the overall motion of the droplet, we consider the $G = 0$ calculations in the type I and type II cases separately.
For $H_{\rm ReHop}$, we use $\epsilon_q \approx \frac{1}{2}w q^2$ at small $q$ to obtain
\begin{equation}
\upsilon^{\text{ReHop}}_0(t) \approx \sqrt{\frac{wt}{\pi}} (1 + i) ~,\;\;\;1 \ll wt \ll M^2. 
\label{eq:upsilonReHop_main}
\end{equation}
By noting from Fig.~\ref{fig:Wdropletnumerics} that the overall droplet does not appear to move, we can view the corresponding decrease in the overlap with the initial state as a consequence of ``diffusive" melting happening at the domain boundary, which can be confirmed by explicit numerical evaluations discussed in  App.~\ref{app:dynamics}.
Turning to $H_{\text{ImHop}}$, we can use $\epsilon_q^{\text{ImHop}} \approx wq$ at small $q$ to derive
\begin{equation}
\upsilon^{\text{ImHop}}_0(t) \approx wt ~,\;\;1 \ll wt \ll M,
\label{eq:upsilonImHop_main}
\end{equation}
which, along with the overall motion of the droplet in Fig.~\ref{fig:Wdropletnumerics}, suggests that this is a consequence of ``ballistic'' motion of the droplet rather than boundary melting physics.
We can also argue that the behavior is similar for $H_{\rm CHop}$ of Eq.~(\ref{eq:CHopdefn}), as soon as $\beta \neq 0$, where the droplet has a similar ballistic motion as $H_{\rm ImHop}$ case. 
We thus see the qualitative difference between type I and type II already at the level of the overlap with the initial droplet.
In the type II cases, we can also characterize the boundary melting physics in addition to the overall ballistic motion by studying the overlap of $\ket{\phi(t)}$ with $\ket{T_G\phi_0}$, choosing optimal $t$-dependent $G$.
For an optimal choice of $G$, we expect this overlap to be close to $1$ if the droplet does not ``melt". 
By observing that the behaviors in Eqs.~(\ref{eq:upsilonReHop_main}) and (\ref{eq:upsilonImHop_main}) essentially appear due to the leading order $q$ behavior in the exponent in Eq.~(\ref{eq:Upsilon_tM}) (i.e., $\epsilon_q t - qG \approx w q t$ at $G=0$), for $H_{\rm ImHop}$, we can consider $G(t) = wt$, which ``cancels" the overall ballistic motion of the droplet and probes its melting. 
In this case, noting that $\epsilon_q^{\text{ImHop}} t - q G(t) \sim w t q^3$, we obtain
\begin{equation}
\upsilon_{G(t)=wt}^{\text{ImHop}}(t) \approx  A (wt)^{1/3} ~,\;\;1 \ll wt \ll M^3.
\label{eq:upsilon_Gwt_ImHop_main}
\end{equation}
where $A \approx 0.411$, and we can interpret the corresponding decrease in the overlap with the ``reference'' droplet as describing a sub-diffusive melting of the droplet at the boundaries, which we further confirm by numerics in App.~\ref{app:dynamics}.
An analogous computation can be performed for  $H_{\text{CHop}}$, where we can again ``cancel" the overall ballistic motion by choosing $G(t) = \beta w t$.
Noting that here we instead obtain $\epsilon_q^{\rm CHop} t - q G(t) \sim \alpha w t q^2$, we obtain 
\begin{align}
\upsilon^{\text{CHop}}_{G(t)=\beta w t}(t) \approx \sqrt{\frac{\alpha wt}{\pi}} (1 + i) ~,\;\;1 \ll \alpha wt \ll M^2, 
\label{eq:upsilon_G_CHop_main}
\end{align}
indicating diffusive melting rather than subdiffusive found for the pure $H_{\text{ImHop}}$ in Eq.~(\ref{eq:upsilon_Gwt_ImHop_main}).
\subsubsection{Generalizations to interacting cases}
While the above calculations were performed for the solvable free-particle instances of type II Hamiltonians, we expect some aspects of the $W$ domain ballistic motion up to largish times and slow melting to hold also for generic interacting type II Hamiltonians. 
Thus, we can view the $W_q$ ``asymptotic scars'' as very long-lived ``quasiparticles'' near $q=0$ playing in some ways the role of the free particles in the above calculations.
Including the effect of their finite lifetimes is not trivial and could modify the droplet melting even qualitatively.
However, we are encouraged by the fact that the parent Hamiltonians of the $W$ QMBS, with the exception of $N_{\text{tot}}$ (which we would exclude if we require the $\ktO$ and $\ktW$ states to be degenerate), are either type I or type II (Thm.~\ref{thm:HW}), and hence their action on the $W$ droplet is localized on the boundaries.
Moreover, following Thm.~\ref{thm:HW}, the interacting type II Hamiltonians are of the form of $H_{\rm ImHop} + H^{\rm int}_{\text{type-I}}$, where $H^{\rm int}_{\text{type-I}}$ is an interacting type I Hamiltonian.
Hence we conjecture that the intermediate time behavior of the droplet under such interacting Hamiltonians consists of chiral ballistic motion that we found for $H_{\rm ImHop}$ followed by diffusive boundary melting properties we found above for $H_{\text{CHop}}$.
We leave numerical explorations of such questions to future work.
\subsubsection{Extension to quenches from a $W^p$ droplet}
Finally, we mention extensions of these calculations to quenching from a $W^p$ droplet under the corresponding type I vs type II Hamiltonians, which we discuss in App.~\ref{app:Wp_droplet_dyn}.
We can in fact reuse the above results in the approximation where the $W^p$ state is viewed as a non-interacting Bose-Einstein Condensate (BEC) of particles, and the BEC droplet is then evolved under the corresponding particle dispersions.
We conjecture that this treatment is qualitatively accurate at low particle densities, and hence predict similar sharp dynamical difference between the two types of Hamiltonians, with particularly interesting distinction for intrinsic domain boundary dynamics for a non-zero density quench (i.e., when $\rho = p/M \neq 0$ and $\lim_{M\to \infty}[\lim_{N \to \infty}...]$); we refer to App.~\ref{app:Wp_droplet_dyn} for more details.

\section{Hamiltonian Types and Boundary Actions}
\label{sec:generalQMBS}
With all the observations from the $W$ state, we now turn to general theorems that all type I and type II parent Hamiltonians of a generic eigenstate $\ket{\psi}$ must satisfy, which are based on truncating the Hamiltonian to some part of the system.
To do this, let us first carefully define the \textit{support} of an operator.
\begin{defin}
    The \textit{support} of an operator $\mathcal{O}$ is all the sites $j$ for which there exists an on-site operator $\mathcal{O}_j$, i.e., one that is defined on the $j$th site local Hilbert-space and tensored identity on every other site, such that
    \begin{align}
        [\mathcal{O},\mathcal{O}_j]\neq 0\quad.
    \end{align}
    In particular, this means that the identity operator has a support of the empty set $\emptyset$. 
    Suppose now we have a (Hermitian or non-Hermitian) extensive local operator $\mathcal{P}$, expressed in terms of some local terms $L_n$ as:
\begin{equation}
    \mathcal{P} = \sum_{n}c_n L_n\quad.
\end{equation}
We can define its truncation to a region $Y$ by only keeping the terms $c_n L_n$ that are supported on the region $Y$.
Clearly, this depends on the choice of $L_n$'s: e.g., we can group several terms into one or split a term into several parts, and this would change the result of the truncation (though the qualitative physics should not depend on reasonable regroupings).
For a systematic mathematical analysis below, it is convenient (and in some arguments important) to view $\{L_n\}$ as a basis in the operator space such as the one in Def.~\ref{def:opbasis}, or the generalized Pauli string basis, and keep the basis fixed in all arguments.
However, any basis composed of operator strings where one of the elements of the local basis is the identity and all others are linearly independent would work.
Note that the truncation depends on the basis choice but is a precise fixed linear operation once the basis choice is fixed.
For what follows, it is also important to observe that if $\mathcal{O}$ is hermitian and $L_n$ are hermitian, then all $c_n \in \mathbb{R}$ and such a truncation is also Hermitian.\footnote{This is the simplest choice preserving Hermiticity under the truncation, but other choices are also possible, e.g., if the set $\{L_n\}$ is closed under hermitian conjugation, which is the case for the operator string basis in Eq.~(\ref{eq:Lbasis}).}
We always assume Hermiticity-preserving truncation below.
\end{defin}
We can now state our main Theorem that characterizes the Hamiltonian types in terms of boundary actions.
\begin{thm}
\label{thm:Hermitiancut}
Given a set of states $\{\ket{\psi_n}\}$, an extensive local parent Hamiltonian (Hermitian operator) in one dimension $H = \sum_j{h_{[j]}}$ is type I or type II if and only if the following condition is satisfied for sufficiently large patches $\Lambda$:
\begin{equation}
H_\Lambda\ket{\psi_n} = (T_{\partial \Lambda} +f_{n, \Lambda})\ket{\psi_n}, \quad \forall n ~,
\label{eq:hermitianrelation}
\end{equation}
where $H_\Lambda$ is the Hamiltonian restricted to the patch, and $T_{\partial \Lambda}$ is an operator strictly localized on the boundaries of $\Lambda$, $f_{n, \Lambda}$ is in general an $n$- and $\Lambda$-dependent constant, and the left (right) boundary part is independent of the location of the right (left) boundary for sufficiently large patch $\Lambda$.
Furthermore, $H$ is type I if and only if $T_{\partial \Lambda}$ can be chosen to be Hermitian, and type II if and only if such a $T_{\partial \Lambda}$ exists but it cannot be made Hermitian.
\end{thm}
The proof is given in App.~\ref{app:cuttingproof}.
We will illustrate applications of this theorem with examples from the $W$ state later in this section.
When we are interested in a single state rather than a set of states, we can w.l.o.g. set $f_{n, \Lambda} = 0$, e.g., by absorbing it into $T_{\partial \Lambda}$.
Notice that the support of $T_{\partial \Lambda}$ is strictly localized for type I and II, but it is in general not unique.
The theorem says that if we are able to choose it to be Hermitian, then we are guaranteed that the Hamiltonian is type I. 
On the contrary, the existence of a non-Hermitian boundary action does not by itself mean that the Hamiltonian is type II, and in addition we need to rule out the existence of a Hermitian boundary action.
We depict these differences between truncations of type I, II, and III parent Hamiltonians for a single state in Fig.~\ref{fig:diffhams}.
The following corollary then holds, where two parent Hamiltonians whose truncations are related by a Hermitian boundary term (as far as their action on the target states is concerned) are in the same equivalence class.
\begin{cor}\label{cor:typeII}
    Two extensively local operators $H^A$ and $H^B$ are in the same type II equivalence class if and only if there exist $\alpha,\beta\in \mathbb{R}/0$ such that for sufficiently large patches $\Lambda$
    \begin{align}
        \left(\alpha H^A_\Lambda-\beta H^B_\Lambda\right) \ket{\psi_n} =(T_{\partial\Lambda} + f_{n, \Lambda})\ket{\psi_n},\quad\forall n,
    \end{align}
    for some Hermitian $T_{\partial\Lambda}$ [strictly localized on the boundaries of $\Lambda$, where the left (right) boundary part is independent of the location of the right (left) boundary], and for some constants $f_{n, \Lambda}$.
\end{cor}
This statement follows directly from the definition of the equivalence classes and the truncation procedure and Thm.~\ref{thm:Hermitiancut}.
\begin{figure*}
    \centering
    \includegraphics[width=0.8\linewidth]{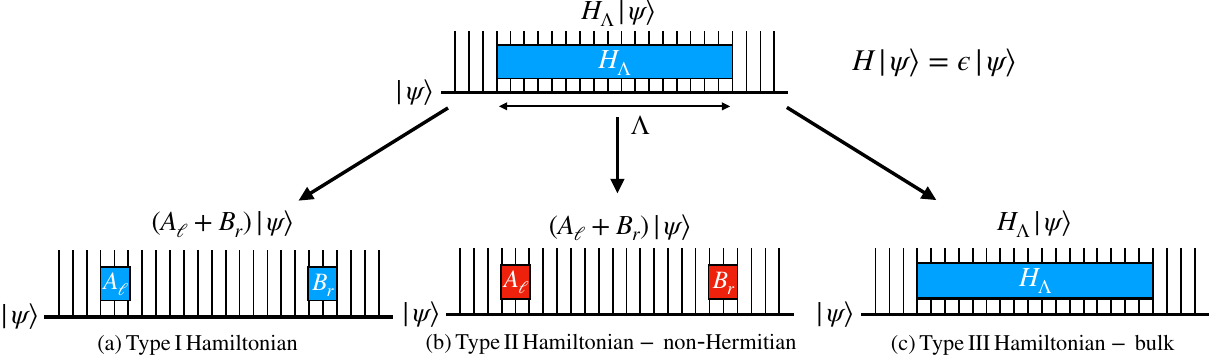}
    \caption{Here we depict three possibilities when we apply a truncated Hamiltonian $H_\Lambda$ (on a contiguous region $\Lambda$) to a state $\ket{\psi}$, with $H\ket{\psi}=\epsilon\ket{\psi}$. (a) Upon truncation, the $H_\Lambda\ket{\psi}$ can be written as boundary terms actions $(A_\ell+B_r)\ket{\psi}$, where $A_\ell$ and $B_r$ are both Hermitian.
    (b) $H_\Lambda\ket{\psi}$ can also result in an action on a state with boundary terms that are non-Hermitian $A_\ell$ and $B_r$, with no Hermitian choices existing. 
    We refer to these Hamiltonians as type II. (c) Type III Hamiltonians are all other cases such as when the truncation cannot be decomposed as any boundary term action. Note that although the depictions resemble a quantum circuit, the pieces are additive instead of multiplicative since we are dealing with a Hamiltonian, not a unitary.}
    \label{fig:diffhams}
\end{figure*}
In the following sections, we will demonstrate how we may use Thm.~\ref{thm:Hermitiancut} to show that the Hamiltonians are not type I, including $H_{\rm ImHop}$ and $\hat{N}_{\rm tot}$, as well as a novel type II Hamiltonian for the $W^2$ state.
\subsection{Examples}
\subsubsection{$H_{\rm ImHop}$ is type II for the $W$ state}
\label{sec:CuttingHImHop}
We will demonstrate the consequences of Thm.~\ref{thm:Hermitiancut} on our $\ket{W}$ state QMBS example, by using it as an alternative proof of Prop.~\ref{prop:HII} where we show that $H_{\rm ImHop}$ is a type II Hamiltonian.

\begin{proof}
A natural restriction of $H_{\rm ImHop}$ to a segment $\Lambda \defn [\ell, \dots, r]$ is
\begin{equation}
\HImHopLambda = \frac{i}{2} \sum_{j=\ell}^{r-1} \left( s_j^\dagger s_{j+1} -s_{j+1}^\dagger s_j \right) ~.
\end{equation}
For example, we obtain this restriction if we use Pauli strings as a basis in the operator space, since $i (s_j^\dagger s_{j+1} - s_{j+1}^\dagger s_j) = (\sigma_j^x \sigma_{j+1}^y - \sigma_j^y \sigma_{j+1}^x)/2$.
It is easy to verify that
\begin{equation}
\HImHopLambda \ktW = \frac{i}{2} (\hat{n}_\ell - \hat{n}_r) \ktW ~,
\label{eq:HImHop_Lambda_ktW}
\end{equation}
where we use short-hand $\hat{n}_j \defn s_j^\dagger s_j = \ketbra{1}{1}_j$.
That is, the action of $\HImHopLambda$ on $\ktW$ can be represented as a boundary action with {\it non-Hermitian} operators near the corresponding boundaries.
However, we will show that it cannot be represented using Hermitian boundary operators.
We will show by contradiction that there are no Hermitian operators $A_\ell$ and $B_r$ acting on finite regions near $\ell$ and $r$ respectively that can realize the boundary action of $\HImHopLambda$ on $\ktW$ in Eq.~(\ref{eq:HImHop_Lambda_ktW}):
\begin{equation}
\frac{i}{2} (\hat{n}_\ell - \hat{n}_r) \ktW = (A_\ell + B_r) \ktW ~.
\label{eq:HImHop_Lambda_ktW_AB}
\end{equation}
Here the ranges of $A_\ell$ and $B_r$ are expected to be bounded by a fixed number $R_{\text{max}}$, while the size of $\Lambda$ is allowed to be large.
Specifically, denoting the supports of $A_\ell$ and $B_r$ as $X_\ell$ and $X_r$ respectively, and assuming w.l.o.g.\ that $\ell \in X_\ell$ and $r \in X_r$, it suffices for $\Lambda$ to be large enough so that $X_\ell$ and $X_r$ are disjoint (e.g., it suffices to have $|r - \ell| \geq 2 R_{\text{max}}$).
Let us assume that such Hermitian $A_\ell$ and $B_r$ exist, and calculate the overlap of both sides of Eq.~(\ref{eq:HImHop_Lambda_ktW_AB}) with $\ktW_{X_\ell} \otimes \ktO_{X_\ell^c}$, where $X_\ell^c$ is the compliment of $X_\ell$.
We obtain\footnote{
A simple way to do this is to consider decomposition $\ktW = c_1 \ktW_{X_\ell} \otimes \ktO_{X_\ell^c} + c_2 \ktO_{X_\ell} \otimes \ktW_{X_\ell^c}$ with non-zero $c_1, c_2$, and use it to show that for any $\mathcal{O}_{X_\ell}$ supported on $X_\ell$ we have $\bra{W}_{X_\ell} \otimes \bra{\bar{0}}_{X_\ell^c} \mathcal{O}_{X_\ell} \ktW = c_1 \bra{W}_{X_\ell} O_{X_\ell} \ktW_{X_\ell}$, while for any $\widetilde{\mathcal{O}}_{X_r}$ supported on $X_r$, since $X_r \subset X_\ell^c$, we have $\bra{W}_{X_\ell} \otimes \bra{\bar{0}}_{X_\ell^c} \widetilde{O}_{X_r} \ktW = c_1 \bra{\bar{0}}_{X_\ell^c} \widetilde{O}_{X_r} \ket{\bar{0}}_{X_\ell^c}$.
Note that the chosen somewhat asymmetric treatment of the left and right boundary regions is not fundamental---it happens to be all that is needed to arrive at the desired contradiction quickly.
}
\begin{equation}
c_1 \frac{i}{2}\bra{W} \hat{n}_\ell \ktW_{X_\ell} = c_1 \bra{W}A_\ell \ktW_{X_\ell} + c_1 \bra{\bar{0}} B_r \ktO_{X_\ell^c} ~.
\end{equation}
with $c_1 = \sqrt{|X_\ell|/N}$, 
where we have used that $n_r \ktO_{X_\ell^c} = 0$ and have represented $\bra{\cdots}_X \bullet \ket{\cdots}_X$ as $\bra{\cdots} \bullet \ket{\cdots}_X$ for compactness.
Now, note that the R.H.S.\ is real due to the assumed Hermiticity of $A_\ell$ and $B_r$, whereas the L.H.S. is purely imaginary since  $\bra{W} \hat{n}_\ell \ktW_{X_\ell} \neq 0$, which is a contradiction.
Thus, the assumption that $\HImHopLambda$ on $\ktW$ has boundary action represented by Hermitian operators is not valid, and hence it is a type II Hamiltonian using Thm.~\ref{thm:Hermitiancut}.
\end{proof}
Interestingly, the above proof of the type II for $H_{\text{ImHop}}$ is purely algebraic, in the sense that it does not use any assumptions about the norms of the terms in the Hamiltonian---it only uses the fact that operator ranges of any presumed type I writing have a fixed bound.
This is unlike our earlier proof evaluating expectation values in the $\ket{W_q}$ state that in addition assumed that the local terms in the type I writing also have bounded norms.
While both arguments are easy to check directly, we do not know if/how the two are related, and the argument using the $\ket{W_q}$ state also proved useful when discussing asymptotic scar states.
\subsubsection{$\hat{N}_{\rm tot}$ is type III for the $W$ state}
\label{sec:NtottypeIII}
Similarly, $\hat{N}_{\rm tot}$ is a type III Hamiltonian, as it cannot even be written as a boundary term upon truncation.
While this has been proven in Cor.~\ref{cor:Ntotnotnonherm}, here we wish to show this via consideration of the boundary action.
\begin{proof}
    Let us prove this by contradiction.
    We take $\Lambda = [\ell, \dots, r]$ and assume that the action of a truncated $\hNtotLambda = \sum_{j=\ell}^r \hat{n}_j$ on the $W$ state can be written as
    \begin{equation}
        \hNtotLambda \ktW = (A_\ell + B_r + f)\ktW \quad,
    \label{eq:Ntotnotboundary}
    \end{equation}
    with $A_\ell$ and $B_r$ being boundary terms (potentially non-Hermitian) that only have support near the left or right boundary, say bounded regions $X_\ell$ and $X_r$, respectively, and we have also allowed general $f\in\mathbb{C}$ for convenience.

    Note that by shifting $A_\ell$ and $B_r$ by appropriate multiples of identity (which does not affect ranges) and redefining $f$, we can w.l.o.g.\ assume $\bra{\bar{0}}_{X_\ell} A_\ell \ket{\bar{0}}_{X_\ell} = \bra{\bar{0}}_{X_r} B_r \ket{\bar{0}}_{X_r} = 0$.
    This means that applying $\bra{\bar{0}}_{\Lambda} \otimes \bra{W}_{\Lambda^c}$ to  both sides of Eq.~(\ref{eq:Ntotnotboundary}) results in $f = 0$, since $\hNtotLambda \ket{\bar{0}}_{\Lambda} = 0$.
    Now assume $X_\ell \cap X_r = \emptyset$ and write $\Lambda = X_\ell \cup X_m \cup X_r$ with $X_m \neq \emptyset$.
    Now applying $\bra{W}_{X_m} \otimes \bra{\bar{0}}_{X_m^c}$ to both sides of Eq.~(\ref{eq:Ntotnotboundary}) and using $\bra{W}_{X_m} \otimes \bra{\bar{0}}_{X_m^c} \hNtotLambda =  \bra{W}_{X_m} \otimes \bra{\bar{0}}_{X_m^c}$, we obtain $f=1$.
    This results in a clear contradiction with $f=0$, which means that a truncated $\hat{N}_{\rm tot}$ cannot be written as any boundary terms, Hermitian or otherwise.
    Per Thm.~\ref{thm:Hermitiancut}, $\hat{N}_{\rm tot}$ is a type III Hamiltonian.
    \end{proof}
\subsubsection{Additional Type II Hamiltonian for the $W^2$ state}
Now we present a new example that shows the power of Thm.~\ref{thm:Hermitiancut} beyond our original scar system. Consider Hamiltonians with $\ket{W^2}$ as an exact eigenstate.
In fact, such a consideration is quite natural as one might be interested in higher particle number states such as $\ket{W^2}$, which already appear as asymptotic scars in the context of $\ket{W}$.
\begin{prop}\label{prop:HImhop2}
    Consider the set of Hamiltonians with the states
    $\{\ket{\bar{0}},\ket{W},\ket{W^2}\}$ as exact eigenstates.
    We know that $H_{\rm ImHop}$ remains a type II Hamiltonian since $\ket{W^2}$ is its eigenstate with eigenvalue zero.
    However, we find that for this enlarged set of states there is at least one other equivalence class of type II Hamiltonians for which $H_{\rm ImHop}^{(2)}$ is a representative member, where
    \begin{align}
        H_{\rm ImHop}^{(2)} :=& \frac{i}{2}\sum_j (s_j^\dag s^\dag_{j+1} s_{j+1}s_{j+2} - s_{j+2}^\dag s_{j+1}^\dag s_{j+1} s_j)\, \nonumber\\
        =& \frac{i}{2}\sum_j \left( \ketbra{110}{011} - \ketbra{011}{110} \right)_{j,j+1,j+2}
        \label{eq:HImHop2}
    \end{align}
    \end{prop}
The proof is given in App.~\ref{app:HImhop2}, where we show that $H_{\rm ImHop}^{(2)}$ is also a type II Hamiltonian by employing the truncation technique used in Sec.~\ref{sec:CuttingHImHop}.
The action of the truncation of $H_{\rm ImHop}^{(2)}$ to a region $\Lambda = [\ell, \dots, r]$ can be written as
\begin{equation}
\HImHopLambda^{(2)} \ket{W^2} = \frac{i}{2} (\hat{n}_\ell \hat{n}_{\ell+1} - \hat{n}_{r-1}\hat{n}_r) \ket{W^2} ~,
\label{eq:HImHop2_Lambda_W2}
\end{equation}
where the boundary term $i (\hat{n}_\ell \hat{n}_{\ell+1} - \hat{n}_{r-1}\hat{n}_r)/2$ is non-Hermitian, for which one can show that there is no Hermitian way of writing this boundary action.
Then, we can also use energetics (but non-algebraic) arguments of the variational state $\ket{W_q}$ to show that $H_{\rm ImHop}$ and $H_{\rm ImHop}^{(2)}$ are in different type II equivalence classes.\footnote{Note that an energetics argument using a twisted $\ket{W^2}$ state, $\ket{W^2_q}:=e^{i q \sum_j j\hat{n}_j}\ket{W^2}$, in analogy to the proof originally given for the type II-ness of $H_{\rm ImHop}$ in App.~\ref{app:HImtypeII}, does not yield a proof of type II-ness of $H_{\rm ImHop}^{(2)}$, since any type I Hamiltonian $H_I$ possesses the energy expectation value
$\bra{W^2_q}H_I\ket{W^2_q}=O(q^2)$, which scales as $O(N^{-2})$, as does $H_{\rm ImHop}^{(2)}$, which yields $\bra{W^2_q}H_{\rm ImHop}^{(2)}\ket{W^2_q}=O(\frac{q}{N})$, and also scales as $O(N^{-2})$.}
%

%
\section{Short Range Entangled Eigenstates/QMBS}
\label{sec:SREQMBS}
We now obtain results on parent Hamiltonians for Short-Range Entangled (SRE) states, which include product states, other states with zero correlation length, and injective Matrix Product States.
\subsection{Zero correlation length states}
An interesting consequence of the proof technique for Thm.~\ref{thm:HW}, i.e., the use of the specific operator-string normal-ordered creation and annihilation basis, is that if target eigenstates consist purely of $\ket{\bar{0}}$ (or, in fact, any product state), then there are no type II or III Hamiltonians.
\begin{prop}
\label{prop:0noHIIHIII}
    The eigenstate set of just $\ket{\psi}_{\rm prod}$, where $\ket{\psi}_{\rm prod}$ is a product state in any number of dimensions, has no type II or III parent Hamiltonians.
\end{prop}
\begin{proof}
We first show this for the $\ket{\bar{0}}$ state\footnote{Exhaustive bond algebra for this QMBS can be generated by nearest-neighbor terms, see Sec.~IIIC1 in~\cite{moudgalya2024symmetries}.}.
Recall the operator expansion of $G$ in Eq.~\eqref{eq:Gexpand} used for the proof of Thm.~\ref{thm:HW}, which we will modify for our proof.
In the case of $\ket{\bar{0}}$, $G$ cannot have $n \geq 1, m = 0$ terms in its expansion, i.e., $c_{j_1,...,j_n}s_{j_1}^\dagger \dots s_{j_n}^\dagger$ terms.
This follows since these terms applied to $\ket{\bar{0}}$ create orthogonal states that must disappear when all terms in $G$ are summed---however, these are fully linearly-independent, which forces $c_{j_1,...,j_n}=0$.
If $G$ is Hermitian, then it also cannot have $n = 0, m \geq 1$ terms in its expansion (otherwise it would have their hermitian conjugates that have already been excluded).
Putting these observations together, we can just group each $n \geq 1, m \geq 1$ operator-string basis term and its hermitian conjugate into $h_X$. As a result, $G$ can be written as a sum of local Hermitian operators that annihilate $\ket{\bar{0}}$, which are of the same range as the needed basis terms.
Any other product state simply requires a local change of basis and/or an increase of the local Hilbert space dimension to analogously create a new normal-ordered creation and annihilation operator-string basis\footnote{
More specifically, for a $d$-dimensional on-site Hilbert space spanned by $\{ \ket{a}, a=0,1,\dots,d-1 \}$, we take on-site operator basis $\{\mathds{1}, \{\ketbra{a}{a}, a \neq 0\}, \{\ketbra{a}{b}, a \neq b\}\}$, which is a generalization of Eqs.~(\ref{eq:opbasis1})-(\ref{eq:opbasis4}) that works for qudits.} for which the same arguments apply.
\end{proof}
More generally, these arguments can be extended to show that any state related to a product state via a strict quantum cellular-automaton (QCA) (i.e., with strictly no exponential tails), has no type II or III parent Hamiltonians.
Strict QCAs are unitary operations on the full Hilbert space that map strictly local (finite-range) operators to strictly local operators within a uniformly bounded neighbourhood of the original operator~\cite{SchumacherWerner2004,FreedmanHastings2020,HaahFidkowskiHastings2022,Farrelly2020Review}.
Given these QCAs, we can establish the following statement.
\begin{prop}
\label{prop:QCA}
    The eigenstate set of just $\ket{\psi}=U_{\rm QCA}\ket{\psi}_{\rm prod}$, where $U_{\rm QCA}$ is a QCA, and $\ket{\psi}_{\rm prod}$ is a product state, has no type II or III parent Hamiltonians.
\end{prop}
Essentially, the set of parent Hamiltonians for $U_{\rm QCA}\ket{\psi}_{\rm prod}$ can be related to the parent Hamiltonians $H_{\rm prod}$ of a single product state $\ket{\psi}_{\rm prod}$ as $U_{\rm QCA} H_{\rm prod} U_{\rm QCA}^\dagger$, which is a set of local Hamiltonians by definition, and it is easy to see that this relation preserves the Hamiltonian types.
A more formal proof is as follows.
\begin{proof}
    Let us define a QCA $\varphi$ as
    \begin{equation}
        \varphi(O_X) =  \widetilde{O}_{\widetilde{X}}:=U_{\rm QCA} O_X U_{\rm QCA}^\dag ~,
    \end{equation}
    such that an operator $O_X$ supported on contiguous sites in $X$ with bounded $|X|$ maps to another operator $\widetilde{O}_{\widetilde{X}}$ with bounded $|\widetilde{X}|=|X|+S\geq |X|$ where $S$ is the spread of the QCA and system-size independent. 
    Since it is just a unitary transformation, QCA $\varphi$ is an invertible map that also obeys the following properties
    \begin{align}
     &\varphi(O_X O_Y)=\varphi(O_X)\varphi(O_Y) \,, \nonumber \\
     &  \varphi(O_X +O_Y) = \varphi(O_X)+\varphi(O_Y) \,, \nonumber\\
&\varphi(O_X^\dag) = [\varphi(O_X)]^\dag \,, \quad \varphi(z\mathds{1}) = z\mathds{1}\,,
    \end{align}
    where $z\in\mathbb{C}$.
    Hence,  given any local operator $O_{X} = \sum c_{j_1, \dots, j_n}^{k_1, \dots, k_m} s_{j_1}^\dagger \dots s_{j_n}^\dagger s_{k_1} \dots s_{k_m}$,  we have that
   \begin{align}
        \widetilde{O}_{\widetilde{X}}=\sum c_{j_1, \dots, j_n}^{k_1, \dots, k_m} \, \tilde{s}_{j_1}^\dagger \dots \tilde{s}_{j_n}^\dagger \tilde{s}_{k_1} \dots \tilde{s}_{k_m}\quad,
    \end{align}
    where $\tilde{s}_{j}:=\varphi(s_{j})$. 
    W.l.o.g., we set $\ket{\psi}_\text{prod}=\ket{\bar{0}}$. Notice that these new $\tilde{s}_{j}$ and $\tilde{s}_{j}^\dag$ operators act in the same way on the state $\ket{\psi}:=U_{\rm QCA}\ket{\bar{0}}$ as $s_j$ and $s^\dagger_j$ on $\ket{\bar{0}}$,  i.e.,  $\tilde{s}_{j} \ket{\psi}=0$,   $\bra{\psi}\tilde{s}^\dag_{j} \ket{\psi}=0$,  $\bra{\psi}\tilde{s}_{k}\tilde{s}^\dag_{j} \ket{\psi}=\delta_{jk}$. 
    Hence the argument for the zero state in Prop.~\ref{prop:0noHIIHIII} can be used here to show that there are no type II or type III Hamiltonians.
\end{proof}
This statement is valid in all dimensions, and the set of states expressible as $U_{\rm QCA}\ket{\psi}_{\rm prod}$ includes many zero correlation length SRE states as well as some zero correlation length long-range entangled (LRE) states such as invertible topological orders, including the Kitaev chain.
The classification of QCAs is well-known in one dimension, and any QCA in one dimension can be expressed as $U_{\rm QCA} = U_{\rm FDQC} T^n$ for bosonic systems~\cite{GrossNesmeVogtsWerner2012}, where $U_{\rm FDQC}$ is a finite-depth quantum circuits (again with no exponential tails), and $T^n$ is a translation operator by $n$ sites.
Hence in one spatial dimension, this shows that any state that can be created by a strict FDQC from a product state cannot have type II or type III Hamiltonians. In one dimensional fermionic systems, QCAs may also involve half-translations~\cite{Po2017RadicalChiralFloquet,PoFidkowskiMorimotoPotterVishwanath2016PRX} such as in the creation of the Kitaev chain ground state from a product state. If follows that the Kitaev chain also has no type II or III parent Hamiltonians.
\subsection{On-site symmetry generators as parent Hamiltonians for Matrix Product States}
We now discuss some general observations about parent Hamiltonian types for short-range entangled states.
While we are not able to derive the full class of type II Hamiltonians for a general MPS $\ket{\psi}$~\footnote{
As reviewed in Sec.~\ref{sec:QMBSgeneral}, exhaustive algebra descriptions in~\cite{moudgalya2023exhaustive} applicable, e.g., to injective MPS, are lacking in such finer locality understanding.}, here we instead focus on MPS that are invariant under a global symmetry of the form $U^\theta = e^{i \theta \sum_j L_j}$, where $L_j$ are on-site symmetry generators.
We then explore the ``type" of the symmetry generator $\sum_j{L_j}$, which, by definition is also a parent Hamiltonian for $\ket{\psi}$. 
In particular we show that for SRE MPS states this cannot be type III, and we obtain conditions for when it cannot be type I.
It is likely that many of the techniques can also be generalized to Hamiltonians composed of commuting terms, but we leave a detailed exploration of such cases to future work.  
We consider a short-range entangled translation-invariant state $\ket{\psi}$ with an MPS form of
\begin{align}
    \ket{\psi}:=\sum_{\{s_j\}}\mathrm{Tr}\left[A^{s_1} \cdots A^{s_N}\right]\ket{s_1 \cdots s_N}\quad,
    \label{eq:MPSform}
\end{align}
where $A^{s}$ are $D\times D$ matrices and $\{s_j\}$ represents the local physical $d$-dimensional Hilbert space degrees of freedom.
Since $\ket{\psi}$ is SRE, the MPS is \textit{injective}~\cite{Perez-Garcia2007} under blocking a finite number of sites ${R_{\rm inj}}$, which says that natural map from the auxiliary Hilbert space $\mathcal{H}^{\otimes 2}_D$ to the physical Hilbert space on $R_{\rm inj}$ sites $\mathcal{H}^{\otimes R_{\rm inj}}_{d}$, given by blocking ${R_{\rm inj}}$ MPS tensors ($A^{s_1} \cdots A^{s_{R_{\rm inj}}}$) is injective, i.e., $\textrm{span}_{\{s_j\}}\{A^{s_1} \cdots A^{s_{R_{\rm inj}}}\}$ is the space of all $D \times D$ matrices.
$R_{\text{inj}}$ is sometimes referred to as the \textit{injectivity length}.
Given this injectivity condition, it is known that if the MPS possess a global on-site unitary symmetry $U^\theta = \otimes_j U^\theta_j$ that leaves the state invariant (i.e., if $U^\theta\ket{\psi}=e^{i\theta \alpha}\ket{\psi}$), then the individual MPS tensors obey the following ``push-through" conditions\footnote{Note that there is usually also a phase factor $e^{i\phi(\theta)}$ that accompanies this equation, but we can  w.l.o.g.\ set it to $1$ by absorbing it into $U_j$ (equivalently, by shifting $L_j$ by an unimportant constant).}
\begin{align}
    \sum_{s_j'}\left[U_{j}^{\theta}\right]_{s_j, s_j'}A^{s_j'}= V(\theta)A^{s_j} V(\theta)^{\dagger},
    \label{eq:pushthrough}
\end{align}
where $V(\theta)$ is a unitary operator that forms a representation of the same symmetry group.
First, we derive a general result on the action of the symmetry operator on the injective MPS as follows.
\begin{prop}
\label{prop:symfrac}
    Short-range entangled translation-invariant states $\ket{\psi}$ that admit an exact injective Matrix Product State form with finite bond dimension and possess a global continuous on-site symmetry $U^\theta=e^{i\theta\sum_j L_j}$, always exhibit a boundary action of the symmetry.
    In particular, the action of the truncated symmetry operator $U_{\Lambda}^\theta$ on a contiguous region $\Lambda = [\ell,\dots,r]$ always obeys
    \begin{equation}
    U_{\Lambda}^\theta\ket{\psi}= e^{i\theta\left(\mathcal{O}_{X_\ell}+\mathcal{\widetilde{O}}_{X_r}\right)}\ket{\psi}\,,
    \label{eq:ULambdaboundary}
    \end{equation}
    where $\mathcal{O}_{X_\ell}$ and $\widetilde{\mathcal{O}}_{X_r}$ are strictly local operators (potentially non-Hermitian) supported on the finite range subregions $X_\ell$ and $X_r$ located on the left and right boundary of $\Lambda$, respectively.
\end{prop}
\begin{proof}
Equation~(\ref{eq:pushthrough}) implies that a truncated symmetry operator $U_{\Lambda}^\theta \defn \otimes_{j\in\Lambda}U_{j}^\theta$ acts as
\begin{align}
    U_{\Lambda}^\theta \ket{\psi}=\sum_{\{s_j\}}&\mathrm{Tr}\left[\cdots A^{s_{\ell-1}}V(\theta)A^{s_{\ell}}\cdots\right. \nonumber \\
    &\left.\cdots A^{s_{r}}V(\theta)^{\dagger}A^{s_{r+1}}\cdots\right]\ket{s_1\cdots s_N}\,.
    \label{eq:Uthetaaction}
\end{align}
Using the injectivity of the MPS after blocking $R_{\text{inj}}$, we can always turn the boundary action of $V(\theta)$ on the MPS into the action of another (generally non-unitary) finite-range operator $W^\theta_{X_\ell}$ on the physical degrees of freedom as\footnote{Note that a similar but slightly different construction was demonstrated in \cite{wang2024generalizedspinhelixstates} for the AKLT state, where they essentially constructed the physical operator whose action on the MPS $A^{s_{\ell-1}} A^{s_{\ell}}$ gives $A^{s_{\ell-1}} V(\theta) A^{s_{\ell}}$.}
\begin{align}
    V(\theta) A^{s_j} \cdots A^{s_{j+R_{\rm inj}-1}}=\sum_{\{s_j'\}}c^{\{s_j\}}_{\{s_j'\}}A^{s'_j} \cdots A^{s'_{j+R_{\rm inj}-1}},
\label{eq:Vthetaaction}
\end{align}
where the above coefficients $\{ c^{\{s_j\}}_{\{s_j'\}} \}$ give the matrix elements of the desired $W^\theta_{X_\ell}$,
\begin{equation}
    c^{\{s_j\}}_{\{s_j'\}} = \bra{s_j, \dots, s_{j+R_{\rm inj}-1}}W^\theta_{X_\ell}\ket{s_j',\dots,s_{j+R_{\rm inj}-1}'}.
\end{equation}
We can similarly obtain the operator $\widetilde{W}^\theta_{X_r}$ by studying the action of $V(\theta)^\dagger$ on the auxiliary degrees of freedom.
The operators $W^\theta_{X_\ell}$ and $\widetilde{W}^\theta_{X_r}$ are also invertible, which can be shown by considering the expansion of $V(\theta)$ in $\theta$ and repeating the analysis.
Hence from Eq.~(\ref{eq:Uthetaaction}) we obtain
\begin{align}
    U_{\Lambda}^\theta \ket{\psi}= W^\theta_{X_\ell} \widetilde{W}^\theta_{X_r} \ket{\psi} = e^{i\theta\left(\mathcal{O}_{X_\ell}+\mathcal{\widetilde{O}}_{X_r}\right)}\ket{\psi}\quad,
\label{eq:Uthetafinal}
\end{align}
where $\mathcal{O}_{X_\ell}$ and $\mathcal{\widetilde{O}}_{X_r}$ are strictly local (generically non-Hermitian) operators with support on at most $R_{\text{inj}}$ sites on the left and right boundaries of the region $\Lambda$.
\end{proof}
Proposition~\ref{prop:symfrac}, along with Thm.~\ref{thm:Hermitiancut}, implies that the symmetry generator $\sum_j L_j$ necessarily has to be a type I or type II Hamiltonian, since its action on $\ket{\psi}$ reduces to a boundary action.
Moreover, it also implies the following corollary. 
\begin{cor}
For any translation-invariant $\ket{\psi}$, if there is an on-site symmetry operator that is not type I or type II, i.e., that is type III, then the state $\ket{\psi}$ is necessarily long-range entangled.
\end{cor}
In particular, this means that type III Hamiltonians for a specific eigenstate can exist only if the state is long-range entangled.
From Sec.~\ref{sec:NtottypeIII}, the action of $e^{i\theta\hat{N}_{\rm tot}}$ on $\ket{W}$ defies symmetry boundary action, such that we conclude that it must be long-range entangled, a known fact that follows from other arguments such as being related to a non-zero momentum state via a finite-depth quantum circuit~\cite{PhysRevX.12.031007}.
Nevertheless, it demonstrates an interesting connection to type III Hamiltonians, which we summarize in the following corollary.
\begin{cor}
    Since $\hat{N}_{\rm tot}$ is a type III parent Hamiltonian for the $\ket{W}$ state, $\ket{W}$ is necessarily a long-range entangled state.
\end{cor}
This follows directly from Prop.~\ref{prop:symfrac} and Sec.~\ref{sec:NtottypeIII}.
We now formulate a sufficient condition for the symmetry generator, in the same setting as in Prop.~\ref{prop:symfrac}, to be type II parent Hamiltonian:
\begin{thm}\label{thm:transfermat}
    Consider an injective MPS $\ket{\psi}$ that is globally symmetric under an on-site unitary symmetry $U^\theta = e^{i\theta \sum_j{L_j}}$, i.e., $U^\theta\ket{\psi} = e^{i\phi}\ket{\psi}$ but not locally symmetric, i.e., $U^\theta_j\ket{\psi} \neq e^{i\phi_j} \ket{\psi}$.
    Then the symmetry generator $\sum_j{L_j}$ is type II if the transfer matrix of the MPS, defined as $E = \sum_{s}{A^s \otimes (A^s)^\ast}$ is full-rank.
\end{thm}
The proof is given in Appendix~\ref{app:transfermatrix}. 
We will demonstrate implications of this theorem for the AKLT state in the next section.
Note that this does not rule out the existence of exponentially localized Hermitian boundary terms, as we discuss in Sec.~\ref{subsubsec:explocherm}.
\subsubsection{Example: The AKLT Ground State}
\label{sec:AKLT}
As an example of Thm.~\ref{thm:transfermat} in action, we consider the AKLT ground states~\cite{AffleckKennedyLiebTasaki1987,AffleckKennedyLiebTasaki1988}, which, for PBC, can be written in terms of Matrix Product States of the form of Eq.~(\ref{eq:MPSform}) with $s_j\in\{-,0,+\}$, and 
\begin{equation}
    A^{\pm}=\pm\sqrt{\frac{2}{3}}\sigma^{\pm}, \;\;\textrm{and}\;\; A^0=-\frac{1}{\sqrt{3}}\sigma^z,
\label{eq:AKLTMPS}
\end{equation} 
with $\sigma$ being the standard Pauli matrices.
This is known to be injective with injectivity length $R_{\rm inj} = 2$, and is symmetric under all $S^\alpha_{\tot} = \sum_j{S^\alpha_j}$ where $\alpha \in \{x, y, z\}$.
Given these properties, Prop.~\ref{prop:symfrac} shows that these symmetry operators $S^\alpha_{\tot}$ are not type III parent Hamiltonians.
This was also demonstrated in \cite{wang2024generalizedspinhelixstates}, where they expressed $S^z_{\tot}$ as a sum of non-Hermitian strictly local operators that annihilate the AKLT ground state. 
We can also derive the non-Hermitian boundary operators explicitly following the steps from Eq.~(\ref{eq:Vthetaaction}) to (\ref{eq:Uthetafinal}) and using the form of the AKLT MPS in Eq.~(\ref{eq:AKLTMPS}).
For example, for $S^z_{\tot}$, we have $U^\theta_j = e^{i \theta S^z_j}$, $V(\theta)  = e^{i \frac{\theta}{2}\sigma^z_j}$ and $\phi(\theta) = 0$, which gives us the boundary operators of Eq.~(\ref{eq:Uthetafinal}) as $ W^\theta_{X_\ell} =e^{i\theta \mathcal{O}_{\ell,\ell+1}}$ and $\widetilde{W}^\theta_{X_r}=e^{i\theta\mathcal{\widetilde{O}}_{r-1,r}}$, with
\begin{align*}
    \mathcal{O}_{\ell,\ell+1}
    &=\frac{1}{2}\bigg[\ketbra{+-}+\ketbra{00}+\left(\ket{00}-\ket{-+}\right)\bra{-+}\\
    &+\ketbra{0+}-\ketbra{0-}+\ketbra{+0}-\ketbra{-0}\bigg]\,,\\
    \mathcal{\widetilde{O}}_{r-1,r}
    &=\frac{1}{2}\bigg[\ketbra{+-}+\ketbra{00}+\left(\ket{00}-\ket{-+}\right)\bra{-+}\\
    &-\ketbra{0+}+\ketbra{0-}-\ketbra{+0}+\ketbra{-0}\bigg],
\end{align*}
which are not Hermitian.
Further, we can check that the MPS tensor transforms non-trivially under the symmetry $U^\theta_j$, since $V(\theta) \neq e^{i\theta\alpha}\mathds{1}$.
The transfer matrix of the AKLT MPS of Eq.~(\ref{eq:AKLTMPS}) can also be computed and is full rank: its eigenvalues are $1$ and three-fold degenerate $-1/3$.
Theorem~\ref{thm:transfermat} then says that $S^z_{\tot}$ is type II, since there are no finite-range Hermitian boundary operators that would reproduce the action of $S^z_{\tot}$ in a patch of the system.
This is consistent with earlier results in~\cite{moudgalya2023exhaustive} that showed that it is not type I using different arguments.
A similar proof also works for $S^x_{\rm tot}$ and $S^y_{\rm tot}$, and we can also argue that these are in different type II equivalence classes by applying Thm.~\ref{thm:transfermat} to show that $a S^\mu_{\tot} - b S^{\nu}_{\tot}$ for $\mu \neq \nu$ and $a, b \in \mathbb{R}$ is type II.
Next, we will explore a closely related example, known as the bosonic Su–Schrieffer–Heeger (SSH) chain, and its relationship with $S_{\rm tot}^z$. We will show that in this case the Hilbert space is an enlarged version of the AKLT Hilbert space which results in $S_{\rm tot}^z$ being a type I Hamiltonian.
\subsubsection{Example: Bosonic Su–Schrieffer–Heeger (SSH) model}
\label{sec:exZ2SPT}
In contrast to the AKLT example, consider another simple example on a 1d chain of length $N$ with periodic boundary conditions, and a local Hilbert space of $\mathbb{C}^2\otimes \mathbb{C}^2$, i.e., two qubits (labeled $\alpha$ and $\beta$) per unit cell (labeled $j$).
The bosonic SSH ground state can be written as
\begin{align}
    \ket{\Psi_0}&=\otimes_{j=1}^N\frac{1}{\sqrt{2}}\left(\ket{\uparrow}_{\beta j}\ket{\downarrow}_{\alpha j+1}-\ket{\downarrow}_{\beta j}\ket{\uparrow}_{\alpha j+1}\right)\,,
\end{align}
with a translation-invariant PBC MPS of the form of Eq.~(\ref{eq:MPSform}) with
\begin{equation}
    A=\frac{1}{\sqrt{2}}
    \begin{pmatrix}
        \ket{\downarrow\uparrow} &\ket{\downarrow\downarrow} \\
        -\ket{\uparrow\uparrow}  & -\ket{\uparrow\downarrow}
    \end{pmatrix},
\end{equation}
which has a physical dimension of $d = 4$, and we have used a compact notation for the tensor $A$, i.e., $A^s = \bra{s}A$ for $s \in \{\uparrow\uparrow, \uparrow\downarrow, \downarrow\uparrow, \downarrow\downarrow\}$.
Such a state is symmetric under $S^z_{\rm tot}=\sum_{j=1}^N (\sigma^z_{\alpha j}+\sigma^z_{\beta j})$.
In this case, the transfer matrix has eigenvalues $1$ and a three-fold degenerate $0$, showing that it is not full-rank and therefore not subject to Thm.~\ref{thm:transfermat}.
In fact, once we apply the truncated symmetry operator $U_\Lambda^\theta=\otimes_{j\in\Lambda}e^{i\theta S^z_{j}}$ on a contiguous region $\Lambda = [\ell, \dots, r]$, it is clear that 
\begin{align}
    U_\Lambda^\theta\ket{\Psi_0}=e^{i\theta\left( \sigma^z_{\alpha\ell}+\sigma^z_{\beta r}\right)}\ket{\Psi_0}\quad.
\end{align}
We can then see that the boundary action of $S^z_{\tot}$ is fully Hermitian, and from Thm.~\ref{thm:Hermitiancut}, it is a type I operator.
This is also very easily verifiable by the following regrouping: $S^z_{\text{tot}} = \sum_{j=1}^N (\sigma^z_{\beta j} + \sigma^z_{\alpha,j+1})$.
Such an observation ties back to Prop.~\ref{prop:QCA}, where we demonstrated that states connected to a product state via a QCA do not possess type II or III Hamiltonians.
In this case, $\ket{\Psi_0}$ is related to a product state via a depth-2 quantum circuit, such that it should come as no surprise that $S_{\rm tot}^z$ has fully Hermitian boundary action, as there are no type II or III Hamiltonians.
On a different note, we also observe that the AKLT ground state can be related to the SSH state by enlarging the spin-1 Hilbert space to that of two spin-$\frac{1}{2}$ per site.
This can cause $S^z_{\tot}$, a type II operator in the AKLT, to become a type I operator, as in the case of the bosonic SSH chain.
This observation, that a type II operator can be ``converted" to a type I operator by enlarging the Hilbert space, was discussed in Ref.~\cite{PhysRevB.108.054412}.
Note, however, that the difference between the AKLT and SSH chains is not simply the enlargement of the Hilbert space:
There is a non-trivial projection involved to go from the latter to the former, so the wavefunctions are different and the types of parent Hamiltonians can be different.
\subsubsection{Exponentially localized Hermitian boundary terms}\label{subsubsec:explocherm}
We finally remark that Thm.~\ref{thm:transfermat} only rules out the existence of strictly local Hermitian terms on the boundary.
In fact, we argue that on any symmetric SRE MPS it is always possible to write the actions of a truncated symmetry operator using localized Hermitian boundary terms with exponentially decaying tails. 
To see this, note that any symmetric SRE state can be obtained from any symmetric product state via a \textit{globally} symmetric unitary transformation of logarithmic depth~\cite{malz2024preparation}.\footnote{This is not always possible with unitaries composed of locally symmetric gates due to Symmetry Protected Topological (SPT) phases, but here we only require globally symmetric unitaries.}
Such unitaries in general transform finite-range operators to localized operators with exponentially decaying tails while keeping the global symmetry operator invariant. 
Since product states only admit type I parent Hamiltonians (Prop.~\ref{prop:0noHIIHIII}), the action of the global symmetry operator should be representable by a strictly local Hermitian boundary operator (Thm.~\ref{thm:Hermitiancut}). 
Hence applying the unitary that transforms the product state to the SRE state, the Hermitian boundary operator should transform to a Hermitian operator with exponentially decaying tails.
It is interesting to note that this means that the action of a strictly local non-Hermitian operator (which always exists according to Prop.~\ref{prop:symfrac}) is equivalent to the action of an exponentially localized Hermitian operator. 
In the context of QMBS, this equivalence is reminiscent of results on the approximate QMBS in the PXP model, where many kinds of models with close-to-perfect QMBS have been studied. 
These include both non-Hermitian models~\cite{Omiya_2023quantum} and models with exponentially decaying terms~\cite{Choi_2019}, hence it would be interesting in the future to explore if they are ``equivalent" in some similar sense.
%

\section{Discussion and Outlook}
\label{sec:discussion}
In this paper, we explored the different kinds of parent Hamiltonians for quantum states, i.e., local Hamiltonians that have a particular state as an eigenstate.
Motivated by observations in previous works on Quantum Many-Body Scars, which is also typical context for the states we consider, we proposed a three-fold classification of the set of parent Hamiltonians, depending on how an extensive local Hamiltonian can be decomposed into sums of strictly local terms having the same set of eigenstates (type I or II if possible, and type III if not), and whether those strictly local terms are Hermitian (type I if possible, and type II if not). 
This refines and extends the classification of Hamiltonians previously proposed in \cite{moudgalya2023exhaustive}, particularly making it more broadly applicable to general quantum states and removing extraneous constraints that came from algebraic definitions used there.
To clearly illustrate this classification, we rigorously obtained the full class of geometrically-local parent Hamiltonians that possess the well-known $W$ state as an eigenstate.
From this result, we showed that several consequences follow immediately, such as if $\ket{W}$ is an eigenstate of such a Hamiltonian, then so is the vacuum state $\ket{\bar{0}}$, which generalizes the statement shown in~\cite{gioia2024wstateuniqueground} that if $\ket{W}$ is a ground state of a Hamiltonian with system-size independent chemical potential, then so is $\ket{\bar{0}}$.
Our result also showed three classes of parent Hamiltonians for the $\ket{W}$ state, type I, type II, and type III, clearly illustrating the classification we proposed.
Finally, on a separate note, we derived several different kinds of long-lived asymptotic scar states that exist for the whole family of Hamiltonian models that have the $W$ state as an exact scar. 
Using the intuition gained from the $W$ state, we derived a theorem (Thm.~\ref{thm:Hermitiancut}) that obtains a relation between the classification of the Hamiltonian types and the action of a truncated Hamiltonian on a patch of the system. 
We illustrated this theorem by demonstrating it on the different types of Hamiltonians for the $\ket{W}$ state, and also by proving examples of type II Hamiltonians for the $\ket{W^2}$, a closely related state.
Finally, we used these results to also derive some statements on the types of parent Hamiltonians for some Short Range Entangled (SRE) states.
We first focused on product states and those obtained by applying Quantum Cellular Automata (QCA) on product states, where we ruled out the existence of parent Hamiltonians of types II and III.
Then we studied the type of symmetry generators (which can be also viewed as parent Hamiltonians) in Matrix Product States invariant under a continuous on-site symmetry, where we ruled out the existence of type III symmetry generators (which led to an alternate proof that $\ket{W}$ is long-range entangled), and provided conditions for ruling out type I symmetry generators (which led to an alternate proof that the total spin operators are type II parent Hamiltonians for the AKLT ground state).
Our work opens up many interesting directions for future studies. 
Perhaps the most immediate is the deeper exploration of parent Hamiltonians of Matrix Product States (MPS), which are likely to yield richer type II and potentially type III parent Hamiltonians, beyond the standard type I parent Hamiltonians that have been ubiquitously the focus of earlier studies~\cite{Perez-Garcia2007, PerezGarciaVerstraeteCiracWolf2008}.
Of course, earlier works were primarily concerned with Hamiltonians with these MPS as \textit{ground states}, which might rule out some of these types, but that is an interesting question to carefully address in the light of our classification.
We have initiated this study by restricting to analysis of symmetry generators for symmetric MPS, but we numerically observe that there are many more classes of type II Hamiltonians when one goes beyond these symmetry operators, and it would be nice to develop a systematic understanding of these cases. 
Further, many examples of QMBS in the literature share some essential features of the $\ket{\bar{0}}$ and $\ket{W}$ states along with the closely related $\ket{W^p}$ states, i.e., they typically consist of a single SRE state and ``quasiparticle excitations" on top of that state~\cite{Moudgalya_2022, Chandran_2023}.
Hence our derivation of the complete class of parent Hamiltonians for these states, arguably the simplest non-trivial examples of QMBS, highlights important concepts that naturally appear in the study of Hamiltonians with QMBS, and paves the way for a systematic classification of such Hamiltonians. 
In the future, it would be interesting to perform similar studies on QMBS with more structure such as entire towers of states, with the simplest example being the ferromagnetic tower of states~\cite{Choi_2019, Mark2020Eta, moudgalya2023exhaustive}, which are simply Dicke states~\cite{Dicke1954Coherence,Bastin2009OperationalFamilies}, essentially like $\ket{W^p}$ with $p$ that can take values from $0$ to $N$.
We have initiated this more general exploration by examining the two-particle $\ket{W^2}$ state and some of its type II parent Hamiltonians (including some initial results for $\ket{W^p}$), but it would be important to perform a more thorough classification.
More general states with these structures have also been referred to as regular language states~\cite{florido2024regular}, and our work can also be viewed as a stepping stone towards understanding classes of parent Hamiltonians for such states. 
There are also numerous examples of QMBS that consist of states that combine the two concepts of MPS and towers of states, i.e., they are quasiparticle towers that are built on top of MPS states~\cite{Moudgalya2020Large}, such as the tower of QMBS in the AKLT Hamiltonian~\cite{Moudgalya_2018exact, Moudgalya_2018entanglement}.
These are likely to exhibit an even richer variety of parent Hamiltonians~\cite{mark2020unified,ren2021deformed, moudgalya2023exhaustive, wang2024generalizedspinhelixstates, zhang2025quantum}. 
Parent Hamiltonians in many examples of these towers of states have been conjectured in \cite{moudgalya2023exhaustive} to have certain specific structures, e.g., those that necessitate \textit{equal spacings} of the towers in the spectrum due to locality of the parent Hamiltonians, which lead to interesting dynamical consequences~\cite{alhambra2020revivals, Moudgalya_2022, odea2025entanglement}.
A general understanding of parent Hamiltonians can lead to proofs of such conjectures. 
Finally, there are other kinds of exotic quantum states for which understanding parent Hamiltonians and locality constraints would be interesting, e.g., exotic QMBS such as those with volume-law entanglement~\cite{Langlett_2022,ivanov2024volume,chiba2024exact, mohapatra2024exact}.
Particularly, it would be interesting to explore if the classification scheme we propose is appropriate for such examples, or if there are finer structures that emerge there. 
Likewise, it would be interesting to generalize our results to higher-dimensions, where finer classifications of parent Hamiltonians could naturally emerge.
Beyond local Hamiltonians, it would also be interesting to generalize notions of parent Hamiltonians to few-body non-local, or long-range parent Hamiltonians that might be of interest in many contexts such as QMBS in long-range systems~\cite{lerose2025theory}, or dark states in central spin models~\cite{villazon2020persistent}.
For the $\ket{W}$ state, some of these questions can be immediately answered from our analysis (since the $W$ state is permutation symmetric and therefore oblivious to any notion of spatial locality), but there might be finer general classifications that might be appropriate in such settings. 
Beyond Hamiltonian systems, it would be interesting to generalize parent Hamiltonians to periodically driven systems where there could be interesting classes of parent Floquet operators.
Several QMBS have been known to be intrinsic to Floquet systems~\cite{sugiura2021manybody, Mizuta2020Exact, ljubotina2024tangent}, and there have been many correspondences established between QMBS Hamiltonians and Floquet unitaries~\cite{Rozon_2022, rozon2023broken}, which should provide inspiration for such generalizations. 
Beyond isolated quantum systems, analogs of such QMBS states have also been found in open quantum systems~\cite{buca2019nonstationary, tindall2020quantum, wang2024embedding, marche2025exceptional, garcia2025lindblad, gotta2025open}, which also motivates generalization to non-Hermitian settings. 
A universal understanding of Hermitian and non-Hermitian terms that possess certain states as eigenstates can also lead to routes to dissipatively engineer such states~\cite{wang2024embedding}. 
On a different note, approximate QMBS of the PXP model~\cite{Turner_2018weak, Turner_2018quantum, Iadecola_2019, Lin_2019, Surace_2020, Serbyn_2021, Desaules2022SchwingerScars, Desaules2023weak, Su2023observation, Daniel2023bridging, Ivanov_2025exact, Mark2025observation} have been conjectured to be related to non-Hermitian models of exact QMBS~\cite{Omiya_2023quantum, Omiya_2023fractionalization}, hence the understanding of non-Hermitian parent Hamiltonians might shed light in that direction as well.
Finally, we presented some dynamical differences between the non-interacting examples of type I and type II Hamiltonians in the context of the $\ket{W}$ state, such as the overall dispersion in free-particle models (and the related dispersion of energy of the $\ket{W_q}$ asymptotic scars in general interacting models).
We further showed the dynamical differences between these two types by studying a droplet of the $W$ state and how it evolves in a diffusive boundary melting manner versus ballistic propagation (with melting on top), respectively.
Although we argue that such dynamical features should be stable even in the presence of interactions, a full numerical calculation is necessary to test this in future work.
More generally, it is important to establish more settings in which qualitative differences between the behaviors of type I, II, and III Hamiltonians can exist, which would validate this classification as being a useful way of understanding ``universal" behaviors of Hamiltonians with QMBS. 
\textit{Note added --} While finalizing this manuscript, we became aware of two independent results on asymptotic QMBS of the $\ket{W}$ state~\cite{gotta2025asymptotic, dooley2025asymptotic}, which overlap with some results in Sec.~\ref{sec:asymptoticscars}.
%

\section*{Acknowledgements}
We especially thank Lorenzo Gotta for sharing results and insightful discussions on overlapping work on asymptotic QMBS of the $W$ state.
We also thank Andrew Ivanov, Daniel Mark, Daniel Ranard, Federica Surace, Ryan Thorngren, and Ruben Verresen for helpful discussions, and Heran Wang for pointing out Ref.~\cite{wang2024generalizedspinhelixstates}.
This work was supported by the Walter Burke Institute for Theoretical Physics at Caltech.
We acknowledge support provided by the Institute for Quantum Information and Matter, an NSF Physics Frontiers Center (NSF Grant PHY-2317110); by the Munich Center for Quantum Science and Technology (MCQST) and the Deutsche Forschungsgemeinschaft (DFG, German Research Foundation) under Germany’s Excellence Strategy--EXC--2111--390814868; and by the National Science Foundation through grant DMR-2001186.
%

\bibliography{main}

\clearpage         
\onecolumngrid  

\appendix
%

%
\section{Numerical methods to find type II and type III operators}\label{app:numerical}
In this Appendix, we discuss some numerical methods for finding type II and type III operators, and their equivalence classes, based on extensions of ideas in \cite{moudgalya2023numerical}. 
Given a state $\ket{\psi}$ (or a set of states $\{\ket{\psi_n}\}$), there are systematic methods for searching for an operator in a desired set (vector space $\mV$) of operators that have these states as eigenstates~\cite{qiranard2017, chertkovclark2018, greiter2018parent, yang2023detecting, yao2022bounding, moudgalya2023numerical}.
Given an orthogonal basis of operators $\{V_\mu\}$ for the vector space $\mV$ we wish our operator to lie in, we can construct the matrices $C^{\rm H}$~\cite{qiranard2017} and $C^{\rm G}$~\cite{chertkovclark2018} as
\begin{align}
    C^{\rm H}_{\mu\nu} &\defn  \frac{1}{2}\bra{\psi} \{V_\mu, V_\nu\} \ket{\psi} - \bra{\psi} V_\mu \ket{\psi} \bra{\psi} V_\nu\ket{\psi} ~, \nn \\
    C^{\rm G}_{\mu\nu} &\defn \bra{\psi}V_\mu V_\nu \ket{\psi} - \bra{\psi} V_\mu \ket{\psi}\bra{\psi}V_\nu \ket{\psi} ~. 
\end{align}
These can be shown to be positive semi-definite, and their eigenvectors with zero eigenvalue are precisely the operators in $\mV$ that have $\ket{\psi}$ as an eigenstate.
The zeros of $C^{\rm H}$ are only Hermitian operators with that property, which span the vector space $\mZ^{\rm H}$, whereas the zeros of $C^{\rm G}$ are all (Hermitian and non-Hermitian) operators with that property, which span the vector space $\mZ^{\rm G}$.
While this is for a single state $\ket{\psi}$, analogous results can be obtained for a set of states $\{\ket{\psi_n}\}$ by simply considering the zeros of the sum of these matrices for the individual states. 
Below, we will assume that we use these methods to obtain the desired set of Hermitian or general operators. 
First, we focus on the vector space  $\mV_{[j, j+R-1]}$ of range-$R$ strictly local operators that have support on sites $[j, j+R-1]$.
Using the methods described above, we can obtain the vector spaces of Hermitian and general non-Hermitian operators that have support on those sites and those that have $\ket{\psi}$ as an eigenstate; we denote these as $\mZ^{\rm H}_{[j, j+R-1]}$ and $\mZ^{\rm G}_{[j, j+R-1]}$.
Using these, we numerically construct the vector spaces $\mZ^{\rm H}_{\loc,R}$ and $\mZ^{\rm G}_{{\rm loc},R}$ of all Hermitian and general operators that consist of all possible strictly local operators that have range $R$, these are given by
\begin{equation}
    \mZ^{\rm H}_{{\rm loc},R} = \bigcup_j\mZ^{\rm H}_{[j, j+R-1]},\;\;\;\mZ^{\rm G}_{{\rm loc}, R} = \bigcup_j\mZ^{\rm G}_{[j, j+R-1]}.
\end{equation}
Separately, we consider the vector space $\mV^{R}$ of all range $R$ operators, including extensive-local operators, which is given by
\begin{equation}
    \mV_{R'} = \bigcup_j{\mV_{[j, j+R-1]}}.
\end{equation}
Using the methods above we find the Hermitian and general non-Hermitian operators in $\mV_{R}$ that have $\ket{\psi}$ as an eigenstate, and we denote those vector spaces by $\mZ^{\rm H}_{{\rm glo}, R}$ and $\mZ^{\rm G}_{{\rm glo}, R}$ respectively.
Now, we can choose $R' \geq R$, and carefully define the equivalence classes of type II and type III operators as follows (note that previous works discussing finding type II equivalence classes~\cite{moudgalya2023exhaustive, moudgalya2023numerical} implicitly set $R' = R$)
\begin{enumerate}
    \item An extensive-local operator of range $R$ with an eigenstate $\ket{\psi}$ is type III if it cannot be written as a sum of strictly local range $R'$ (Hermitian or non-Hermitian) terms that have $\ket{\psi}$ as an eigenstate. We can then define equivalence classes of such operators (equivalent up to the addition of type II operators of range $R'$).
    The inequivalent type III operators are those that lie in the vector space $\mZ^{\rm H}_{{\rm glo}, R} - (\mZ^{\rm H}_{{\rm glo}, R} \cap \mZ^{\rm G}_{{\rm loc}, R'})$, hence the number of type III equivalence classes is 
    \begin{equation}
        N^{\rm III}_{R, R'} = \text{dim}\ \mZ^{\rm H}_{{\rm glo}, R} - \text{dim}\ (\mZ^{\rm H}_{{\rm glo}, R} \cap \mZ^{\rm G}_{{\rm loc}, R'}) = \text{dim}\ (\mZ^{\rm H}_{{\rm glo}, R} \cup \mZ^{\rm G}_{{\rm loc}, R'}) - \text{dim}\ \mZ^{\rm G}_{{\rm loc}, R'}.
    \end{equation}
    \item An extensive-local operator of range $R$ with an eigenstate $\ket{\psi}$ is type II if it cannot be written as a sum of strictly local range $R'$ Hermitian terms that have $\ket{\psi}$ as an eigenstate, but can be written as a sum of strictly local range-$R'$ non-Hermitian terms that have $\ket{\psi}$ as an eigenstate. 
    We can then define equivalence classes of such operators (equivalent up to the addition of type I operators of range $R'$).
    Note that the inequivalent type II and type III operators lie in the vector space $\mZ^{\rm H}_{{\rm glo}, R} - (\mZ^{\rm H}_{{\rm glo}, R} \cap \mZ^{\rm H}_{{\rm loc}, R'})$, hence the total number of type II and type III equivalence classes is given by
    \begin{gather}
        N^{\rm II}_{R, R'} + N^{\rm III}_{R, R'}  = \text{dim}\ \mZ^{\rm H}_{{\rm glo}, R} - \text{dim}\ (\mZ^{\rm H}_{{\rm glo}, R} \cap \mZ^{\rm H}_{{\rm loc}, R'}) = \text{dim}\ (\mZ^{\rm H}_{{\rm glo}, R} \cup \mZ^{\rm H}_{{\rm loc}, {R'}}) - \text{dim}\ \mZ^{\rm H}_{{\rm loc}, R'}\nn \\
        \implies\;\;N^{\rm II}_{R,R'} = \text{dim}\ (\mZ^{\rm H}_{{\rm glo}, R} \cup \mZ^{\rm H}_{{\rm loc}, {R'}}) - \text{dim}\ (\mZ^{\rm H}_{{\rm glo}, R} \cup \mZ^{\rm G}_{{\rm loc}, R'}) + \text{dim}\ \mZ^{\rm G}_{{\rm loc}, R'} - \text{dim}\ \mZ^{\rm H}_{{\rm loc}, R'}.
    \end{gather} 
\end{enumerate}
Note that in the main text, we have assumed $R'$ to be any finite but arbitrary large number in the thermodynamic limit $N \rightarrow \infty$. 
The underlying belief is that $N^{\rm II}_{R, R'}$ and $N^{\rm III}_{R, R'}$ are \textit{independent} of $R'$ for some $R' \geq R_{\rm c}$ for some finite $R_{\rm c}$ that might depend on $R$.
This is analytically provable for the $W$ state as discussed in Sec.~\ref{sec:WstateQMBS}, but in general we need to rely on numerical results.
However, since the time complexity of these methods scale polynomially in $N$ and exponentially in $R$ and $R'$, probing the desired limit is often infeasible in practice.
It would hence be important to develop better methods for this purpose in the future, perhaps by imposing additional conditions such as translation invariance.
%

%
\section{Bond and Commutant Algebras for $\ket{\bar{0}}$ and $\ket{W}$ as QMBS}\label{app:algebras}
Since for extensive-local Hamiltonians $\ket{\bar{0}}$ is always an eigenstate if $\ket{W}$ is an eigenstate, this motivates us to consider them as QMBS, and examine the corresponding bond and commutant algebras introduced in Sec.~\ref{subsec:qmbsconn}. 
First, we can consider the commutant [analogous to Eq.~(\ref{eq:degeneratescar})]
\begin{align}
    \mathcal{C}^{\rm deg}_{\ket{\bar{0}},\ket{W}}
    = \lgen\{ \ketbra{\bar{0}},
    \ketbra{W}{\bar{0}},
    \ketbra{\bar{0}}{W},
    \ketbra{W} \}\rgen\quad,
    \label{eq:comdeg}
\end{align}
whose associated bond-algebra corresponds to all Hamiltonians with a degenerate QMBS system of $\ket{\bar{0}}$ and $\ket{W}$.
We can analytically derive the form of the exhaustive bond algebra via an exhaustive extension of the so-called Shiraishi-Mori construction~\cite{Shiraishi_2017,moudgalya2023exhaustive}.
\begin{prop}
\label{lem:bondalg}
    The exhaustive bond algebra of $\mathcal{C}^{\rm deg}_{\ket{\bar{0}},\ket{W}}$ on a length $N$ system is given by
\begin{equation}
\mA^{\text{deg}}_{\ket{\bar{0}},\ket{W}} = \mA_{1,2,\dots,N}^\text{SMobc} \defn \lgen \{ \ketbra{S}{T}_{j,j+1}, 1 \leq j \leq N-1 \} \rgen ~, 
\label{eq:2sitealgebraOBC}
\end{equation}
where $\ket{T}_{j,k} \defn \ket{11}_{j,k}$ and $\ket{S}_{j,k} \defn (\ket{10} - \ket{01})_{j,k}$.
\end{prop}
Note that the generators are chosen as if they are nearest-neighbor terms on an OBC chain, hence the label ``obc",  but the argument works for any system on any graph with such a path.
We have used the same convention as in~\cite{moudgalya2023exhaustive}, where Hermitian conjugates of all exhibited generators as well as the identity operator are also implicitly included, hence the algebra is $\dagger$-algebra,  e.g.,  it immediately follows that $\ketbra{T}{S}_{j,j+1}$ and $\ketbra{T}{T}_{j,j+1}$ on the same bonds are also in the bond algebra.
The generators can also be written in the hard-core boson creation/annihilation language as $\ketbra{S}{T}_{j,j+1} = (s_j^\dag - s_{j+1}^\dag) s_j s_{j+1}$  and their Hermitian conjugates.
``SM" in the label stands for the form being motivated by the Shiraishi-Mori construction of scar Hamiltonians, for which the commutant algebra approach in~\cite{moudgalya2023exhaustive} provides an exhaustive extension.
The proof of Prop.~\ref{lem:bondalg} is given below. 
\begin{proof}
The argument below is very similar to the argument for the OBC AKLT ground states as scars in App.~B in~\cite{moudgalya2023exhaustive}, which we use here with only small adaptations.
For any two-sites $j$ and $k$, two-site Schmidt vectors in $\ktO$ and $\ktW$ belong to the span of $\ket{00}_{j,k}$ and $(\ket{10} + \ket{01})_{j,k}$.
Denoting orthogonal states $\ket{T}_{j,k} \defn \ket{11}_{j,k}$ and $\ket{S}_{j,k} \defn (\ket{10} - \ket{01})_{j,k}$, we have that $\ktO$ and $\ktW$ are annihilated by $\ketbra{\dots}{T}_{j,k}$ and $\ketbra{\dots}{S}_{j,k}$.
It is easy to check that on a system of size $N$, the common kernel of $\{ \ketbra{\dots}{T}_{j,j+1}, \ketbra{\dots}{S}_{j,j+1}, 1 \leq j \leq N-1 \}$ is given precisely by the span of $\{ \ktO, \ktW \}$.
Hence there are precisely two linearly independent states annihilated by this algebra $\mA_{1,2,\dots,N}^\text{SMobc}$, namely $\{ \ktO_{1,2,\dots,N}, \ktW_{1,2,\dots,N} \}$ for any $N$.
Here and below, when we say ``annihilated by the algebra $\mA_{1,2,\dots,N}^\text{SMobc}$'' we mean annihilated by all the non-trivial generators, i.e., excluding the identity.
It is easy to see that the projector $\ketbra{11 \cdots 1}_{1,2,\dots,N}$ belongs to $\mA_{1,2,\dots,N}^\text{SMobc}$.
Proving that $\mA_{1,2,\dots,N}^\text{SMobc}$ and $\mathcal{C}^{\rm deg}_{\ket{\bar{0}},\ket{W}}$ are commutants of each other is then equivalent to proving the following Lemma.
\begin{lemma}\label{lem:OBC}
For any $N$, the algebra $\mA_{1,2,\dots,N}^\text{SMobc}$ acts irreducibly in the orthogonal complement to the two states $\{ \ktO_{1,2,\dots,N}, \ktW_{1,2,\dots,N} \}$.
Denoting an orthonormal basis in this space as $\{ \ket{\psi_\alpha}, \alpha = 1, \dots, 2^N-2 \}$, this is equivalent to the statement that $\ketbra{\psi_\alpha}{11 \cdots 1}_{1,2,\dots,N} \in \mA_{1,2,\dots,N}^\text{SMobc}$ for all $\alpha$, since all operators of the form $\ketbra{\psi_\alpha}{\psi_\beta}$ can be generated from these and their Hermitian conjugates (since the bond algebra is a $\dagger$-algebra). 
\end{lemma}
\begin{proof}
We proceed by induction, assuming this holds for some $N = k$, we show that it is true for $N = k + 1$.
We begin the induction from $N = 2$, where the span of the corresponding two states $\{ \ktO_{1,2}, \ktW_{1,2} \}$ is simply the span of $\{ \ket{00}_{1,2}, (\ket{10} + \ket{01})_{1,2} \}$.
Hence by definition in Eq.~(\ref{eq:2sitealgebraOBC}), $\mA_{1,2}^\text{SMobc}$ acts irreducibly in the space spanned by $\{ \ket{T}_{1,2}, \ket{S}_{1,2} \}$ that is the orthogonal complement to $\{ \ktO_{1,2}, \ktW_{1,2} \}$.
For induction, we assume the claim holds for $N = k$, or any $k$ consecutively labeled sites in general.
This then implies that the irreducibility holds for the algebras $\mA_{1,2,\dots,k}^\text{SMobc}$ and $\mA_{2,3,\dots,k+1}^\text{SMobc}$.
Given $\left(\ketbra{\psi_\alpha}{11 \cdots 1} \right)_{1,2,\dots,k} \in \mA_{1,2,\dots,k}^\text{SMobc}$ and similar ketbras in $\mA_{2,3,\dots,k+1}^\text{SMobc}$, we can then combine with appropriate $\ketbra{T}{T}, \ketbra{S}{T}$ near the ends to obtain the following ketbras from $\mA_{1,2,\dots,k,k+1}^\text{SMobc}$:
\begin{align*}
& \left(\ketbra{\psi_\alpha}{11 \cdots 1} \right)_{1,2,\dots,k} \left(\ketbra{11}{11} \right)_{k,k+1} = (\ket{\psi_\alpha}_{1,2,\dots,k} \otimes \ket{1}_{k+1}) \bra{11 \cdots 11}_{1,2,\dots,k,k+1} ~; \\
& \left(\ketbra{\psi_\alpha}{11 \cdots 1} \right)_{1,2,\dots,k} \left(\big[\ket{10} - \ket{01} \big] \bra{11} \right)_{k,k+1} = (\ket{\psi_\alpha}_{1,2,\dots,k} \otimes \ket{0}_{k+1}) \bra{11\cdots 11}_{1,2,\dots,k,k+1} ~; \\
& \left(\ketbra{\psi_\alpha}{11 \cdots 1} \right)_{2,3,\dots,k+1} \left(\ketbra{11} \right)_{1,2} = (\ket{1}_1 \otimes \ket{\psi_\alpha}_{2,3,\dots,k+1}) \bra{11 \cdots 11}_{1,2,\dots,k,k+1} ~; \\
& \left(\ketbra{\psi_\alpha}{11 \cdots 1} \right)_{2,3,\dots,k+1} \left(\big[\ket{10} - \ket{01} \big] \bra{11} \right)_{1,2} = -(\ket{0}_1 \otimes \ket{\psi_\alpha}_{2,3,\dots,k+1} ) \bra{11 \cdots 11}_{1,2,\dots,k,k+1} ~,
\end{align*}
where we have used the fact that $\mA^{\text{SMobc}}_{1,\dots,k+1}$ is generated by the algebras $\mA^{\text{SMobc}}_{1,2,\dots,k}$ and $\mA^{\text{SMobc}}_{k,k+1}$ or $\mA^{\text{SMobc}}_{2,3,\dots,k+1}$ and $\mA^{\text{SMobc}}_{1,2}$. 
We can now show that the kets appearing on the R.H.S., namely,
\begin{equation}
\ket{\psi_\alpha}_{1,2,\dots,k} \otimes \ket{1}_{k+1} ~, \quad
\ket{\psi_\alpha}_{1,2,\dots,k} \otimes \ket{0}_{k+1} ~, \quad 
\ket{1}_1 \otimes \ket{\psi_\alpha}_{2,3\dots,k+1} ~, \quad
\ket{0}_1 \otimes \ket{\psi_\alpha}_{2,3\dots,k+1} ~, \label{eq:append_right_left}
\end{equation}
not all linearly independent, span the orthogonal complement to the target space $\{ \ktO_{1,2,\dots,k,k+1}, \ktW_{1,2,\dots,k,k+1} \}$.
To do so, it is sufficient to show that any state orthogonal to the span of the states in Eq.~(\ref{eq:append_right_left}) is annihilated by the algebra $\mA_{1,2,\dots,k,k+1}^{\text{SMobc}}$ (whose null space is precisely the target space).
Consider any $\ket{\phi}_{1,2,\dots,k,k+1}$ orthogonal to the states in Eq.~(\ref{eq:append_right_left}) and decompose it as
\begin{align}
\ket{\phi}_{1,2,\dots,k,k+1} = \ket{u_1}_{1,2,\dots,k} \otimes \ket{1}_{k+1} + \ket{u_0}_{1,2,\dots,k} \otimes \ket{0}_{k+1} ~. 
\end{align}
Requiring orthogonality to the first two states in Eq.~(\ref{eq:append_right_left}), we conclude that $\ket{u_{1/0}}_{1,2,\dots,k}$ are orthogonal to all $\ket{\psi_{\alpha}}_{1,2,\dots,k}$ and hence $\ket{\phi}_{1,2,\dots,k,k+1}$ is annihilated by $\mA^\text{SMobc}_{1,2,\dots,k}$.
By an identical argument using orthogonality to the last two states in Eq.~(\ref{eq:append_right_left}), we conclude that $\ket{\phi}_{1,2,\dots,k,k+1}$ is annihilated by $\mA^\text{SMobc}_{2,3,\dots,k+1}$.
Since $\mA^\text{SMobc}_{1,2,\dots,k,k+1}$ is completely generated by $\mA^\text{SMobc}_{1,2,\dots,k}$ and $\mA^\text{SMobc}_{2,3,\dots,k+1}$, we conclude that $\ket{\phi}_{1,2,\dots,k,k+1}$ is annihilated by $\mA^\text{SMobc}_{1,2,\dots,k,k+1}$ and hence must be in the span of the target space $\{ \ktO_{1,2,\dots,k,k+1}, \ktW_{1,2,\dots,k,k+1} \}$.
This proves that the states in Eq.~(\ref{eq:append_right_left}) indeed span the orthogonal complement to the $\{ \ktO_{1,2,\dots,k,k+1}, \ktW_{1,2,\dots,k,k+1} \}$ states.
Combining the arguments, this proves the claim for $N = k + 1$, completing the induction and hence proving the claim for all $N$.
\end{proof}
This proves that $\mA^\text{SMobc}_{1,2,\dots,N}$ and $\mathcal{C}^{\rm deg}_{\ket{\bar{0}},\ket{W}}$ on sites $1,2,\dots,N$ are commutants of each other, i.e., $\mA^\text{SMobc}_{1,2,\dots,N}$ is the exhaustive bond algebra for Hamiltonians with degenerate scars $\{ \ktO, \ktW \}$ on the system with $N$ sites.
While we have used several generators on each link $j,j+1$ to make the proof particularly transparent, we expect that a smaller number of generators is sufficient, e.g., for $N \geq 3$
\begin{equation}
\mA_{1,2,\dots,N}^\text{SMobc} = \lgen \{ h_{j,j+1}, 1 \leq j \leq N-1 \} \rgen ~,
\label{eq:2sitealgebraOBCgen}
\end{equation}
where $h_{j,j+1}$ is a generic Hermitian operator acting in the space spanned by $\{ \ket{T}_{j,j+1}, \ket{S}_{j,j+1} \}$.
A sufficient condition is that this holds on three-site systems, since then all of the generators used in the lemma proof belong to $\lgen h_{j,j+1}, h_{j+1,j+2} \rgen$. 
We conclude with a simple remark.
By reviewing the above proof, we see that for any Hermitian $h_X$ that annihilates $\{ \ktO, \ktW \}$ and whose support lies inside a region $X$, such $h_X$ can be generated using two-site generators of the exhibited form inside this region.
Indeed, $h_X$ annihilates naturally defined states $\ktO_X, \ktW_X$ on $X$, and the above argument goes through by using consecutively labeled sites from $X$ only.
If $X$ is a bounded region, this conceptually covers all strictly local annihilators.
By extending such bounded $X$ to a bounded covering region with sufficiently nice connectivity (namely, that allows a path that visits all sites using nearest-neighbor connections), the generators can be taken to be nearest-neighbor terms.
\end{proof}

The main message of this Proposition is that \text{all} Hamiltonians that have the $\ket{0}$ and $\ket{W}$ as degenerate eigenstates can be generated {\it in the algebra sense}---i.e., allowing also products of the generators---from the above range-2 generators.
In particular, this is true for all Hamiltonians in Eq.~(\ref{eq:HW}) with $\omega = 0$.
The way the Lemma is proven, any finite-range annihilator $h_X$ can be written in terms of the above generators restricted to the region $X$.
The Lemma implies that $H_{\text{ImHop}}$ can be also written in terms of such generators on the full chain.
However, it does not say how many generators are multiplied in the process to achieve this.
The fact that $H_{\text{ImHop}}$ is a type II operator (e.g., proven purely algebraically in Sec.~\ref{subsec:Wdifferenttypes}) means that extensively many such generators need to be multipled and added, with many cancellations magically happening, to obtain $H_{\text{ImHop}}$.
This shows the limits of the usefulness of the purely algebra thinking and the need to bring in locality considerations separately.
Furthermore, it is hard to determine whether there may be other such type II Hamiltonians in this bond algebra language.
Clarifying this non-trivial aspect is the highlight of analysis of type II Hamiltonians using the operator-string basis in the main text.
The full resolution of locality considerations for this QMBS system is one of the main achievements of this paper.
In the presence of lifting operators, we will also consider the commutant [analogous to Eq.~(\ref{eq:Cscar})]
\begin{align}
    \mathcal{C}^{\text{non-deg}}_{\ket{\bar{0}},\ket{W}}=\lgen\{\ketbra{\bar{0}},\ketbra{W}\}\rgen\quad,
\end{align}
which corresponds to the non-degenerate QMBS with $\ket{\bar{0}}$ and $\ket{W}$, where the bond algebra is $\lgen\{\mA^{\text{deg}}_{\ket{\bar{0}},\ket{W}},\hat{N}_{\rm tot}\}\rgen$, i.e., includes the extensive-local lifting operator $\hat{N}_{\rm tot}$.
In fact, we show a stronger statement below.
\begin{prop}
    \label{prop:Hlift}
    The commutant of $\mathcal{C}^{\text{non-deg}}_{\ket{\bar{0}},\ket{W}}$, is spanned---as a linear space over $\mathbb{C}$---by $\mA^{\text{deg}}_{\ket{\bar{0}},\ket{W}}$ and $\hat{N}_{\rm tot}$,
\begin{equation}
\mA^{\text{non-deg}}_{\ket{\bar{0}},\ket{W}} = \text{span} \{
\mA^{\text{deg}}_{\ket{\bar{0}},\ket{W}}, N_{\text{tot}} \}\quad.
\end{equation} 
\end{prop}
\begin{proof}
We may write a pictorial matrix proof, where we know that the bond algebra $\mA^{\text{deg}}_{\ket{\bar{0}},\ket{W}}$ (closed under Hermiticity) possesses the commutant $\mathcal{C}^{\rm deg}_{\ket{\bar 0},\ket{W}}$. By DCT, the commutant of $\mathcal{C}^{\rm deg}_{\ket{\bar 0},\ket{W}}$ is $\mA^{\text{deg}}_{\ket{\bar{0}},\ket{W}}$.
In terms of matrices in the many-body Hilbert spaces, elements of $\mA^{\text{deg}}_{\ket{\bar{0}},\ket{W}}$ are block-diagonal acting as scalars in the $\ket{\bar{0}}, \ket{W}$ scar space and as arbitrary matrices in the orthogonal complement.
Since there are only two states in the degeneracy breaking QMBS system, adding any Hermitian element that breaks this degeneracy to the algebra $\mA^{\text{deg}}_{\ket{\bar{0}},\ket{W}}$ immediately gives the space of the commutant of $\mathcal{C}^{\rm non-deg}_{\ket{\bar 0},\ket{W}}$ via the DCT.
So we can conclude that the linear span of $\mA^{\text{deg}}_{\ket{\bar{0}},\ket{W}}$ and $N_{\text{tot}}$ is the commutant of $\mathcal{C}^{\rm non-deg}_{\ket{\bar 0},\ket{W}}$.
\end{proof}
%

%
\section{Proof of Theorem~\ref{thm:HW}}
\label{sec:proof}
We now provide a proof of Thm.~\ref{thm:HW} by obtaining conditions on the existence of operators of the form Eq.~(\ref{eq:Lbasis}) in the expansion of any  extensive local operator with $\ket{W}$ as an eigenstate.
We first do not impose any Hermiticity, and derive the conditions in Tab.~\ref{tab:Lterms},  and later discuss which terms are further forbidden when the Hermiticity is imposed,  which leads to the conditions in Tab.~\ref{tab:HIterms}.
\subsection{Conditions on general non-Hermitian terms}\label{sec:nHconditions}
\subsubsection{Operators with $n\geq1$, $m= 0$ are forbidden}
\label{sec:creationforbidden}
Given the condition in Eq.~\eqref{eq:conditionW}, we will first show that it is impossible for $\mathcal{O}$ to contain elements with pure creation operators, i.e., a term with basis element $s^\dag_{j_1}...s^\dag_{j_n}$ 
with $n\geq 1$, $m=0$, and $j_1,...,j_n\in A$ with $A$ being a contiguous region of size $|A|\leq R$.
Let us first consider the simpler case of $n\geq 2$ and $m=0$, such that
\begin{equation}
s^\dag_{j_1}...s^\dag_{j_n}\ket{W}=\frac{1}{\sqrt{N}}\sum_{l=1}^N s^\dag_{j_1}...s^\dag_{j_n}s_l^\dag\ket{\bar{0}},
\end{equation}
where we see that $s^\dag_{j_1}...s^\dag_{j_n}$ produces an $(n+1)$-particle state.
In order for there to be a linear combination of operators, including $s^\dag_{j_1}...s^\dag_{j_n}$, where $\ket{W}$ is an eigenstate, all these higher-particle states must cancel.
However, consider the specific non-vanishing product state that appears in the above sum
$s^\dag_{j_1}...s^\dag_{j_n}s_l^\dag\ket{\bar{0}}$
with $j_1,...,j_n\in A$, but $(|l-j_i|\mod N)\gg R$ $\forall j_i\in A$.\footnote{Here and below, while we write conditions such as distances $\gg R$, for the presented arguments to work it is sufficient to require $N > 3R$. However, we will not consider the situation $R$ being dependent on $N$ in this work.}
Such a state cannot be cancelled using any other linear combination of range-$R$ terms in $\mathcal{L}$ acting on the $W$ state, except for itself $-s^\dag_{j_1}...s^\dag_{j_n}$, which means that $\ket{W}$ can never be an eigenstate in the presence of a non-zero $s^\dag_{j_1}...s^\dag_{j_n}$.
So we see that $\mathcal{O}$ cannot contain a non-zero term with $n\geq 2$ and $m=0$.
Now, consider the special case of $n=1$ and $m=0$, where we similarly see that, upon the application of $c_j s^\dag_j$ (where $c_j\in \mathbb{C}$ is some coefficient), we obtain a non-vanishing piece of the form
$\frac{c_j}{\sqrt{N}} s^\dag_{j}s_l^\dag\ket{\bar{0}},$
with $(|j-l|\mod N)\gg R$, which must be cancelled in order for $\ket{W}$ to be an eigenstate.
Notice that unlike the cases with $n\geq 2$ where we were immediately forced to have $c^{j_1...j_n}=0$, one can potentially cancel this specific product term with exactly one other operator: $-c_j s^\dag_l$.
However, now similar to $l$, consider a third site $p$ with $(|p-j|\mod N)\gg R$ and $(|p-l|\mod N)\gg R$, where we must similarly have a term $-c_j s^\dag_p$.
The existence of both of these terms $-c_j s^\dagger_p$ and $-c_j s^\dag_l$ leads to the creation of the state
$-\frac{2c_j}{\sqrt{N}} s^\dag_{l}s_p^\dag\ket{\bar{0}}$.
This can no longer can be cancelled by any of the available terms since we have used up all our degrees of freedom in setting the other two terms to zero.
Hence $c_j=0$, and $\mathcal{O}$ cannot have any terms with $n\geq 1,m=0$, as is summarized in the first row of Tab.~\ref{tab:Lterms}.
Note that already at this step of the argument we can claim that for such finite-range extensive local operators that have $\ket{W}$ as an eigenstate, the state $\ket{\bar{0}}$ must also be an eigenstate (in fact, annihilated by all present basis terms in $\mathcal{O}$ other than the identity).
\subsubsection{Conditions on operators with other $n$ and $m$}
We will now exhaustively analyze terms with all other possible $n$ and $m$ combinations.
Operators with $n\geq 0$ creation operators combined with $m\geq 2$ annihilation operators (e.g., $c^{j_1 j_2}_{k_1 k_2}s^\dag_{j_1} s^\dag_{j_2} s_{k_1} s_{k_2}$) always annihilate $\ket{W}$ such that $\ket{W}$ is always an eigenstate with eigenvalue $0$, as seen in the second row of Tab.~\ref{tab:Lterms}.
All such terms with $n\geq 0$ and $m\geq 2$ are by themselves finite-range annihilators of $\ket{W}$.
Now, two cases remain: First, let us analyze the case where $m=1$ annihilation operator is combined with $n\geq2$ creation operators such as $c^{j_1...j_n}_k s_{j_1}^\dag ... s_{j_n}^\dag s_{k}$. 
Since this changes the total particle number to $n$, $\ket{W}$ can only be an eigenstate of such an operator if there exists other terms that cancel it.
This term, when applied to $\ket{W}$, gives
$$c^{j_1...j_n}_ks_{j_1}^\dag ... s_{j_n}^\dag s_{k}\ket{W}=\frac{c^{j_1...j_n}_k}{\sqrt{N}}s_{j_1}^\dag ... s_{j_n}^\dag\ket{\bar{0}}\quad,$$
which can only be cancelled by terms that create $n$-particle states, and the only such possible terms are those with $n$ creation operators and $m=1$ annihilation operators or $n-1$ creation operators and $m=0$ annihilation operators.
Since we have already ruled out the existence of $m=0$ terms, it must necessarily be the case that $ s_{j_1}^\dag ... s_{j_n}^\dag \left(\sum_k c^{j_1...j_n}_k s_{k}\right)$ with $\sum_k c^{j_1...j_n}_k=0$.
Since each element in $\mathcal{L}$ have finite range $R$, the condition $\sum_k c^{j_1...j_n}_k=0$ must be satisfied locally for $k\in X$ with $|X|\leq 2R$.
By construction, all such terms are naturally spanned by terms on size $2R$ or less.
For $n=0$, $m=1$, a similar logic applies, since $c_k s_{k}\ket{W} = \frac{c_k}{\sqrt{N}} \ket{\bar{0}}$, which can only be cancelled when $\sum_k c_k=0$.
The result of this analysis is given in the third and fourth row of Tab.~\ref{tab:Lterms}.
All such terms can be decomposed as $c_k = d_k - d_{k-1}$ on a PBC chain, hence $\sum_k c_k s_k = \sum_k d_k (s_k - s_{k+1})$, i.e., spannable by range-2 such terms $s_k - s_{k+1}$ that are annihilators of $\ket{W}$.
Finally, let us analyze the case when $n=m=1$ such as the term $c^j_k  s_j^\dag s_k$, where we have reduced the problem to a simple hopping matrix. For the condition in Eq.~\eqref{eq:conditionW} to be satisfied, the applied operators must not change the particle number of $\ket{W}$, where the only such term are other hopping terms such that $\sum_k c^j_{k} = \lambda,\forall j$, as is shown in the fifth row of Tab.~\ref{tab:Lterms}.
\subsection{Conditions on Hermitian terms}\label{sec:Hconditions}
Now, let us reconsider the effects of Hermiticity and locality.
Let us define Hermitian $h_X$ terms, used in Eq.~\eqref{eq:HW}, which are type I operators.
In the hard-core boson language, these are made of finite-range $R$ hoppings, i.e.,
\begin{equation}
h_X=\sum_{{\rm supp}(L)\in X} c_L L+h.c.,
\end{equation}
where $c_L\in\mathbb{C}$, $L$ (given in Def.~\ref{def:opbasis}) is supported on at most $R$ contiguous sites, and the entire operator is supported on a region of size $|X|\leq R_{\text{max}}$ for some finite $R_{\text{max}}$, and
\begin{align}
    h_X\ket{W}=h_X\ket{\bar{0}}=0\quad.
    \label{eq:annihilationcond}
\end{align}
Note that we will show in the following analysis that one can safely choose $R_{\text{max}}=2R$.
\subsubsection{Operators forbidden due to Hermitian conjugation}
Re-analyzing the terms in Tab.~\ref{tab:Lterms}, the Hermiticity requirement implies that $m\geq 1$ and $n=0$ terms are forbidden since it's Hermitian conjugate involves only creation operators, which were prohibited from the analysis in Sec.~\ref{sec:nHconditions}. 
Similarly, the $n\geq2$ and $m=1$ also encapsulates the $n=1$ and $m\geq 2$ case due to Hermiticity.
Since the Hermitian conjugated part automatically annihilates the state, all extensive-local Hermitian operators with $n\geq2$ and $m=1$ still remain spannable by finite range $h_X$ terms with $n\geq2$ and $m=1$.
Similarly, terms with $n,m\geq 2$ still remain annihilators of $\ket{W}$, and are fully spannable by their $h_X$ counterparts on finite $|X|\leq 2R$.
We summarize the above statements in the first two rows of Tab.~\ref{tab:Lterms}, and preemptively state that this is the exhaustive list which span all type I operators, as will be shown in subsequent sections where we show that $\hat{N}_{\rm tot}$ and $H_{\rm ImHop}$ are type III and type II operators, respectively.
\subsubsection{Operators with $n = m = 1$ and deriving $\hat{N}_{\tot}$ and $H_{\rm ImHop}$}
The question whether terms in $n=m=1$ are spannable by Hermitian terms $h_X$ with $X\leq R_{\max}$, given in Tab.~\ref{tab:HIterms}, remains.
Let us analyze this case for a general $n=m=1$ operator $\mathcal{O}$ that satisfies $\mathcal{O}\ket{W}=\lambda\ket{W}$ with $\lambda\in \mathbb{R}$ due to the Hermiticity of $\mathcal{O}$. Such an operator can be written as
    $$\mathcal{O}=\sum_{\substack{|j-k|\leq R}} c_{j,k} s_j^\dag s_k
    \quad,$$
    where $c_{j,k}=c_{k,j}^*$ for Hermiticity.
    These coefficients must satisfy $\sum_k c_{j,k}=\lambda$ $\forall j$.
    We may consider the real and imaginary components of $c_k$ separately, i.e.,
    \begin{align}
        \sum_k {\rm Re}(c_{jk}) = \lambda ~, \qquad \sum_k {\rm Im}(c_{jk}) = 0,\nonumber
    \end{align}
    such that we may consider decomposing $\mathcal{O}=\mathcal{O}_{\rm Re}+\mathcal{O}_{\rm Im}$, where $\mathcal{O}_{\rm Re}$ and $\mathcal{O}_{\rm Im}$ contain the real and imaginary coefficients, respectively.
    
    Let us first consider $\mathcal{O}_{\rm Re}$, by assuming $c_{jk}\in \mathbb{R}$ which implies a symmetric matrix $c_{jk}=c_{kj}$.
    In this case, all possible finite-range hoppings can be built from type I Hermitian operators $\{P_{j,j+\alpha}^{\rm Re}\}$, given by
    \begin{align}
        P_{j,j+\alpha}^{\rm Re}&:= s_j^\dag s_{j+\alpha}+s_k^\dag s_j-s_j^\dag s_j-s_{j+\alpha}^\dag s_{j+\alpha}\quad,
        \label{eq:PRe}
    \end{align}
    for $1<\alpha\leq R$.
    Note that $P_{jj+\alpha}^{\rm Re}\ket{W}=0$. These terms are all linearly-independent of each other.
    To show that $\mathcal{O}_{\rm Re}$ can be spanned by $\{P_{j,j+\alpha}^{\rm Re}\}$, consider the following procedure where we can `clean up' the $c_{jk}$ matrix: one can first remove the maximum range $R$ jumps by adding the corresponding $-c_{j,j+R}P_{j,j+R}^{\rm Re}$ terms.
    This term effectively removes the range $R$ jumps at the expense of adding entries in the diagonal at $c_{jj}$ and $c_{j+R \,j+R}$.
    In fact, since the above Hermitian operators are linearly-independent, we can repeat this removal process for all range $\alpha$ non-zero $c_{jj+\alpha}>0$ jumps with $R\geq\alpha\geq 1$ by adding a $-c_{j,j+\alpha}P_{j,j+\alpha}^{\rm Re}$ term to cancel the jump term while adding entries in the diagonal at $c_{j,j}$ and $c_{j+\alpha, j+\alpha}$.
    At the end of this process, we will only have terms in the diagonal $c_{j,j}$ entries. However, recall that at every step, we obeyed $\sum_k c_{j,k}=\lambda$ (since the coefficients in $P_{j,k}^{\rm Re}$ obeys $\sum_k c_{j,k}=0$), which implies that $c_{j,j}=\lambda$. This means that we have found a way to decompose any $\mathcal{O}_{\rm Re}$ as
    \begin{align}
        \mathcal{O}_{\rm Re}
        =\lambda \hat{N}_{\rm tot}+\sum^N_{\substack{j=1\\0<\alpha\leq R}} c_{j,j+\alpha}P_{j,j+\alpha}^{\rm Re}\quad,
        \label{eq:ORehop}
    \end{align}
    for $c_{j,k}\in \mathbb{R}$.
    Now, consider $O_{\rm Im}$ by setting $c_{j,k}$ to be purely imaginary such that we have an antisymmetric matrix $c_{j,k}=-c_{k,j}$. In this case the type I Hermitian operator basis $\{P_{j,j+\alpha}^{\rm Im}\}$ is given by
    \begin{align}
        P_{j,j+\alpha}^{\rm Im}&:= i s_j^\dag s_{j+\alpha}-i\sum_{n=1}^{\alpha}s_{j-1+n}^\dag s_{j+n}+h.c.\quad,
        \label{eq:PIm}
    \end{align}
    where $2\leq\alpha\leq R$.
    Notice that, unlike the real case, there is no non-zero Hermitian term such as $s_j^\dag s_j$.
    We may once again aim to reduce the size of $\mathcal{O}_{\rm Im}$ by the procedure of removing the range $R$ jumps by adding $-c_{j,j+R}P_{j,j+R}^{\rm Im}$ at the expense of adding range 1 hoppings (instead of diagonal terms as was the case in $\mathcal{O}_{\rm Re}$).
    Again, due to the linear-independence of the basis, we can follow this line of logic and remove all $>1$ jumps by replacing them with range 1 hoppings.
    Due to $\sum_k c_{j,k}=0$ and $c_{j,k}=-c_{k,j}$, the only possibility for what remains is
    \begin{align}
        H_{\rm ImHop}= \frac{i}{2}\sum_j \left(s_j^\dag s_{j+1}-s_{j+1}^\dag s_j\right)\quad,
    \end{align}
    up to some coefficient $t\in\mathbb{R}$.
    It follows that $\mathcal{O}_{\rm Im}$ can always be written as
    \begin{align}
        \mathcal{O}_{\rm Im}=t H_{\rm ImHop}+\sum_{\substack{j,0<\alpha\leq R}} c_{j,j+\alpha}P_{j,j+\alpha}^{\rm Im}\quad.
        \label{eq:OImhop}
    \end{align}
    This analysis has exhausted all operator possibilities and thus, putting together Eq.~\eqref{eq:ORehop} and \eqref{eq:OImhop} with Tab.~\ref{tab:HIterms}, we have established the general form of $H$ to be the one in Eq.~\eqref{eq:HW}, where we have reincluded the identity operator. As promised, we have shown that $h_X$ with $|X|\leq R_{\text{max}}$ with $R_{\text{max}}=2R$ can additively span all possible extensive-local Hermitian operators, except $\hat{N}_{\rm tot}$ and $H_{\rm ImHop}$, which we will show in the main text cannot be spanned with Hermitian $h_X$ even when $|X|<L$.
    Note that the Hamiltonian $H$ given in an undetermined basis, such as the one in Eq.~\eqref{eq:localHam}, can always be rewritten in terms of creation and annihilation operators (i.e. the hard-core boson language) via the basis in Definition~\ref{def:opbasis}.
    This process may result in an identity element of $O(N)$, but all hard-core boson operators will possess a finite norm, and thus the above analysis holds. \qed
    %

%
\section{$W$ state as a ground state}\label{app:nonuniqueGS}
\subsubsection{Prop.~\ref{prop:Wgroundstate}: $\ket{W}$ is never the unique ground state of any local Hamiltonian}
\begin{proof}
    W.l.o.g., we can set $\Omega = 0$ in Eq.~(\ref{eq:HW}).
    We can prove Eq.~(\ref{eq:W2energy}) by first noticing that $\ket{W^2}$ is an exact eigenstate of $\hat{N}_{\rm tot}$ and $H_{\rm ImHop}$ with eigenvalues $2$ and $0$ (note that the latter is somewhat non-trivial, but known through computations of the Dzyaloshinskii-Moriya term, e.g., in Refs.~\cite{Mark2020Eta, Moudgalya_2022, PRXQuantum.2.020330}), respectively, such that
    \begin{equation}
    \bra{W^2}\left(\omega\hat{N}_{\rm tot}+t H_{\rm ImHop}\right)\ket{W^2}=2\omega.
    \end{equation}
    So we are left with evaluating
    $\bra{W^2}\sum_{X, |X|\leq R_{\text{max}}} h_X\ket{W^2}$
    with $h_X$'s that are of bounded norm and annihilate $\ket{\bar{0}}$ and $\ket{W}$, given by terms in Tab.~\ref{tab:HIterms}, and there are at most an $O(N)$ sum of them.
    Consider a given term $h_X$.
    Since it has range $|X|\leq R_{\text{max}}$, it helps to perform a Schmidt decomposition of $\ket{W^2}$ for the partitioning into the region $X$ and its compliment $X^c$, i.e.,
    \begin{equation}
    \label{eq:schmidtW2}
    \ket{W^2}=\frac{1}{\sqrt{\binom{N}{2}}}\bigg[\sqrt{|X||X^c|}\ket{W}_X\otimes \ket{W}_{X^c} 
    + \sqrt{\binom{|X^c|}{2}}\ket{\bar{0}}_X\otimes\ket{W^2}_{X^c}+\sqrt{\binom{|X|}{2}}\ket{W^2}_X\otimes \ket{\bar{0}}_{X^c}\bigg],
    \end{equation}
    where $\ket{\psi}_{X}$ is the $\ket{\psi}\in\{\ket{W},\ket{W^2},\ket{\bar{0}}\}$ state defined on the $X$ region with the appropriate normalization, and similarly for $\ket{\psi}_{X^c}$. 
    Now, notice that all $h_X$ terms, by construction, annihilate the first two terms in Eq.~(\ref{eq:schmidtW2}), leaving only the last term to possibly contribute a non-zero expectation value for $h_X$.
    Since the norm of $h_X$ is bounded to be $O(1)$, this means that
    \begin{equation}
    \left| \bra{W^2}h_X\ket{W^2} \right| \leq \norm{h_X}\frac{\binom{|X|}{2}}{\binom{N}{2}}=O(N^{-2}),
    \end{equation}
    where $\norm{\cdots}$ is the operator norm and $|X|$ is finite.
    Hence, the expectation value of an $O(N)$ sum of $h_X$-type terms is bounded by
    \begin{equation}
    \big| \bra{W^2}\sum_{|X|  \leq R_{\text{max}}} h_X\ket{W^2} \big| \leq O(N^{-1}).
    \end{equation}
    Putting all the above pieces together, we arrive at Eq.~\eqref{eq:W2energy}.
\end{proof}

\section{Proofs of Type II-ness}\label{app:typeIIness}
\subsubsection{Prop.~\ref{prop:HII}: $H_{\rm ImHop}$ is type II --- proof via $W_q$ ``dispersion''}
\label{app:HImtypeII}
\begin{proof}
Consider a one-particle state with wavevector $q$,
\begin{equation}
\ket{W_q} 
\defn \frac{1}{\sqrt{N}} \sum_j e^{-i q j} s_j^\dagger \ket{\bar{0}}  
= e^{-i q \sum_j j \hat{n}_j} \ktW ~,
\end{equation}
with $q = \frac{2\pi m}{N},~ m \in \mathbbm{Z}$, and $\hat{n}_j = s_j^\dag s_j$.
This state can be also viewed as a twisted $\ktW$ state.

Suppose $h_X$ is a strictly local annihilator of $\ktW$, with support inside a finite contiguous region $X$ of size $|X| \leq R_{\text{max}}$ for some $R_{\text{max}}$.
Denoting the complement to $X$ in the chain of length $N$ as $X^c$, with $|X^c| = N - |X|$, we can write
\begin{equation}
\ktW = \sqrt{\frac{|X|}{N}} \ktW_X \otimes \ktO_{X^c} + \sqrt{\frac{|X^c|}{N}} \ktO_X \otimes \ktW_{X^c} ~,
\label{eq:WXXc}
\end{equation}
with naturally defined normalized states $\ktW_A, \ktO_A$ on subregions $A = X, X^c$.
Clearly,
\begin{equation}
h_X \ktW = 0 \quad \implies \quad h_X \ktW_X = 0~, \quad h_X \ktO_X = 0~.
\end{equation}

We will show that the expectation value of a single $h_X$ in the twisted state $\ket{W_q}$ scales as $\sim \frac{q^2}{N}$ at small $q$.
We first write, similarly to Eq.~(\ref{eq:WXXc}),
\begin{equation}
\ket{W_q} = \sqrt{\frac{|X|}{N}} \ket{W_q}_X \otimes \ktO_{X^c} + \sqrt{\frac{|X^c|}{N}} \ktO_X \otimes \ket{W_q}_{X^c} ~,
\end{equation}
with naturally defined normalized states $\ket{W_q}_A$, $A = X, X^c$.
Since $h_X \ktO_X = 0$ and $\bra{W_q}_{X^c} \ket{\bar{0}}_{X^c} = 0$, we have:
\begin{equation}
\bra{W_q} h_X \ket{W_q} = \frac{|X|}{N} \bra{W_q}_X h_X \ket{W_q}_X ~.
\end{equation}
Suppose $j_0$ is a location inside the region $X$, e.g., in the middle of the region, then we can write $\ket{W_q}_X = e^{-i q \sum_{j \in X} j \hat{n}_j} \ket{W}_X = e^{-i q j_0} e^{-i q \sum_{j \in X} (j-j_0) \hat{n}_j} \ket{W}_X$, since $\sum_{j \in X} \hat{n}_j \ket{W}_X = \ket{W}_X$.
Denoting $\hat{B}_X = \sum_{j \in X} (j-j_0) \hat{n}_j$, we have
\begin{equation}
\begin{aligned}
\bra{W_q}_X h_X \ket{W_q}_X
= \bra{W}_X e^{i q \hat{B}_X} \, h_X \, e^{-i q \hat{B}_X} \ket{W}_X \approx \bra{W}_X \left( h_X + i q \left[ \hat{B}_X, h_X \right] - \frac{q^2}{2} \left[ \hat{B}_X, \left[ \hat{B}_X, h_X \right] \right] \right) \ket{W}_X ~,
\end{aligned}
\label{eq:Wq_hX_Wq_expand}
\end{equation}
where we have assumed that $q$ is small and have used Baker-Campbell-Hausdorff expansion including $O(q^2)$ terms.
Using $h_X \ket{W}_X = 0$ and, {\it crucially, hermiticity} of $h_X$, we see that $O(q^0)$ and $O(q^1)$ terms vanish.
Taking $j_0$ as a mid-point of $X$ and using $\|\hat{n}_j\| = 1$, we have $\| \hat{B}_X \| \leq \sum_{j \in X} |j-j_0| \approx |X|^2/4$ (ignoring fine details of the discrete sums).
Since $\| \ket{W}_X \| = 1$ and we assume that $h_X$ has a bounded norm, we have $\bra{W_q}_X h_X \ket{W_q}_X \sim q^2$ and hence $\bra{W_q} h_X \ket{W_q} \sim \frac{q^2}{N}$, as claimed.
We can obtain a more accurate estimate as far as dependence on $|X|$ is concerned by using the detailed structure of $\hat{B}$ and $\ktW_X$.
First, since $h_X \ktW_X = 0$, we have
$\bra{W}_X [\hat{B}, [\hat{B}, h_X]] \ktW_X = -2 \bra{W}_X \hat{B} h_X \hat{B} \ktW_X$.
Next, we have
\begin{equation}
\begin{aligned}
& |\bra{W}_X \hat{B} h_X \hat{B} \ktW_X | \leq \| h_X \| \| \hat{B} \ktW_X \|^2 = \| h_X \| \frac{1}{|X|} \sum_{j \in X} (j-j_0)^2 \approx \| h_X \| \frac{|X|^2}{12},
\end{aligned}
\end{equation}
where for the sake of illustration, we are doing the discrete sum only schematically, appropriate for largish $|X|$ but easy to extend to small $|X|$.
Putting everything together, we obtain
\begin{equation}
\left| \bra{W_q}_X h_X \ket{W_q}_X \right| \lesssim \frac{q^2 |X|^2 \|h_X\|}{12}
\quad \implies \quad \left| \bra{W_q} h_X \ket{W_q} \right| \lesssim \frac{q^2 |X|^3 \|h_X\|}{12 N} ~.
\label{eq:bound_Wq_hX_Wq}
\end{equation}

For an example to develop some familiarity, consider a Hermitian strictly local annihilator
\begin{align}
h_X = \frac{1}{2} (\ket{10} - \ket{01})(\bra{10} - \bra{01}&)_{j,j+m}=\frac{1}{2}(s_j^\dag s_j+s_{j+m}^\dag s_{j+m}-s_{j+m}^\dag s_j-s_j^\dag s_{j+m}) ~, \nonumber\\
\bra{W_q} h_X \ket{W_q} &= \frac{1 - \cos(qm)}{N} \approx \frac{q^2 m^2}{2N} ~.
\end{align}
For $m = 1$, this nearest-neighbor singlet projector (equivalent to ferromagnetic Heisenberg interaction), is closely related to terms in $H_{\text{ReHop}}$:
$h_X = (n_j + n_{j+1} - s_j^\dagger s_{j+1} - s_{j+1}^\dagger s_j )/2 - n_j n_{j+1}.$
The scaling at small $q$ is in agreement with the general bound.

On the other hand, consider a non-Hermitian strictly local annihilator
\begin{align}
g_X = \frac{i}{2} ( \ket{10} + \ket{01})(  \bra{01}-  \bra{10})_{j,j+m}&=\frac{i}{2}(s_j^\dag s_{j+m}-s_{j+m}^\dag s_j-s_j^\dag s_j +s^\dag_{j+m} s_{j+m})\nonumber\\
\bra{W_q} g_X \ket{W_q} &= \frac{\sin(qm)}{N} \approx \frac{qm}{N} ~. 
\end{align}
Going over the derivation of the $O(q^2)$ scaling in the Hermitian case, we see that the key difference in the non-Hermition case is that an $O(q)$ contribution would remain in an analog of Eq.~(\ref{eq:Wq_hX_Wq_expand}) because $g_X^\dagger \ket{W}_X \neq 0$.
We are now ready to argue that the translationally-invariant pure imaginary hopping Hamiltonian $H_{\text{ImHop}}$ is not type I Hamiltonian in the following sense (which is stronger/more sharply specified than our general definition of type I in Sec.~\ref{subsec:classes}):
We say that a Hamiltonian is type I on a chain of length $N$ if there exist fixed numbers $R_{\text{max}}$, $M$, $C$ such that its $N$-site incarnation for any $N$ can be written as 
\begin{equation}
H_{\text{I}}^{(N)} = \sum_{|X| \leq R_{\text{max}}}^{N_{\text{terms}} \leq C N} h_X ~,
\end{equation}
where $h_X$ are strictly local Hermitian annihilators of $\{ \ktO, \ktW \}$ of range $|X| \leq R_{\text{max}}$ and of norm bounded by $\|h_X\| \leq M$, while $N_{\text{terms}}$ is the total number of terms in the sum and is bounded by $C N$, with $R_{\text{max}}$, $M$, and $C$ independent of $N$.
Indeed, in this case, using Eq.~(\ref{eq:bound_Wq_hX_Wq}), we have
\begin{align}
\left| \bra{W_q} H_{\rm I}^{(N)} \ket{W_q} \right| \lesssim \frac{C}{12} q^2 R_{\text{max}}^3 M ~.
\label{eq:limitHIWq}
\end{align}
On the other hand, $\ket{W_q}$ is an exact eigenstate of $H_{\rm ImHop}$ with expectation value
\begin{align}
\bra{W_q} H_{\rm ImHop} \ket{W_q} = \sin(q) ~.
\end{align}
If we assume that $H_{\text{ImHop}}$ is type I (with fixed but otherwise arbitrary $R_{\text{max}}$, $M$, and $C$), taking $q = 2\pi/N$ and sufficiently large $N$, we would have a contradiction.
Hence $H_{\text{ImHop}}$ is not type I.
On the other hand, $H_{\rm ImHop}$ is type II since we can write it in terms of the non-Hermitian (superscript `\text{non-herm}') local operators
\begin{equation}
P_j^{\text{non-herm}} = \frac{i}{2} (s_j^\dag s_{j+1} -  s_{j+1}^\dag s_j -  s_j^\dag s_j +  s_{j+1}^\dag s_{j+1}) ~,\end{equation}
where $P_j^{\text{non-herm}} \ket{W} = 0$, however $(P_j^{\text{non-herm}})^\dag \ktW \neq 0$.
Using these non-Hermitian annihilators, we see that $H_{\text{ImHop}}$ is given by
\begin{equation}
H_{\text{ImHop}}= \sum_j P_j^{\text{non-herm}} ~.
\end{equation}

\end{proof}

\subsubsection{Prop. \ref{prop:HImhop2}: $H^{(2)}_{\rm ImHop}$ is type II ---proof via boundary action}
\label{app:HImhop2}

\begin{proof}
The following proof follows exactly the same steps as the proof in Sec.~\ref{sec:CuttingHImHop} for the case with $\ket{\bar{0}}$ and $\ket{W}$ as eigenstates.
A natural restriction of $H_{\rm ImHop}^{(2)}$ to a segment $\Lambda = [\ell, \dots, r]$ is
\begin{equation}
\HImHopLambda^{(2)} = \frac{i}{2} \sum_{j=\ell}^{r - 2} (s_j^\dag s^\dag_{j+1} s_{j+1}s_{j+2} - s_{j+2}^\dag s_{j+1}^\dag s_{j+1} s_j).
\end{equation}
It is easy to verify that
\begin{equation}
\HImHopLambda^{(2)} \ket{W^2} = \frac{i}{2} (\hat{n}_\ell \hat{n}_{\ell+1} - \hat{n}_{r-1} \hat{n}_r) \ket{W^2}. 
\label{eq:HImHop_Lambda_ktW2}
\end{equation}
That is, the action of $\HImHopLambda$ on $\ket{W^2}$ can be represented as a boundary action with non-Hermitian operators near the corresponding boundaries.
However, we will show by contradiction that it cannot be represented using Hermitian boundary operators.

We will demonstrate by contradiction that no Hermitian operators $A_\ell$ and $B_r$, acting on finite regions near $\ell$ and $r$ respectively, can realize the boundary action of $H^{(2)}_{{\rm ImHop},\Lambda}$ on $\ket{W^2}$ in Eq.~\eqref{eq:HImHop_Lambda_ktW2}:
\begin{align}
     \frac{i}{2} (\hat{n}_\ell \hat{n}_{\ell+1} - \hat{n}_{r-1} \hat{n}_r) \ket{W^2}=(A_\ell+B_r)\ket{W^2}\quad.
     \label{eq:HImhop2LambdaAB}
\end{align}
Here, we assume that the ranges of $A_\ell$ and $B_r$ are bounded by a fixed number $R_{\text{max}}$, while the size of $\Lambda$ is allowed to be large.
Specifically, denoting the supports of $A_\ell$ and $B_r$ as $X_\ell$ and $X_r$ respectively, and assuming w.l.o.g.\ that $\ell, \ell+1 \in X_\ell$ and $r-1, r \in X_r$, it suffices for $\Lambda$ to be large enough such that $X_\ell$ and $X_r$ are disjoint. For example, if the ranges of $A_\ell$ and $B_r$ are bounded by $R_{\text{max}}$, it suffices to have $|r - \ell| \geq 2 R_{\text{max}}$.

Let us assume that such Hermitian $A_\ell$ and $B_r$ exist and calculate overlap of both sides of Eq.~\eqref{eq:HImhop2LambdaAB} with $\ket{W^2}_{X_\ell} \otimes \ktO_{X_\ell^c}$, where $X_\ell^c$ is the compliment of $X_\ell$. We obtain
\begin{equation}
c_1 \frac{i}{2}\bra{W^2}_{X_\ell}  \hat{n}_\ell \hat{n}_{\ell+1} \ket{W^2}_{X_\ell}= c_1 \bra{W^2}_{X_\ell} A_\ell\ket{W^2}_{X_\ell}+c_1 \bra{\overline{0}}_{X_\ell^c} B_r \ktO_{ X_\ell^c}  ~,
\end{equation}
with $c_1=\sqrt{\binom{|X_\ell|}{2}/\binom{N}{2}}$,
where we have used $\hat{n}_{r-1}\hat{n}_r\ket{\bar{0}}_{X_\ell^c}=0$. Now, notice that the R.H.S.\ is real due to the assumed Hermiticity of $A_\ell$ and $B_r$, whereas the L.H.S.\ is purely imaginary, since $\bra{W^2}_{X_\ell}  \hat{n}_\ell \hat{n}_{\ell+1} \ket{W^2}_{X_\ell}\neq0$ and is also real.
Thus, the assumption that $\HImHopLambda^{(2)}$ on $\ket{W^2}$ has boundary action represented by Hermitian operators is not valid, and hence it is a type-II Hamiltonian.
Now, to show that $H_{\rm ImHop}^{(2)}$ is in a different type II equivalence class than $H_{\rm ImHop}$, observe that $\ket{W_q}$, defined in Eq.~\eqref{eq:Wqboost}, has eigenvalue of
\begin{align}
    H_{\rm ImHop}\ket{W_q}=\sin q \ket{W_q},
\end{align}
while we deduced in Eq.~(\ref{eq:limitHIWq}) that for type I Hamiltonians [of norm $O(N)$] we have
$\bra{W_q} H_{\text{type-I}} \ket{W_q} = O(q^2).$
On the other hand, The eigenvalue of $H_{\rm ImHop}^{(2)}$ for the state $\ket{W_q}$ is
\begin{align}
    H_{\rm ImHop}^{(2)}\ket{W_q}=0\quad.
\end{align}
Thus, we see that there is no way for an addition of a type I operator to $H_{\rm ImHop}^{(2)}$ to ever match the linear $q$ dependence in the eigenvalue of $H_{\rm ImHop}$, meaning that $H_{\rm ImHop}^{(2)}$ and $H_{\rm ImHop}$ are not in the same equivalence class.
\end{proof}
We can note an immediate generalization of the first part of the proof:
\begin{equation}
H_{\rm ImHop}^{(p)} := i\sum_j (s_j^\dag s^\dag_{j+1} \dots s^\dag_{j+p-1} s_{j+1} s_{j+2} \dots s_{j+p} - \text{H.c.})
= i\sum_j (\ketbra{11 \dots 10}{01 \dots 11} - \text{H.c.})_{j,\dots,j+p} 
\label{eq:HImHop_p}
\end{equation}
is a type II Hamiltonian in when $\ket{W^{p'}}$ states are eigenstates for $p \leq p'$.
We can show that $H_{\rm ImHop}^{(p)}$ annihilates all $\ket{W^m}$ states, $m = 0, 1, \dots, N$ (i.e., all Dicke states, or equivalently the ferromagnetic tower of states).
This is clear for $m < p$, while for $m \geq p$ one can use same methods that show that the ferromagnetic tower is annihilated by the DMI Hamiltonian~\cite{Mark2020Eta}, which is essentially $H_{\rm ImHop}^{(1)}$.
We conjecture that $\{ H_{\rm ImHop}^{(p')}, 1 \leq p' \leq p\}$ are independent type II Hamiltonians for families of models that have the $\ket{W^p}$ state as an exact scar, but leave studying this for future work.
%

\section{Asymptotic QMBS for the $W$ state}
\label{app:asymptoticscar}
\subsubsection{Prop. \ref{prop:asymptoticscar}: If $\ket{W}$ is an exact QMBS, then $\ket{W_q}$ is an asymptotic QMBS with lifetime $\gtrsim O(N)$}
\begin{proof}
    We will show that the energy variance $\Delta H_{W_q}^2$ of $\ket{W_q}$ is bounded as
    \begin{align}
        \Delta H_{W_q}^2 :=\langle H^2\rangle_q-\langle H\rangle_q^2\leq O(q^2)\quad,
    \end{align}
    where $H$ is the general Hamiltonian given in Eq.~(\ref{eq:HW}) and $\langle \dots \rangle_q := \bra{W_q} \dots \ket{W_q}$.
    First, let us note that $\ket{W_q}$ is an eigenstate of $t H_{\rm ImHop}+\omega\hat{N}_{\rm tot}+\Omega\mathds{1}$ with the eigenvalue $t\sin q +\omega+\Omega$.
    Hence, these do not enter the energy variance, which is then given by
    \begin{align}
        \Delta H_{W_q}^2 = \big\langle \big( \sum_{|X|\leq R_{\text{max}}}h_X \big)^2 \big\rangle_q - \big\langle \sum_{|X|\leq R_{\text{max}}} h_X \big\rangle_q^2 \quad \leq \big\langle \big( \sum_{|X|\leq R_{\text{max}}}h_X \big)^2 \big\rangle_q.
    \end{align}
    To obtain further bounds, recall from Prop.~\ref{cor:W0deg} that any term in $h_X$, decomposed in the normal-ordered basis in Eq.~(\ref{eq:Lbasis}), contains $\geq 1$ creation operators composed with $\geq 1$ annihilation operators.
    This implies that a product $h_X h_Y$ with $X \cap Y = \emptyset$ must contain two separated annihilation operators that will annihilate $\ket{W_q}$, i.e., $h_X h_Y \ket{W_q} = 0$.
    When $X\cap Y\neq\emptyset$, then it can happen that a non-trivial combination of creation and annihilation operators at some sites $j \in X\cap Y$ does not result in immediate annihilation of $\ket{W_q}$.
    Using this logic, we can write
    \begin{align}
        \big\langle \big(\sum_{|X|\leq R_{\text{max}}}h_X \big)^2 \big\rangle_q 
        =  \big\langle \big(\sum_{\substack{|X|\leq R_{\text{max}},~ |Y| \leq R_{\text{max}},\\ X \cap Y \neq \emptyset}} h_X h_Y \big) \big\rangle_q 
        = \big\langle \sum_{|\widetilde{X}| < 2R_{\text{max}}} \widetilde{h}_{\widetilde{X}} \big\rangle_q\quad,
    \end{align}
    for some $O(N)$ finite-range and bounded Hermitian operators $\widetilde{h}_{\widetilde{X}}$
    (where in the second expression $h_X \neq h_Y$, we can combine the term with its partner to obtain a hermitian $\widetilde{h}_{\widetilde{X}} = h_X h_Y + h_Y h_X$ on $\widetilde{X} = X \cup Y$).
    Again using the derivation of Prop.~\ref{prop:HII}, we can conclude that
    \begin{equation}
    \big\langle \sum_{|\widetilde{X}| < 2R_{\text{max}}} \widetilde{h}_{\widetilde{X}} \big\rangle_q 
    \leq O(q^2).
    \end{equation}
    Combining all these observations, we have the energy variance
    \begin{align}
        \Delta H_{W_q}^2 \leq O(q^2) \quad,
    \end{align}
    which means that for $q= \frac{2\pi m}{N}$ for $m\in \mathbb{Z}$, $|m| \ll N$, the energy variance tends to zero as $1/N^2$ in the thermodynamic limit, which implies that the lifetime of $\ket{W_q}$ goes as $\sim N$.
\end{proof}

\subsubsection{Prop.~\ref{prop:W2}: If $\ket{W}$ is an exact scar, then $\ket{W^2}$ is an asymptotic scar with lifetime $\gtrsim O(N^{1/2})$}
\label{app:W2}

\begin{proof}
Let us show that the energy variance of $\ket{W^2}$ is given by
    \begin{align}
        \Delta H_{W^2}^2 \defn \langle H^2 \rangle_{W^2} - \langle H \rangle_{W^2}^2 \leq O(N^{-1}) \quad,
    \end{align}
    where $H$ is the general Hamiltonian given in Eq.~(\ref{eq:HW}) [with $\Omega = 0$ w.l.o.g.] and $\langle \dots \rangle_{W^2} \defn \bra{W^2} \dots \ket{W^2}$.
    Recall from the proof of Prop.~\ref{prop:Wgroundstate} in App.~\ref{app:nonuniqueGS} that $\ket{W^2}$ is an eigenstate of $\mathds{1}$, $\hat{N}_{\rm tot}$, and $H_{\rm ImHop}$, hence the energy variance simplifies to
    \begin{equation}
    \Delta H_{W^2}^2 = \bigg\langle \bigg(\sum_{|X|\leq R_{\text{max}}} h_X \bigg)^2 \bigg\rangle_{W^2} -\bigg\langle \sum_{|X|\leq R_{\text{max}}} h_X \bigg\rangle_{W^2}^2\quad \leq \bigg\langle \bigg(\sum_{|X|\leq R_{\text{max}}} h_X \bigg)^2 \bigg\rangle_{W^2},
    \end{equation}
    where $h_X$ annihilate the $\ket{W}$ and $\ket{\bar{0}}$ states (and are comprised of contributions given in Tab.~\ref{tab:HIterms}).
    To further bound the variance, observe that for a given $h_X$, using the Schmidt-decomposed form of $\ket{W^2}$ over regions $X$ and its compliment $X^c$, given in Eq.~\eqref{eq:schmidtW2}, we arrive at
    \begin{align}
        h_X \ket{W^2} = \sqrt{\frac{\binom{|X|}{2}}{\binom{N}{2}}}\, h_X \ket{W^2}_X \otimes \ket{\bar{0}}_{X^c}\quad,
        \label{eq:hX_W2}
    \end{align}
    where we have used the fact that $h_X \ket{W}_X = h_X \ket{\bar{0}}_X = 0$.
    Now, we notice that when we apply another annihilator term $h_Y$ to this equation, if there is no overlap between regions $X$ and $Y$, i.e., $X \cap Y = \emptyset$, then
    $h_Y h_X \ket{W^2} = 0$.
    However, if $X\cap Y\neq\emptyset$, then using Eq.~(\ref{eq:hX_W2}) the terms are still bounded by
    \begin{equation}
    |\langle h_Y h_X\rangle_{W^2}|\leq \norm{h_X}\norm{h_Y}\frac{\sqrt{\binom{|X|}{2} \binom{|X\cup Y|}{2}}}{\binom{N}{2}},
    \end{equation}
    where we have also used Schmidt decomposition of $\bra{W^2}$ for the bipartition into $X \cup Y$ and its complement.\footnote{Note that we could employ the hermiticity of $h_Y$, which has not been used so far, to improve this bound replacing $|X \cup Y| \to |Y|$ by using action of $h_Y$ on the $\bra{W^2}$ in the spirit of Eq.~(\ref{eq:hX_W2}).}
    This means that
    \begin{align}
        \bigg\langle \bigg(\sum_{|X|\leq R_{\text{max}}}h_X \bigg)^2\bigg\rangle_{W^2}= \bigg\langle \sum_{\substack{|X|\leq R_{\text{max}},\\ |Y|\leq R_{\text{max}},\\ X\cap Y\neq\emptyset}}h_Y h_X\bigg\rangle_{W^2}\leq O({N}^{-1})\,,
    \end{align}
    since the sum in the second expression is over an $O(N)$ number of terms.
    Thus, we have derived Eq.~(\ref{eq:W2variance}), which implies that the lifetime of $\ket{W^2}$ is $\sim\sqrt{N}$.
\end{proof}

\subsubsection{Prop.~\ref{prop:Wpasymptoticscar}: If $\ket{W}$ is an exact scar, then $\ket{W^p}$ is an asymptotic scar with lifetime $\gtrsim O(N^{1/2})$}
\label{app:Wpasymptoticscar}
\begin{proof}
    Our arguments roughly follow the proof of Prop.~\ref{prop:W2} in the main text of $\ket{W^2}$ as an asymptotic scar, with new technical steps needed for $p \geq 3$.
    We start with writing of the parent Hamiltonian given by Eq.~(\ref{eq:HW}).
    Consider a given $h_X$ with bounded range and bounded norm.
    The Schmidt decomposition of $\ket{W^p}$ for a bipartition into regions $X$ and its compliment $X^c$ is given by
    \begin{equation}
        \ket{W^p} = \sum_{l=0}^{p}  f_{l}^{X,N} \ket{W^l}_X\otimes \ket{W^{p-l}}_{X^c} ~, \qquad
        f_{l}^{X,N} \defn \sqrt{\frac{\binom{|X|}{l}\binom{N-|X|}{p-l}}{\binom{N}{p}}} \approx \sqrt{\frac{\binom{|X|}{l} p!}{(p-l)!}}\, \frac{1}{\sqrt{N^l}} ~,
        \label{eq:flapprox}
    \end{equation}
    where we have assumed for simplicity $|X| \geq p$ (otherwise, the sum over $l$ terminates at $|X|$).
    In the second equation, we have also shown the behavior of the amplitudes $f_{l}^{X,N}$ for large $N$, assuming $|X|$ and $p$ are fixed ($l \leq \text{min}\{|X|,p\}$).
    Upon application of $h_X$, recall that $h_X\ket{\bar{0}}_X=h_X\ket{W}_X=0$, such that
    \begin{equation}
    h_X \ket{W^p} = \sum_{l=2}^{p} f_{l}^{X,N} h_X \ket{W^l}_X\otimes \ket{W^{p-l}}_{X^c}\,\;\text{and}\;\;\langle h_X \rangle_{W^p} = \sum_{l=2}^p (f_{l}^{X,N})^2 \bra{W^l}_X h_X \ket{W^l}_X ~,
    \label{eq:hXWp}
    \end{equation}
    where $\langle...\rangle_{W^p}:=\bra{W^p}...\ket{W^p}$.
    This is upper-bounded by
    \begin{equation}
    |\langle h_X \rangle_{W^p}| \leq \|h_X\| \sum_{l=2}^p (f_{l}^{X,N})^2 \approx \|h_X\| (f_{2}^{X,N})^2
    \approx \|h_X\| \frac{|X|(|X|-1) p(p-1)}{2 N^2}
    \end{equation}
    for large $N$.
    Hence,
    \begin{equation}
    \big|\big\langle \sum_X h_X \big\rangle_{W^p}\big| \leq \frac{C}{N} ~.
    \label{eq:h_XWpexp}
    \end{equation}
    Thus, the energy of $\ket{W^p}$ (as a trial state) can be shifted from $p \omega$ by an amount that decreases as $1/N$ (where we have used that $H_{\text{ImHop}} \ket{W^p} = 0$ and $N_{\text{tot}} \ket{W^p} = p \ket{W^p}$).
    Consider now the calculation of the variance.
    In the case of $p=2$ in the previous section, we had $h_Y h_X \ket{W^p} = 0$ if $X \cap Y = \emptyset$.
    This holds also for $p = 3$ but is no longer true for $p \geq 4$.
    Nevertheless, we can controllably analyze such contribution as follows.
    Assuming $X \cap Y = \emptyset$ we apply $h_Y$ to the left equation in Eq.~(\ref{eq:hXWp}) to obtain
    \begin{equation}
        h_Y h_X\ket{W^p} = \sum_{l,m=2}^{p} f_{l}^{X,N}f_{m}^{Y,N-|X|} h_X\ket{W^l}_X \otimes h_Y\ket{W^m}_Y \otimes \ket{W^{p-l-m}}_{(X\cup Y)^c} ~,
    \end{equation}
    where we have successively Schmidt decomposed the state, first into $X$ and $X^c$, and then we have further partitioned $X^c$ into $Y$ and $(X \cup Y)^c$.
    It follows that
    \begin{equation}
        \langle h_Y h_X \rangle_{W^p} = \sum_{l,m=2}^p f^{X\cup Y,N}_{l+m}f_{l}^{X,N}f^{Y,N-|X|}_{m} \,
        \bra{W^{l+m}}_{X\cup Y}h_Y h_X \ket{W^l}_X \otimes \ket{W^m}_Y ~.
    \end{equation}
    The above equation for $X \cap Y = \emptyset$ is upper-bounded by
    \begin{equation}
        |\langle h_Y h_X \rangle_{W^p}| \leq \|h_X\| \|h_Y\| \sum_{l,m=2}^p f^{X\cup Y,N}_{l+m}f_{l}^{X,N}f^{Y,N-|X|}_{m} 
        = O(N^{-4}) ~,
        \label{eq:hyhxnooverlap}
    \end{equation}
    where the leading scaling with $N$ comes from the $l = m = 2$ term.

    On the other hand, when $X \cap Y \neq \emptyset$, we can write, e.g., 
    \begin{equation}
        h_Y h_X \ket{W^p} = \sum_{l=2}^p f_l^{X\cup Y,N} h_Y h_X \ket{W^l}_{X\cup Y} \otimes \ket{W^{p-l}}_{(X\cup Y)^c} ~,
    \end{equation}
    such that
    \begin{equation}
        \langle h_Y h_X\rangle_{W^p}=\sum_{l=2}^{p} \left(f_{l}^{X\cup Y,N}\right)^2\bra{W^l}_{X\cup Y}  h_Y h_X\ket{W^l}_{X\cup Y}.
    \end{equation}
    Such an expectation value is upper-bounded by
    \begin{equation}
        |\langle h_Y h_X\rangle_{W^p}| \leq \sum_{l=2}^{p} \left(f_{l}^{X\cup Y,N}\right)^2 \|h_Y h_X\|
        \approx \left(f_{2}^{X\cup Y,N}\right)^2  \|h_Y h_X\| = O(N^{-2}) ~.
        \label{eq:hyhxoverlap}
    \end{equation}
    Putting Eqs.~(\ref{eq:hyhxnooverlap}) and (\ref{eq:hyhxoverlap}) together, we can upper-bound
    \begin{equation}
        \big|\big\langle \big(\sum_X h_X \big)^2\big\rangle_{W^p}\big| \leq O(N^{-1}) ~,
    \end{equation}
    since there are $O(N^2)$ contributions from $X \cap Y = \emptyset$, each bounded by $O(N^{-4})$, and $O(N)$ contribution from $X \cap Y \neq \emptyset$, each bounded by $O(N^{-2})$.
    The energy variance of $\ket{W^p}$, calculated using the above equation and Eq.~(\ref{eq:h_XWpexp}), is
    \begin{align}
        \big\langle \big(\sum_X h_X \big)^2\big\rangle_{W^p}-\big\langle \sum_X h_X \big\rangle_{W^p}^2 \leq O(N^{-1}) \quad.
    \end{align}
    This gives a lifetime of $\ket{W^p}$ to also be $\sim \sqrt{N}$, where $p \ll N$ to have controlled approximations such as in Eq.~(\ref{eq:flapprox}).
\end{proof}

\section{Details on the dynamical signatures of Hamiltonian types}
\label{app:dynamics}
\subsection{Early time dynamics of the $W$ droplet evolution} \label{app:earlytime}
Here we analyze behavior of the overlap reduction $\Upsilon_{G=0}(t,M)$, defined in Eq.~(\ref{eq:defUpsilon}), at very small $t$.
Using Eq.~(\ref{eq:Upsilon_tM}), we obtain:
\begin{equation*}
\lim_{t \to 0} \Upsilon_{0}(t,M) = \frac{1}{M} \int_{-\pi}^\pi \frac{dq}{2\pi} \, \frac{\sin^2(qM/2)}{\sin^2(q/2)} \,\left(\frac{\epsilon_q^2}{2} t^2 + i \epsilon_q t \right) ~,
\end{equation*}
with $O(t^4)$ corrections to the real part and $O(t^3)$ to the imaginary part.
We now specialize to the dispersions in Eqs.~(\ref{eq:epsReHop})-(\ref{eq:epsCHop}).
At very early times, i.e., $wt \ll 1$, we can easily evaluate the $t \to 0$ behaviors:
\begin{align}
& \lim_{t \to 0} \Upsilon_0^{\text{ReHop}}(t,M) = i\frac{w t}{M} + \frac{w^2 t^2}{2M} ~, \\
& \lim_{t \to 0} \Upsilon_0^{\text{ImHop}}(t,M) = \frac{w^2 t^2}{2M} ~, \\
& \lim_{t \to 0} \Upsilon_0^{\text{CHop}}(t,M) = i\frac{\alpha w t}{M} +\frac{(\alpha^2 + \beta^2) w^2 t^2}{2M}
~\label{eq:earlytimes},
\end{align}
where for simplicity we have assumed $M \geq 2$ in each case.
We can also understand these results using the very early time expansion 
\begin{equation}
\braket{\phi_0}{\phi(t)} \approx 1 - i t \bra{\phi_0} H \ket{\phi_0} - \frac{t^2}{2} \bra{\phi_0} H^2 \ket{\phi_0} ~,
\end{equation}
employing real-space expressions for $H = H_{\text{ReHop}}$ in Eq.~(\ref{eq:HI}) and $H = H_{\text{ImHop}}$ particularly in Eq.~(\ref{eq:HImHop_nonherm_annih}).
The overall $1/M$ factor survives in Eq.~(\ref{eq:earlytimes}) even though $H$ is extensive because $H_{\text{ReHop}}$ is type I and $H_{\text{ImHop}}$ is type II, and hence both can be written as sums of local annihilators of the $\ket{W}, \ket{\bar{0}}$ states (see Sec.~\ref{subsec:Wdifferenttypes} and the quoted real-space forms).
Hence the action  of $H$ on $\ket{\phi_0}$ in both cases reduces to actions only near the boundaries of the initial $W$ state domain, and one can verify the above early-time results using this real-space picture.
The annihilators are Hermitian in the first case and non-Hermitian in the second case, and the details of the annihilators matter in the specific calculations of $\bra{\phi_0} H \ket{\phi_0}$ and $\bra{\phi_0} H^2 \ket{\phi_0} = \|H \ket{\phi_0}\|^2$, but the overall prediction for the fidelity $|\!\braket{\phi_0}{\phi(t)} \!|^2 \approx 1 - \frac{a}{M} t^2$ is similar in both cases, and also in the more general type II case $H = H_{\text{CHop}}$.
[In the $H_{\text{ReHop}}$ case, we see a pure imaginary piece $\sim t$ in the overlap, and this reduces somewhat the decrease in the fidelity compared to the $H_{\text{ImHop}}$ case; on the other hand, the $H_{\text{CHop}}$ lies between the two cases.]
All in all, the early-time behavior of the overlap with the initial state does not appear to differentiate qualitatively between the type I and type II Hamiltonians, and it is not clear from this analysis if there can appear a qualitative difference between the two types at later times.
We finally expect this structure to generalize to arbitrary extensive-local Hamiltonians that have $\ktW, \ktO$ as exact eigenstates, including interacting ones, since for any such Hamiltonian we know from Eq.~(\ref{eq:HW}) that its action on $\ket{\phi_0}$ reduces to actions only near the domain boundaries.
One may ask about early time analysis of $\braket{T_G\phi_0}{\phi(t)}$, Eq.~(\ref{eq:Upsilon_tM}), perhaps revealing the difference between the two types.
However, it is not clear how to make $\lim_{t\to 0}$ analysis physically interesting/meaningful, either with fixed $G \neq 0$ or attempting $G = ut$ (problematic since $G$ must be an integer).
We can see a more significant early-time difference by examining observable $n_j(t)$,
\begin{equation}
n_j(t) = \bra{\phi_0} e^{iHt} n_j e^{-iHt} \ket{\phi_0} \approx \bra{\phi_0} (n_j + it[H, n_j] - \frac{t^2}{2}[H, [H, n_j]]) \ket{\phi_0} ~. 
\end{equation}
Focusing on the $O(t)$ piece, it can also be related to expectation values of the appropriate local current operators, defined as $i[H, n_j] \defn J_{j-1} - J_j$, which for the two cases are $J_j^{\text{ReHop}} = -\frac{i}{2} (s_j^\dagger s_{j+1} - \text{H.c.})$ and $J_j^{\text{ImHop}} = \frac{1}{2} (s_j^\dagger s_{j+1} + \text{H.c.})$.
We have
\begin{equation}
\bra{\phi_0} J_j^{\text{ReHop}} \ket{\phi_0} = 0~, \;\; \forall j~; \qquad
\bra{\phi_0} J_j^{\text{ImHop}} \ket{\phi_0} = \frac{1}{M} \delta_{j \in [1, \dots, M-1]} ~. 
\end{equation}
Hence, in the $H_{\text{ReHop}}$ case there is no $O(t)$ time dependence in $n_j(t)$ for all sites.
On the other hand, in the $H_{\text{ImHop}}$ case, since the expectation values of the local current operators are non-vanishing only within the droplet, we get $dn_j(t)/dt|_{t=0} = \frac{1}{M} (\delta_{j,M} -\delta_{j,1})$, which one can think of as an early time precursor of the ballistic motion of the droplet observed at $w t \gtrsim 1$ in this case.
Note that while we obtained these results more simply by appealing to the evaluations of the corresponding current operators, the fact that only the sites near the boundary are affected is related to being able to write $H$ as a sum of local annihilators of $\ktW$ and $\ktO$, either Hermitian or non-Hermitian. 
\subsection{Intermediate time dynamics of the $W$ droplet evolution}\label{app:inttime}
\begin{figure*}
\includegraphics[scale=1]{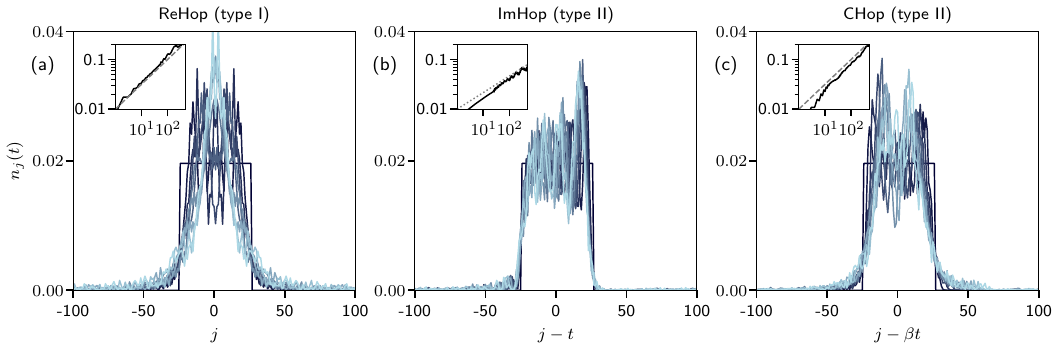}
\label{fig:Wdropletorbitals}
\caption{Evolution of the orbitals $\ket{\phi_j(t)}$ visualized using $n_j(t)$ under the Hamiltonians (a) $H_{\rm ReHop}$, (b) $H_{\rm ImHop}$, and (c) $H_{\rm CHop}$ (with $\alpha = \beta = 0.5$).
Data shown for the initial droplet of size $M=51$ in the PBC chain of length $N = 201$ (same as in Fig.~\ref{fig:Wdropletnumerics}), with uniform time steps from $t = 0$ (darkest) to $t = 200$ (lightest) in steps of $20$.
Note the shift in the $x$ axes in the ImHop and CHop panels that compensates for the ballistic propagation of the droplet seen in Fig.~\ref{fig:Wdropletnumerics}.
The insets show the ``leakage''  of the particle number from the initial domain in the ReHop case and the appropriately shifted domains in the ImHop and CHop cases as a function of time [i.e., $\sum_{j,\, |j - G(t)| > M/2} n_j(t)$ with appropriate $G(t)=0$, $wt$, and $\beta wt$ for the three cases respectively].
The dashed guidelines in the insets of (a) and (c) denote $\sim \sqrt{t}$, and the dotted guideline in the inset of (b) denotes $\sim \sqrt[3]{t}$, which is consistent with the analysis in the text.
}
\end{figure*}
Here we present the details on the intermediate times $wt \gtrsim 1$, where we find qualitative difference between the two types.
Note that the discussion up to Eq.~(\ref{eq:Upsilon_tM}) in the main text is general and valid for any dispersion $\epsilon_q$.
Here we will distinguish the type I ($H_{\text{ReHop}}$) vs type II ($H_{\text{ImHop}}$ and $H_{\rm CHop}$) Hamiltonians.
\begin{figure*}
\centering
\includegraphics[scale=1]{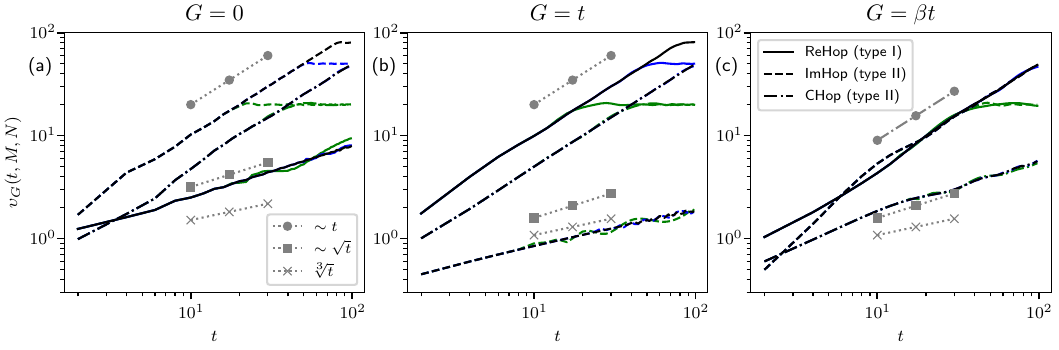}
\caption{
Plots of $\upsilon_G(t, M, N) \defn M \Upsilon_G(t, M, N)$ as evaluated exactly from Eq.~(\ref{eq:defUpsilon}) for times $t$ chosen such that indicated $G(t) \in \mathbb{Z}$.
Data shown for $N = 200$ and different values of $M = 20, 50, 80$ (green, blue, black) for three kinds of Hamiltonians $H_{\text{ReHop}}$, $H_{\text{ImHop}}$, and $H_{\text{CHop}}$ (with $\alpha = \beta = 0.5$) for three different choices of $G(t)$. 
(a) $G = 0$ shows that $\upsilon_G^{\text{ImHop}}$ and $\upsilon_G^{\text{CHop}}$ exhibit linear growth due to ballistic motion, and $\upsilon_G^{\text{ReHop}}$ exhibits a $\sqrt{t}$ growth due to boundary diffusive melting, as discussed in the text.
(b) [resp.\ (c)] $G(t) = wt$ [resp.\ $G(t) = \beta wt$] shows that the ballistic motion under $H_{\rm ImHop}$ [resp.\ $H_{\rm CHop}$] can be canceled by studying the overlap w.r.t.\ a moving droplet with this $G(t)$, $\upsilon_G^{\rm ImHop}$ [resp.\ $\upsilon_G^{\rm CHop}$] then exhibits a $\sqrt[3]{t}$ [resp.\ $\sqrt{t}$] growth, indicating subdiffusive [resp.\ diffusive] boundary melting, as discussed in the text;
on the other hand, $\upsilon_G^{\rm ReHop}$ and $\upsilon_G^{\rm CHop}$ [resp.\ $\upsilon_G^{\rm ReHop}$ and $\upsilon_G^{\rm ImHop}$] exhibit linear growths since they are moving ballistically w.r.t.\ this moving droplet.
Note that legends are common for all the panels, and are shown in only one of them to avoid cluttering.}
\label{fig:scalings}
\end{figure*}
Recall from the discussion in Sec.~\ref{subsubsec:intermediate_time} that for largish $t$ we expect 
that important wave-vectors in the calculation of $\Upsilon_G(t,M)$, Eq.~(\ref{eq:Upsilon_tM}), are small $q \sim q_*(t)$ set by $|\epsilon_{q_*(t)} t - q_* G| \sim O(1)$.
We furthermore assume that $M$ is large enough such that $q_*(t) M \gg 1$.
In this regime the factor $\sin^2(qM/2)$ in Eq.~(\ref{eq:Upsilon_tM}) oscillates very quickly over the important $q$ integration range, and we can replace it by its average value of $1/2$, obtaining:
\begin{equation}
\Upsilon_G[t, M \gg q_*^{-1}(t)] \approx \frac{1}{M} \upsilon_G(t) ~, \;\;\;
\upsilon_G(t) = \int_{-\pi}^\pi \frac{dq}{2\pi} \, \frac{1 - \cos(\epsilon_q t - q G) + i \sin(\epsilon_q t - q G)}{2 \sin^2(q/2)} ~.
\label{eq:Upsilon_t_largeM}
\end{equation}
Note that $\upsilon_G(t)$ is independent of $M$ (anticipating implicitly that a potentially time-dependent $G$ is also independent of $M$), and we can then think of it as characterizing the behavior of the domain boundary (here two boundaries, one at the left and the other at the right end of the domain), in the regime $N \to \infty$ (thermodynamic limit) first and then $M \to \infty$ (very large initial domain size) second.
For finite $M \gg 1$, we implicitly assume that the physics actions near the two domain boundaries are ``independent" in the considered time regime.
\subsubsection{Overlaps with the initial droplet:
Diffusive boundary melting for $H_{\text{ReHop}}$}
Let us first consider $G = 0$ calculations.
We will discuss the type I and type II cases separately, analyzing the appropriate scaling behavior in each case.
We can also verify these scalings numerically by exact evaluations of expressions in Eq.~\eqref{eq:defUpsilon} for $G = 0$, as shown in Fig.~\ref{fig:scalings}.
We start with $H_{\text{ReHop}}$.
In this case $\epsilon_q^{\text{ReHop}} \approx \frac{1}{2} w q^2$ for small $q$, and we see that important $q_*(t) \sim 1/\sqrt{wt} \ll 1$ for $\sqrt{wt} \gg 1$.
The condition on $M$ justifying the large domain size regime in Eq.~(\ref{eq:Upsilon_t_largeM}) is then $M \gg \sqrt{wt}$. 
Since we are in the largish time regime, we can also extend the range of $q$ integration to the whole real line, and upon simple variable rescaling obtain
\begin{equation}
\upsilon^{\text{ReHop}}_0(t) \approx \sqrt{\frac{wt}{2}} \int_{-\infty}^\infty \frac{d\tilde{q}}{\pi}\, \frac{1 - \cos(\tilde{q}^2) + i \sin(\tilde{q}^2)}{\tilde{q}^2} = \sqrt{\frac{wt}{\pi}} (1 + i) ~.
\label{eq:upsilonReHop}
\end{equation}
Thus, in the scaling regime $1 \ll wt \ll M^2$, the wavefunction overlap reduction from $1$ is proportional to $\frac{1}{M} \sqrt{wt}$, which we can loosely view as some ``diffusion-like'' physics happening with the domain boundaries.
We can also see similar scaling with time in the leakage analysis shown in the inset of Fig.~\ref{fig:Wdropletorbitals}(a), where we define the leakage as the summed up $n_j(t)$ over $j$'s outside of the initial droplet location.
Turning to $H_{\text{ImHop}}$, we have $\epsilon_q^{\text{ImHop}} \approx wq$ for small $q$.
Hence, we see that the important $q_*(t) \sim 1/wt \ll 1$ for $wt \gg 1$.
The large domain condition in Eq.~(\ref{eq:Upsilon_t_largeM}) is then $M \gg wt$.
For such largish times we obtain
\begin{equation}
\upsilon^{\text{ImHop}}_0(t) \approx wt \int_{-\infty}^\infty \frac{d\tilde{q}}{\pi} \, \frac{1 - \cos(\tilde{q})}{\tilde{q}^2} = wt ~,
\label{eq:upsilonImHop}
\end{equation}
where we have used the fact that $\sin(\epsilon_q t)$ is odd under $q \rightarrow -q$, and hence does not contribute.
The scaling regime here is $1 \ll wt \ll M$, and the wavefunction reduction from $1$ is proportional to $\frac{1}{M} wt$.
We can loosely view this as some ``ballistic-like'' physics happening with the domain boundaries as far as the wavefunction overlap is concerned, and this is confirmed through exact numerics in Fig.~\ref{fig:scalings}.
We will get a more precise picture of what is happening with the $W$ droplet in this case by considering non-zero time-dependent $G$ in the next subsection.
We thus see that already at the level of calculating the overlap with the initial state, there is a qualitative difference in the behaviors of the $H_{\text{ReHop}}$ versus $H_{\text{ImHop}}$ models in such $W$ droplet quench setting.
We can also argue that for $H_{\rm CHop}$ of Eq.~(\ref{eq:CHopdefn}), as soon as $\beta \neq 0$, the behavior in the regime of interest [formally, large $t$ upon $\lim_{M\to \infty} (\lim_{N \to \infty} \cdots)$] is controlled by the $H_{\text{ImHop}}$ part, hence the same ``ballistic'' decrease of the overlap.
We can also see this numerically, as demonstrated for $H_{\rm CHop}$ in Fig.~\ref{fig:scalings}.
\subsubsection{Chiral ballistic motion and subdiffusive boundary melting under $H_{\text{ImHop}}$}
Motivated by the observation of the ballistic motion of the $W$ droplet in the type II cases in Fig.~\ref{fig:Wdropletnumerics}, we now consider the overlap of $\ket{\phi(t)}$ with $\ket{T_G\phi_0}$, choosing optimal $t$-dependent $G$.
We consider $H_{\text{ImHop}}$ first.
In this case, since $\epsilon^{\text{ImHop}}_q$ is odd in $q$, $\upsilon_G(t)$ has only the real part.
Expecting that the largish $t$ behavior is dominated by small $q$, we Taylor-expand
\begin{equation}
\epsilon_q^{\text{ImHop}} t - q G \approx (wt - G) q - \frac{1}{6} w t q^3 ~.
\end{equation}
To obtain small $\upsilon_G(t)$ [hence large $\braket{T_G\phi_0}{\phi(t)}$], we want to make the numerator in Eq.~(\ref{eq:Upsilon_t_largeM}) small at small $q$ to suppress the effects of the vanishing denominator, and we can achieve this by choosing $G(t) = wt$.
In this case $\epsilon_q^{\text{ImHop}} t - q G(t) \sim w t q^3$, and we see that important $q_\star(t) \sim 1/(w t)^{1/3}$ for this calculation, assuming $(w t)^{1/3} \gg 1$.
The condition on $M$ justifying the large domain size regime in Eq.~(\ref{eq:Upsilon_t_largeM}) is then $M \gg (w t)^{1/3}$.
For such largish times, we can extend the range of $q$ integration to $\mathbb{R}$, and upon simple variable rescaling obtain
\begin{equation}
\upsilon_{G(t)=wt}^{\text{ImHop}}(t) \approx \left(\frac{wt}{6}\right)^{1/3} \int_{-\infty}^\infty \frac{d\tilde{q}}{\pi} \, \frac{1 - \cos(\tilde{q}^3)}{\tilde{q}^2} = A (wt)^{1/3} ~,
\label{eq:upsilon_Gwt_ImHop}
\end{equation}
with numerical constant $A = 3^{1/6} \Gamma(2/3)/(2^{1/3} \pi) \approx 0.411$.
This scaling form is valid for $1 \ll wt \ll M^3$.
We see that the overlap $\braket{T_{G(t)=wt} \phi_0}{\phi(t)} \approx 1 - A (wt)^{1/3}/M$ is close to unity in this regime, supporting the picture that the $W$ droplet is moving as a whole with velocity $w$.
We can think of this aspect of the droplet dynamics as ballistic, in agreement with the discussion of the overlap with the initial state in Eq.~(\ref{eq:upsilonImHop}).
The decrease of the overlap $\braket{T_{G(t)=wt} \phi_0}{\phi(t)}$ suggests that the droplet ``melts'' near the boundaries, but this process is sub-diffusive.
We can see this also in the leakage analysis shown in the inset of Fig.~\ref{fig:Wdropletorbitals}(b), where we measure the leakage as the summed up $n_j(t)$ over $j$'s outside of the corresponding ``reference droplet'' given by $\ket{T_{G(t)=wt} \phi_0}$.
Note that there is no contradiction between $\upsilon_0^{\text{ImHop}}(t)$ in Eq.~(\ref{eq:upsilonImHop}) being valid only for $wt \ll M$ and $\upsilon_{G(t)=wt}^{\text{ImHop}}(t)$ in Eq.~(\ref{eq:upsilon_Gwt_ImHop}) being valid for much longer times $wt \ll M^3$; these different conditions were needed for the corresponding approximations of the full expression in Eq.~(\ref{eq:Upsilon_tM}) to be valid for the different $G=0$ and $G=wt$ calculations.
In fact, from the latter result and the picture of the ballistic motion of the $W$ droplet, we can infer that at time $t < M/w$, there is a geometric overlap with the initial droplet over the region of size $M - wt$, and hence $\braket{\phi_0}{\phi(t)} \approx 1 - wt/M$, in agreement with the former result for $wt \ll M$.
From this picture, we expect this to be a good approximation for all $wt \lesssim M$, which is confirmed by our numerical calculations in Fig.~\ref{fig:scalings}.
Once $wt > M$, the approximations giving $\upsilon_0^{\text{ImHop}}(t)$ in Eq.~(\ref{eq:upsilonImHop}) cannot not be valid, and we expect $\braket{\phi_0}{\phi(t)} \approx 0$ (we do not aim to provide a detailed description of the decay to zero in time).
Finally, we note that the chiral ballistic motion of the $W$ droplet as a whole cannot continue to arbitrary long time if $M$ is finite.
The reason it persists for such long times up to $wt \sim M^3$ is that the initial $W$ droplet forms a particular ``wavepacket'' around $q=0$ (hence moving with group velocity $d\epsilon/dq|_{q=0}$\footnote{We expect generic non-interacting type II Hamiltonians to have non-zero such velocity.
As an example, a combination of the nearest-neighbor $H_{\rm ImHop}$ and the second-neighbor pure imaginary hopping (which can be shown to be type II in the same equivalence class as $H_{\rm ImHop}$) has $\epsilon_q = w_1 \sin(q) + w_2 \sin(2q)$ and hence $d\epsilon/dq|_{q=0} = w_1 + 2 w_2$.
This vanishes at special $w_2 = -w_1/2$, but in this case the total Hamiltonian can be rewritten in terms of local Hermitian annihilators of the $W$ state, namely using three-site ``loops'' $i (s_j^\dagger s_{j+1} + s_{j+1}^\dagger s_{j+2} + s_{j+2}^\dagger s_j - \text{H.c.})$, and is hence type I.}) that ``melts'' very slowly, because the melting happens only from the droplet boundaries.
We close with a remark on how this might generalize to interacting Hamiltonians.
Note that the non-interacting model $H_{\text{ImHop}}$ is integrable, and the total particle current $J_{\text{tot}}^{\text{ImHop}} = \frac{w}{2} \sum_j (s_j^\dagger s_{j+1} + \text{H.c.})$ is an integral of motion in the PBC chain and cannot be degraded.
Hence, since the initial $W$ droplet is a state with non-zero $J_{\text{tot}}^{\text{ImHop}}$, even after the droplet melts, the ``melted parts'' continue going around the chain in circles.
(For simplicity we are mentioning only one integral of motion, while there are also longer-range extensive local integrals of motion, which can be viewed as specific linear combinations of the conserved occupation numbers of the $k$-space orbitals.)
Such persistence of the directional flow to infinite time is special to the non-interacting case, but will go away for generic even small interactions.
On the other hand, we conjecture that the initial directional motion, until the domain melts, is primarily driven by the Hamiltonian action near the boundaries and not by the model integrability, and hence conjecture that it would be present also for more general type II Hamiltonians.
\subsubsection{Chiral ballistic motion and diffusive boundary melting under $H_{\text{CHop}}$}
We now discuss a more general type II Hamiltonian, namely the $H_{\text{CHop}}$ with the dispersion $\epsilon_q^{\text{CHop}}$ of Eq.~(\ref{eq:epsCHop}).
In this case, Eq.~(\ref{eq:Upsilon_t_largeM}) will have both real and imaginary parts.
To study the motion of the droplet as a whole, we again consider small $q$ behavior
\begin{equation}
\epsilon_q^{\text{CHop}} t - q G \approx (\beta w t - G) q + \frac{1}{2}\alpha w t q^2 - \frac{1}{6} \beta w t q^3 ~.
\end{equation}
We see that as long as $\alpha \neq 0$, the $O(q^3)$ term is subdominant and can be dropped when analyzing the longish time behavior.
To minimize $\upsilon_G(t)$, we again choose $G$ such as to cancel the $O(q)$ term, which is achieved with $G(t) = \beta w t$.
The subsequent analysis then becomes identical to that for the diffusive melting for the $H_{\text{ReHop}}$ that gave Eq.~(\ref{eq:upsilonReHop}), with a simple substitution $w \to \alpha w$.
Hence, we obtain
\begin{align}
\upsilon^{\text{CHop}}_{G(t)=\beta w t}(t) \approx \sqrt{\frac{\alpha wt}{\pi}} (1 + i) ~,
\label{eq:upsilon_G_CHop}
\end{align}
which is valid for $1 \ll \alpha w t \ll M^2$.
Thus, for more general type II hopping, the melting is diffusive rather than subdiffusive found for the pure $H_{\text{ImHop}}$ in Eq.~(\ref{eq:upsilon_Gwt_ImHop}).
However, all discussion following Eq.~(\ref{eq:upsilon_Gwt_ImHop}) still holds with this substitution, i.e., using the diffusive melting.
%

\subsection{Generalization to QMBS models with $W^p$ scars}
\label{app:Wp_droplet_dyn}

In Sec.~\ref{sec:dynamical}, we showed that for QMBS models with the $\ket{W}$, $\ket{\bar{0}}$ states as exact scars, quenching from a $W$-like domain inside the vacuum provides a qualitative dynamical distinction between type I and type II Hamiltonians.
In this Appendix, we conjecture a simple generalization to QMBS models with the higher-particle-number $\ket{W^p}$ in Eq.~(\ref{eq:Whigher}) as QMBS, for $p \leq p_{\max}$ with $p_{\max}$ independent of system size.
We could also in principle allow $p$ to be as large as $N$, in which case these QMBS are just the states of the ferromagnetic tower, a.k.a.\ Dicke states.
Note that we have not fully characterized such QMBS models for $p \geq 2$, but we know some type I and type II Hamiltonian instances that readily connect to the $p = 1$ study in the main text and allow a nice reuse/extension of results.
We first remind that the many-body Hamiltonian $H_{\text{ImHop}}$ remains a type II QMBS Hamiltonian for any $p$, a.k.a.\ the Dzyaloshinskii-Moriya Interaction (DMI) Hamiltonian $H_{\text{DMI}}$~\cite{Mark2020Eta, moudgalya2023exhaustive}, which, in the spin language can be written as
\begin{equation}
    H_{\text{DMI}} = \sum_j{(\vec{S}_j \times \vec{S}_{j+1})\cdot \hat{\alpha}} ~, \quad \alpha \in \{x, y, z\}.
\label{eq:HDMI}
\end{equation}
On the other hand, $H_{\text{ReHop}}$ is not a QMBS Hamiltonian for $p \geq 2$.
However, we can rectify this by a simple modification that connects to the spin-1/2 Heisenberg Hamiltonian:
\begin{equation}
\begin{aligned}
& H_{\text{ReHop}} - \sum_j n_j n_{j+1} 
 = \frac{1}{2} \sum_j \left(n_j + n_{j+1} - 2 n_j n_{j+1} - s_j^\dagger s_{j+1} - s_{j+1}^\dagger s_j \right) \\
&= \frac{1}{2} \sum_j (\ket{10} - \ket{01}) (\bra{10} - \bra{01})_{j,j+1} = \sum_j \left(\frac{1}{4} - \vec{S}_j \cdot \vec{S}_{j+1} \right) = H_{\text{Heis}} ~.
\end{aligned}
\end{equation} 
Clearly, $H_{\text{ReHop}}$ and $H_{\text{Heis}}$ act identically in the single-particle subspace, in particular on the $\ket{W}$ and $\ket{W_q}$ states.
The Heisenberg model is a well-known type I Hamiltonian for the full ferromagnetic tower of states, with the individual terms in the above writing annihilating all these states~\cite{Mark2020Eta, moudgalya2023exhaustive}.
In the context of such spin-1/2 Hamiltonians, we can view the $\ket{\bar{0}}$ state as a particular fully polarized state, and then the $\ket{W}$ and $\ket{W_q}$ states are zero-momentum and general-momentum magnon ``excitations,'' which are exact eigenstates of both $H_{\text{Heis}}$ and $H_{\text{DMI}}$.
In this picture, we can view the dispersions $\epsilon_q^{\text{ReHop}}$ and $\epsilon_q^{\text{ImHop}}$ as the corresponding magnon dispersions under the two Hamiltonians, in particular capturing the distinct and all-important small $q$ behaviors in the two cases.
Furthermore, we can qualitatively view the $W^p$ state as a Bose-Einstein Condensate (BEC) of $p$ such zero-momentum magnons:
In general, the magnons are not free, e.g., two magnons cannot be on the same site; however, starting with a naive ``free-magnon'' BEC wavefunction and simply correcting it by excluding contributions with multiply-occupied sites happens to give the exact $W^p$ state, and the BEC approximation is even quantitatively accurate in the very small density limit, $\rho = p/N \to 0$.
This is similar to what is done in the Holstein-Primakoff transformation of the ferromagnetic Heisenberg model~\cite{HolsteinPrimakoff1940}.
We can then consider initializing the system in a state containing a droplet of the $W^p$ state in the vacuum background,
\begin{equation}
\ket{\psi_0} = \ket{W^p}_{[1\dots M]} \otimes \ket{\bar{0}}_{[M+1 \dots N]} ~,
\end{equation}
and evolving the system under the $H_{\text{Heis}}$ vs $H_{\text{DMI}}$ Hamiltonians as the simplest instances of type I vs type II Hamiltonians, or $\alpha H_{\text{Heis}} + \beta H_{\text{DMI}}$ as a more general type II Hamiltonian.
Most importantly, from the preceding discussion, we can view the $W^p$ state droplet in the low-density limit (small $p/M$) as a BEC of $p$ essentially non-interacting magnon particles into the same ``zero momentum (bosonic) orbital'' on the domain $[1 \dots M]$, which is the same as the single-particle $W$ state on that domain whose dynamics we studied in Sec.~\ref{sec:dynamical}.
Hence, we can reuse the free-particle calculation from the main text and expect
\begin{align}
\braket{T_G \psi_0}{\psi(t)} &\approx \big( \braket{T_G \phi_0}{\phi(t)} \big)^p = \big[1 -  \Upsilon_G(t,M) \big]^p = \left( 1 -  \frac{1}{M} \upsilon_G(t) \right)^p ~.
\label{eq:Wpoverlap}
\end{align}
The overlap calculation would be exact for a BEC droplet of free bosons evolving under dispersion $\epsilon_q$.\footnote{Here we refer to the free-boson result that if $\ket{\psi}$ is a normalized BEC wavefunction for condensing $p$ bosons into an orbital $\ket{\phi}$, and similarly $\ket{\psi'}$ is a BEC for $p$ bosons in $\ket{\phi'}$, then $\braket{\psi}{\psi'} = \left(\braket{\phi}{\phi'}\right)^p$.
We also note that the calculation for the BEC droplet quench is meaningful even when the free-boson Hamiltonian has negative energies and the BEC is not its ground state.}
$\Upsilon_G(t,M)$ is from Eq.~(\ref{eq:Upsilon_tM}), which is already in the thermodynamic limit. 
In the last expression, we are considering the ``largish $t$'' and large $M$ scaling regime given in Eq.~(\ref{eq:Upsilon_t_largeM}), and we expect this expression to be accurate in such regime while keeping $p$ fixed (since the droplet is then in the very low density limit).
We can then obtain the final results for different questions of interest by reusing the corresponding single-particle $\upsilon_G(t)$:
$\upsilon^{\text{ReHop}}_0(t)$ in Eq.~(\ref{eq:upsilonReHop}) and $\upsilon^{\text{ImHop}}_0(t)$ in Eq.~(\ref{eq:upsilonImHop}) to calculate overlaps with the initial state in the $H_{\text{Heis}}$ and $H_{\text{DMI}}$ cases respectively;
and $\upsilon^{\text{ImHop}}_{G(t)=wt}(t)$ in Eq.~(\ref{eq:upsilon_Gwt_ImHop}) and $\upsilon^{\text{CHop}}_{G(t)=\beta wt}(t)$ in Eq.~(\ref{eq:upsilon_G_CHop}) to describe the droplet melting in the $H_{\text{DMI}}$ and more general CHop type II Hamiltonians.
Furthermore, we conjecture that the picture of the $W^p$ state as the BEC is still a qualitatively good approximation at a finite but small density of particles in the initial domain, $0 < \rho = p/M \ll 1$.
In this case, we obtain for the overlap with the initial state (using $\upsilon_{G=0} =: \upsilon_0$)
\begin{equation}
\braket{\psi_0}{\psi(t)} \approx \left(1 - \frac{1}{M} \upsilon_{0}(t) \right)^{\rho M} \approx e^{-\rho \upsilon_{0}(t)} ~.
\end{equation}
Note that the domain size $M$ dropped out of the very last expression, and only dependence on the intensive property $\rho$ of the domain remains.
[The coefficient of $\upsilon_{0}(t)$ in the exponent is not quantitatively accurate except for very low $\rho$, but we conjecture that the qualitative character of the result is more general.]
We can then view this expression as an intrinsic dynamical property of the domain boundaries---here, two essentially independent boundaries, one at the left end and the other at the right end of the domain---between the $\ket{W^p}$ state at particle density $\rho$ and the vacuum.
Using Eqs.~(\ref{eq:upsilonReHop}),~(\ref{eq:upsilonImHop}), we see that the resulting decay behaviors in time are qualitatively different between the two types: 
\begin{equation}
\braket{\psi_0}{\psi(t)} \approx
\begin{cases}
e^{-(C'+iC'') \rho \sqrt{wt}} \qquad \text{for}~\epsilon_q^{\text{ReHop}}~\text{(type I)}, \\
e^{-C \rho wt} \qquad \text{for}~\epsilon_q^{\text{ImHop}}~\text{or}~\epsilon_q^{\text{CHop}}~\text{(type II)} ~.
\end{cases}
\end{equation}

We know that the faster decay of the overlap with the initial state in the type II case reflects the ballistic motion of the domain.
If instead we overlap with the reference $W^p$ droplet shifted by $G(t) = \alpha w t$ (considering only $\epsilon^{CHop}_q$ for more generic behavior), we have
\begin{equation}
\braket{T_{G(t) = \alpha w t} \psi_0}{\psi(t)} \approx \left(1 - \frac{1}{M} \upsilon_{G}(t) \right)^{\rho M} \approx e^{-\rho \upsilon_{G}(t)} \approx e^{-(C' + i C'') \rho \sqrt{\alpha w t}} ~.
\end{equation}

While the calculations used the nearly-free-particle approximate model treatments, we conjecture that the above qualitatively different behaviors will hold also for general extensive-local interacting QMBS Hamiltonians that are type I vs type II.
It would be interesting to test these conjectures using direct numerical simulations in future work.
\section{Proof of Theorem~\ref{thm:Hermitiancut}}
\label{app:cuttingproof}
\begin{proof}
We assume a one-dimensional system with $N$ sites, and assume PBC.
Let us expand $H$ in terms of Pauli strings
\begin{equation}
    H = \sum_{j = 1}^N{P_{[j]}},\;\;\;
    P_{[j]} \defn \sum_{q = 1}^{R}\sum_{\mu \in \mP_q}{P^{\mu}_{[j,j+q-1]}},
\label{eq:Hfull}
\end{equation}
where $\mP_q$ is some subset of Pauli strings of length $q$, and to avoid overcounting we assume the leftmost and rightmost sites on each string are non-trivial Paulis (i.e., not identity); here and below, we also implicitly include amplitudes in $P^{\mu}_{[j,j+q-1]}$.
Hence $P_{[j]}$ is the sum of all non-trivial Pauli strings with their left ends at site $j$ (more explicitly in the case of PBC: for $N \gg R$ the ``left-most'' is defined in the sense of circulation $1 \to 2 \to 3 \dots \to N-1 \to N \to 1$). 
Using this expression, we can define the meaning of $H_{\Lambda_{\ell,r}}$ restricted to a patch $\Lambda_{\ell,r} = [\ell,\dots,r]$ as
\begin{equation}
    H_{\Lambda_{\ell,r}} = \sum_{j = \ell}^{r}{\sum_{q = 1}^{\min(R, r - j +1)}\sum_{\mu \in \mP_q}{P^{\mu}_{[j,j+q-1]}}},
\label{eq:Hlamfull}
\end{equation}
where we have ensured that all Pauli strings in $H$ that are completely within the patch $\Lambda_{\ell,r}$ are included in $H_{\Lambda_{\ell,r}}$, and anything that is not completely within that patch is excluded.\footnote{Note that defining a Hamiltonian restriction to a region depends on the basis choice in the single-site operator space, e.g., a length-2 Pauli string $\sigma_j^z \sigma_{j+1}^z = 4\hat{n}_j \hat{n}_{j+1} - 2\hat{n}_j - 2\hat{n}_{j+1} + 1$ contains length-2 and length-1 strings if the basis $\hat{n}_j \defn (1+\sigma_j^z)/2, \sigma_j^{0,x,y}$ is used instead of $\sigma_j^z,\sigma_j^{0,x,y}$.
This is not an issue as long as we fix single-site Hermitian operator basis and definitions of operator strings and their non-trivial (i.e., non-identity) constituents once and for all arguments.}
We now prove that if $H$ is type I, it satisfies the condition of Eq.~(\ref{eq:hermitianrelation}) with explicitly constructed boundary operator satisfying the stated properties.
By definition of type I, the Hamiltonian $H$ can be written as a sum of strictly local Hermitian terms that individually have $\{\ket{\psi_n}\}$ as eigenstates with eigenvalues $\{\varepsilon_{[j],n}\}$, i.e.,
\begin{equation}
    H = \sum_{j = 1}^N{h_{[j]}},\;\;\;h_{[j]}\ket{\psi_n} = \varepsilon_{[j],n}\ket{\psi_n}~, \quad \forall n.
\label{eq:typeIdefn}
\end{equation}
Note that there can be different such writings with all contributions annihilating $\{\ket{\psi_n}\}$ and not labelled by $j$ in general, though they can always be grouped by their locality into $\{h_{[j]}\}$; any non-uniqueness in the final writing is not important below as long as we keep the same choice throughout.
Assuming w.l.o.g. the individual strictly local terms $h_{[j]}$ are localized in a patch of size $R_{\max} \geq R$ near the corresponding $j$, for any patch $\Lambda_{\ell,r}$ of size greater than $2 R_{\max}$, we can always write
\begin{equation}
    H_{\Lambda_{\ell,r}} = A_{[\ell]} + B_{[r]} + \sum_{j,~\text{supp}(h_{[j]}) \in \Lambda_{\ell,r}}{h_{[j]}},
\label{eq:Hlamexpand}
\end{equation}
where $A_{[\ell]}$ and $B_{[r]}$ are terms with support at most on sites $[\ell, \ell + R_{\max} -1]$ and $[r - R_{\max} + 1, r]$, i.e., on the ``boundaries" of the patch $\Lambda$, and the $h_{[j]}$'s appearing in Eq.~(\ref{eq:Hlamexpand}) are in the ``bulk" of the patch.
This can be done by expanding each $h_{[j]}$ in Pauli strings, and defining the ``internal" $h_{[j]}$'s as the ones whose Pauli strings are all in $\Lambda_{\ell,r}$.
Then $A_{[\ell]} + B_{[r]}$ contains all Pauli strings within $\Lambda_{\ell,r}$ from all $h_{[j]}$'s which have some \textit{but not all} Pauli strings within $\Lambda$; these are then uniquely split between $\ell$ and $r$ (from $h_{[j]}$'s near $\ell$ and $r$ respectively, unique for large enough $|r-\ell|$) to give separate $A_{[\ell]}$ and $B_{[r]}$, and these can be defined each independent of the other edge: $A_{[\ell]}$ contains all strings within $\Lambda_{\ell,r}$ from all $h_{[j]}$ whose some \textit{but not all} strings lie withing $\Lambda_{\ell,r}$ and reside near $\ell$.
Note that this procedure defines the same $H_{\Lambda_{\ell,r}}$ irrespective of the specific $\{h_{[j]}\}$ used to write $H$ in Eq.~(\ref{eq:typeIdefn}):  Indeed, all Pauli strings in the intermediate steps, potentially appearing multiple times, are examined for their spatial content relative to $\Lambda_{\ell,r}$ and are aggregated in exactly the same way to reproduce the unique decomposition of $H$ in the Pauli strings and its restriction to $\Lambda_{\ell,r}$.
The specific $A_{[\ell]}$ and $B_{[r]}$ defined by the above procedure do depend on the choice of $\{ h_{[j]} \}$ used in the writing of $H$, but that is alright since we can fix one choice and use it throughout.
Hence, using the condition of Eq.~(\ref{eq:typeIdefn}) we obtain that 
\begin{equation}
    H_{\Lambda_{\ell,r}}\ket{\psi_n} = (A_{[\ell]} + B_{[r]} + f_{n,\Lambda_{\ell,r}})\ket{\psi_n} \defn (T_{\partial \Lambda}+ f_{n, \Lambda} )\ket{\psi_n}, 
\label{eq:Hlambound}
\end{equation}
where $T_{\partial \Lambda}$ is an operator with support strictly on $\mathcal{O}(1)$ sites on the boundaries of the patch $\Lambda$, and $f_{n,\Lambda_{\ell,r}} =\sum_{j,~\text{supp}(h_{[j]}) \in \Lambda_{\ell,r}}{\varepsilon_{[j],n}}$ is the sum of the eigenvalues of $\ket{\psi_n}$ under the $h_{[j]}$s that appear in Eq.~(\ref{eq:Hlamexpand}).
Note that the $A_{[\ell]}$'s and $B_{[r]}$'s are unique, and by the very construction they are Hermitian and are individually independent of the location of the other boundary for large enough $|r-\ell|$.
We can also prove the converse, i.e., if for sufficiently large patches $\Lambda_{\ell,r}$, we are able to write 
\begin{equation}
    H_{\Lambda_{\ell,r}}\ket{\psi_n} = (A_{[\ell]} + B_{[r]} + f_{n, \Lambda_{\ell, r}})\ket{\psi_n},
\label{eq:Lambdaconsteqn}
\end{equation}
where $A_{[\ell]}$ and $B_{[r]}$ are Hermitian terms with support at most on sites $[\ell, \ell + R_A -1]$ and $[r - R_B+1, r]$ respectively where $R_A$ and $R_B$ are $\mathcal{O}(1)$ and $A_{[\ell]}$ and $B_{[r]}$ are independent of $r$ and $\ell$ respectively for sufficiently large $|r-\ell|$, and $f_{n, \Lambda_{\ell, r}}$ is a constant that can depend on $n$ and $\Lambda_{\ell, r}$, we show that $H$ is type I. 
Since we assume Eq.~(\ref{eq:Lambdaconsteqn}) holds for all sufficiently large patches, for a smaller patch $\Lambda_{\ell+1,r}$ we have
\begin{equation}
    H_{ \Lambda_{\ell+1,r}}\ket{\psi_n} = (A_{[\ell+1]} + B_{[r]} +f_{n, \Lambda_{\ell+1, r}})\ket{\psi_n},  
\label{eq:Lambdabarconsteqn}
\end{equation}
where we have assumed that the right edge operator $B_{[r]}$ is independent of the location of the left edge since the patch $\Lambda_{\ell+1,r}$ is still sufficiently large. 
Combining Eqs.~(\ref{eq:Lambdaconsteqn}) and (\ref{eq:Lambdabarconsteqn}), we obtain
\begin{gather}
    (H_{\Lambda_{\ell,r}} - H_{ \Lambda_{\ell+1,r}})\ket{\psi_n} = (A_{[\ell]} - A_{[\ell+1]} + f_{n, \Lambda_{\ell, r}} - f_{n, \Lambda_{\ell+1, r}})\ket{\psi_n} \nn \\
    \implies (P_{[\ell]} + A_{[\ell+1]} - A_{[\ell]} + f_{n, \Lambda_{\ell+1, r}} - f_{n, \Lambda_{\ell, r}})\ket{\psi_n} = 0,  
\label{eq:Hlamcondition}
\end{gather}
where we have assumed that the size of the patches $\Lambda_{\ell,r}$ and $\Lambda_{\ell+1,r}$ are larger than $R$, the length of the largest Pauli string in $H$, and hence following Eqs.~(\ref{eq:Hfull}) and (\ref{eq:Hlamfull}) we have $H_{\Lambda_{\ell,r}} - H_{ \Lambda_{\ell+1,r}} = P_{[\ell]}$.
Using Eqs~(\ref{eq:Hfull}) and (\ref{eq:Hlamcondition}), we can directly show that $H$ is type I, since
\begin{equation}
    H = \sum_{j = 1}^L{h_{[j]}},\;\;\;h_{[j]} \defn P_{[j]} + A_{[j+1]} - A_{[j]},\;\;\;h_{[j]}\ket{\psi_n} = (f_{n, \Lambda_{j, r}} - f_{n, \Lambda_{j+1, r}})\ket{\psi_n} \defn \varepsilon_{[j], n}\ket{\psi_n}, 
\label{eq:TypeIproof}
\end{equation}
where we have used the assumption that the left edge operator $A_{[\ell]}$ from the defining equation Eq.~(\ref{eq:Lambdabarconsteqn}) is independent of the right edge.
Note that since $\{ A_{[\ell]} \}$ are Hermitian, the obtained strictly local symmetric terms $\{h_{[j]}\}$ are also Hermitian, thus completing the proof for type I.
Note that in the definition of $\varepsilon_{[j],n}$ in Eq.~(\ref{eq:TypeIproof}), in particular the fact that it should not depend on $r$ should be automatically follow from the fact that Eq.~(\ref{eq:Hlamcondition}) for an extensive local Hamiltonian, which should impose some structure on $f_{n, \Lambda_{\ell, r}}$. However, for our purposes, it is not necessary to get into the details of this structure.
The proof for type II part,  i.e.,  showing that a Hamiltonian $H$ is type II if and only if its action can be written as a boundary non-Hermitian action,  proceeds similarly. 
Since by definition of type II, any Hamiltonian $H$ can be written as a sum of strictly local non-Hermitian terms that individually have $\{\ket{\psi_n}\}$ as eigenstates,  we can again write Eq.~(\ref{eq:typeIdefn}) with general non-Hermitian $h_{[j]}$. 
Eqs.~(\ref{eq:Hlamexpand}) and (\ref{eq:Hlambound}) follow similarly,  with the only change being that $A_{[\ell]}$ and $B_{[r]}$ are now not guaranteed to be Hermitian. 
The converse proof also follows similarly, where if $A_{[\ell]}$ and $B_{[r]}$ \textit{cannot be made} Hermitian in Eq.~(\ref{eq:Lambdaconsteqn}), then we can guarantee a rewriting of the Hamiltonian as Eq.~(\ref{eq:TypeIproof}),  which is a sum of strictly local non-Hermitian terms.
\end{proof}

%
\section{Proof of Theorem \ref{thm:transfermat}}
\label{app:transfermatrix}
\begin{proof}
To prove this, we first show that the action of a truncated symmetry operator $U^\theta_\Lambda$ on a sufficiently large contiguous region $\Lambda = [\ell, \cdots, r]$ cannot be written as a product of two strictly local \textit{unitary} operators $S^\theta_{X_\ell}$ and $\widetilde{S}^\theta_{X_r}$ that have supports on boundary subregions $X_\ell = [\ell, \cdots, x_\ell]$ and $X_r = [x_r, \cdots, r]$, if the transfer matrix of the MPS is full-rank.
That is, even though we know from Eq.~\eqref{eq:Uthetafinal} that there are generically \textit{non-unitary} boundary operators that reproduce the action of $U^\theta_\Lambda$, we show that
\begin{equation}
        U^\theta_\Lambda\ket{\psi} \neq  S^\theta_{X_\ell} \widetilde{S}^\theta_{X_r}\ket{\psi}\;\text{if $E$ is full-rank},
\label{eq:Uthetacondition}
\end{equation}
for any unitary $S^\theta_{X_\ell}$ and $\widetilde{S}^\theta_{X_r}$.
This also means that under such conditions, $W^\theta_{X_\ell}$ and $W^\theta_{X_r}$ in Eq.~(\ref{eq:Uthetafinal}) cannot be unitary, and hence $O_{X_\ell}$ and $O_{X_r}$ cannot be Hermitian. 
We first note that the action of the L.H.S. in Eq.~(\ref{eq:Uthetacondition}) on the MPS gives us Eq.~(\ref{eq:Uthetaaction}).
We can perform a Schmidt decomposition of the injective MPS over the region $\Lambda$ and its compliment to show that the following conditions are equivalent
\begin{gather}
    U^\theta_\Lambda \ket{\psi} = S^\theta_{X_\ell} \widetilde{S}^\theta_{X_r}\ket{\psi}\; 
\iff\;U^\theta_\Lambda \ket{\psi^{\alpha_\ell, \alpha_{r+1}}_\Lambda} = S^\theta_{X_\ell} \widetilde{S}^\theta_{X_r}\ket{\psi^{\alpha_\ell, \alpha_{r+1}}_\Lambda},
\label{eq:twocond}
\end{gather}
where $\{\ket{\psi^{\alpha_\ell, \alpha_{r+1}}_{\Lambda}}\}$ are the truncated MPS states on the region $\Lambda$, defined as
\begin{equation}
    \ket{\psi^{\alpha_\ell,\alpha_{r+1}}_{\Lambda}} \defn \sum_{\{s_j\}, \{\alpha_{j \neq \ell, r+1}\}}{A^{s_{\ell}}_{\alpha_\ell,\alpha_{\ell+1}} \cdots A^{s_r}_{\alpha_r, \alpha_{r+1}}}\ket{s_\ell \cdots s_r}.
\label{eq:psitrunc}
\end{equation}
Note that below for brevity we will suppress the auxiliary indices such as $\{\alpha_j\}$ on the states $\ket{\psi^{\alpha_\ell,\alpha_{r+1}}_{\Lambda}}$ and refer to their collection simply as $\ket{\psi_{\Lambda}}$.
Using Eq.~(\ref{eq:Uthetaaction}), we have
\begin{equation}
    U^\theta_\Lambda \ket{\psi_{\Lambda}} = \sum_{\{s_j\}}{V(\theta) A^{s_{\ell}} \cdots A^{s_r} V(\theta)^\dagger}\ket{s_\ell \cdots s_r} ~.
\label{eq:truncstate}
\end{equation}
If the equality in Eq.~(\ref{eq:twocond}) has to hold, we can further argue that as long as $|\Lambda - (X_\ell \cup X_r)|$ is larger than the injectivity length $R_{\rm inj}$, it has to be the case that\footnote{The essential argument is that for fixed matrices $M_1$, $M_2$, $M_3$, $M_4$ if we have the matrix relation $M_1 N M_2 = M_3 N M_4$ , that holds for \textit{all} matrices $N$, then we must have $M_1 = \lambda M_3$ and $M_2 = \lambda^{-1} M_4$ for some number $\lambda$.
In this context, denoting the dimensions of the on-site physical and auxiliary Hilbert spaces as $d$ and $D$, we interpret Eq.~(\ref{eq:twocond}) as such a matrix relation choosing $N$ to be the MPS in the region $\Lambda - (X_\ell \cup X_r)$ (which, depending on the choices of the physical indices, can be any $D \times D$ matrix due to injectivity), and choose $M_1$ and $M_3$ to be the left multipliers on both sides in Eq.~(\ref{eq:twocond}) [both interpreted as $D d^{|X_\ell|} \times D$ matrices] leading to Eq.~(\ref{eq:Xltrunc}), and $M_2$ and $M_4$ to be the right multipliers [both interpreted as $D \times D d^{|X_r|}$ matrices] leading to Eq.~(\ref{eq:Xrtrunc}).}
\begin{align}
    &S^\theta_{X_\ell}\ket{\psi_{X_{\ell}}} = \sum_{\{s_j\}}{V(\theta) A^{s_{\ell}} \cdots A^{s_{x_\ell}}}\ket{s_\ell \cdots s_{x_\ell}},\label{eq:Xltrunc}\\
    &\widetilde{S}^\theta_{X_r}\ket{\psi_{X_{r}}} = \sum_{\{s_j\}}{A^{s_{x_r}} \cdots A^{s_r}} V(\theta)^\dagger\ket{s_{x_r}  \cdots s_{r}}\label{eq:Xrtrunc}
\end{align}
where $\ket{\psi_{X_\ell}}$ is the state $\ket{\psi}$ truncated on the region $X_\ell$, similar to Eq.~(\ref{eq:psitrunc}), and similarly for $\ket{\psi_{X_r}}$.
We then consider the overlaps of the states on L.H.S. and R.H.S. on themselves in Eq.~(\ref{eq:Xltrunc}) by contracting the physical indices in the MPS block and retaining the suppressed auxiliary indices.
If $S^\theta_{X_{\ell}}$ is unitary, we obtain
\begin{equation}
    E^{|X_\ell|} = [V(\theta) \otimes V^\ast(\theta)]E^{|X_\ell|},\qquad E \defn \sum_{s}{A^s \otimes (A^s)^\ast},
\label{eq:transferMPS}
\end{equation}
where the $\otimes$ is defined over the auxiliary Hilbert space, and $E$ is the \textit{transfer matrix} of the MPS tensor $A$. 
If $E$ is full-rank, i.e., if all eigenvalues of $E$ are non-zero, it follows that $V(\theta) \otimes V^\ast(\theta) = \mathds{1}$, which can only occur if $V(\theta) = e^{i\theta\alpha}\mathds{1}$, i.e., if the MPS tensor is locally symmetric under $U^\theta$.
Hence, given any symmetric MPS that is not locally symmetric, if all eigenvalues of its transfer matrix are non-zero, $S^\theta_{X_\ell}$ cannot be unitary.
It follows that $\mathcal{O}_{X_\ell}$ and $\widetilde{\mathcal{O}}_{X_r}$ in Eq.~(\ref{eq:ULambdaboundary}) cannot be Hermitian, which by Thm.~\ref{thm:Hermitiancut} implies that $\sum_j{L_j}$ is a type II parent Hamiltonian of $\ket{\psi}$. 
Note that for the transfer matrix of an injective MPS (which can be shown to have a unique largest eigenvalue $1$), the full rank condition implies that the correlation length (which is defined as $-1/\log |\xi_2|$, where $\xi_2$ is the second largest eigenvalue of $E$) is non-zero .
\end{proof}
\end{document}